\documentclass[sigplan,screen,preprint]{acmart}

\acmJournal{PACMPL}
\acmNumber{CONF} 
\acmDOI{} 

\setcopyright{none}

\usepackage{booktabs}   
\usepackage{subcaption} 

\usepackage{etoolbox}
\newtoggle{full}
\toggletrue{full}

\usepackage[inline]{enumitem} 

\usepackage{listings}
\usepackage{mathpartir}
\usepackage{thm-restate}
\usepackage{xcolor}
\definecolor{Maroon}{rgb}{0.5, 0.0, 0.0}
\definecolor{MediumVioletRed}{rgb}{0.78, 0.08, 0.52}
\definecolor{ltblue}{rgb}{0.23, 0.21, 0.98}

\usepackage{wrapfig}

\usepackage{amsmath}

\usepackage{tikz}
\usetikzlibrary{cd, calc, positioning, decorations.markings, decorations.pathmorphing, shapes, decorations.pathreplacing}

\usepackage{cancel}

\usepackage{xspace}
\usepackage{fancyvrb}

\definecolor{dkblue}{rgb}{0,0.1,0.5}
\definecolor{lightblue}{rgb}{0,0.5,0.5}
\definecolor{dkgreen}{rgb}{0,0.4,0}
\definecolor{dk2green}{rgb}{0.4,0,0}
\definecolor{dkviolet}{rgb}{0.6,0,0.8}
\definecolor{mantra}{rgb}{0.2,0.6,0.2}
\definecolor{gotcha}{rgb}{0.8,0.2,0}
\definecolor{ocre}{RGB}{243,102,25} 
\definecolor{dkolive}{RGB}{85, 107, 47}
\definecolor{pine}{RGB}{1, 121, 111}
\definecolor{DarkSlateBlue}{RGB}{72,61,139}
\definecolor{dkred}{RGB}{139, 0, 0}
\definecolor{coffee}{RGB}{111, 78, 55}

\usepackage[british]{babel}
\usepackage[nomargin,inline,draft]{fixme}
\fxsetup{theme=color,mode=multiuser}
\FXRegisterAuthor{FF}{aFF}{FF} 
\FXRegisterAuthor{NY}{aNY}{NY} 
\FXRegisterAuthor{DC}{aDC}{DC} 
\FXRegisterAuthor{LG}{aLG}{LG} 

\newcommand{\code}[1]{\ensuremath{\texttt{#1}}} 
\newcommand{\rulename}[1]{\DefTirName{#1}\xspace}
\newcommand{\dbj}{de Bruijn\xspace}

\newcommand{\SEP}{\mathbin{\mathbf{|\!\!|}}}
\newcommand{\theDiagram}{\text{the diagram}}

\ifdefined\NOCOLOR
  \newcommand{\withcolor}[2]{#2} 
\else
  \newcommand{\withcolor}[2]{\colorlet{currbkp}{.}\color{#1}{#2}\color{currbkp}}
\fi

\ifdefined\NOCOLOR
\else
  
\fi

\newcommand{\colorproc}{Maroon} 
\newcommand{\colorexp}{MediumVioletRed} 
\newcommand{\colorte}{black}
\newcommand{\colorlp}{blue} 
\newcommand{\colorgt}{violet} 
\newcommand{\colorcog}{pine} 
\newcommand{\colorcol}{dkolive} 
\newcommand{\colorpre}{coffee} 

\newcommand{\colorrules}{DarkSlateBlue} 

\let\DefTirNameOld\DefTirName
\renewcommand{\DefTirName}[1]{\withcolor\colorrules{\DefTirNameOld{\footnotesize[{#1}]}}}

\newcommand{\ofLt}{\vdash_{\text{\textsc{lt}}}}
\newcommand{\dslName}{\textsf{Zooid}\xspace}

\newcommand{\dproc}[1]{\withcolor{\colorproc}{#1}} 
\newcommand{\dcoq}[1]{\withcolor{\colorexp}{#1}} 

\newcommand{\proc}{\dproc{\code{proc}}}
\newcommand{\Proc}{\code{Proc}}
\newcommand{\pend}{\dproc{\code{finish}}}
\newcommand{\pjump}[1]{\dproc{\code{jump} \; #1}}
\newcommand{\ploop}[2]{\dproc{\code{loop} \; #1 \; \{ #2 \}}}
\newcommand{\precv}[2]{\dproc{\code{recv} \; #1 \; \{ #2 \}}}
\newcommand{\psend}[3]{\dproc{\code{send} \; #1 \; (#2).\; #3}}

\newcommand{\pread}[2]{\dproc{\code{read} \; {#1} \; (x.{#2})}}
\newcommand{\pwrite}[3]{\dproc{\code{write} \; {#1} \; {#2} \; {#3}}}
\newcommand{\pinteract}[3]{\dproc{\code{interact} \; {#1} \; {#2} \; (x.{#3})}}

\newcommand{\pacti}{\ensuremath{\withcolor{blue}{\code{act}_i}}\xspace}
\newcommand{\pactr}{\ensuremath{\withcolor{blue}{\code{act}_r}}\xspace}
\newcommand{\pactw}{\ensuremath{\withcolor{blue}{\code{act}_w}}\xspace}

\newcommand{\zooid}{\dproc{\code{Z}}}
\newcommand{\zend}{\dproc{\code{finish}}}
\newcommand{\zjump}[1]{\dproc{\code{jump}} \; #1}
\newcommand{\zloop}[2]{\dproc{\code{loop} \; #1 \; ( #2 )}}
\newcommand{\zif}[1]{\dproc{\code{if} \; #1 }}
\newcommand{\zthen}[1]{\dproc{\code{then} \; #1 }}
\newcommand{\zelse}[1]{\dproc{\code{else} \; #1 }}

\newcommand{\zalts}[1]{\dproc{\begin{array}[t]{@{}l@{}} [#1] \end{array}}}
\newcommand{\zsend}[3]{\dproc{\code{send} \; #1 \;  (#2, \; #3) \code{!} \;}}
\newcommand{\zrecv}[3]{\dproc{\code{recv} \; #1 \;  (#2, \; #3) \code{?} \;}}
\newcommand{\zbranch}[2]{\dproc{\code{branch} \; #1 \;  #2}}
\newcommand{\zselect}[2]{\dproc{\code{select} \; #1 \;  #2}}

\newcommand{\zfun}[1]{\code{fun}\ #1\ \Rightarrow\ }

\newcommand{\zinteract}[2]{\dproc{\code{interact} \; #1 \ \  #2}}

\newcommand{\zbalt}[2]{\dproc{\; #1, \;  #2 \code{?} \; }}
\newcommand{\zcase}[3]{\dproc{\code{case} \; #1 \Rightarrow #2, \; #3 \code{!} \;}}
\newcommand{\zskip}[2]{\dproc{\code{skip} \Rightarrow #1, \; #2 \code{!} \;}}
\newcommand{\zdflt}[2]{\dproc{\code{otherwise} \Rightarrow #1, \; #2 \code{!} \;}}

\newcommand{\zooidb}{\dproc{\zooid^b}}
\newcommand{\zooids}{\dproc{\zooid^s}}

\newcommand{\eif}[1]{\dcoq{\code{if} \; #1}}
\newcommand{\ethen}[1]{\dcoq{\code{then} \; #1}}
\newcommand{\eelse}[1]{\dcoq{\code{else} \; #1}}
\newcommand{\tfun}[2]{\dcoq{\code{fun} \; #1 \Rightarrow #2}}
\newcommand{\expr}{\dcoq{e}}
\newcommand{\vx}{\dcoq{x}}

\newcommand{\subtrace}{\preceq}
\newcommand{\zooidTheoremCoq}{\code{process\_traces\_are\_global\_types}}

\newcommand{\Rpipe}{\code{pipeline}}
\newcommand{\Rpingpong}{\code{ping\_pong}}
\newcommand{\RpipeLT}{\code{pipeline$_{lt}$}}
\newcommand{\AliceLT}{\code{alice$_{lt}$}}
\newcommand{\BobLT}{\code{bob$_{lt}$}}
\newcommand{\BobProc}{\code{bob}}
\newcommand{\AliceProc}{\code{alice}}
\newcommand{\coqProj}{\code{\textbackslash}\code{project}}
\newcommand{\coqGet}{\code{\textbackslash}\code{get}}
\newcommand{\coqkwstyle}{\ttfamily\color{dkviolet}}
\newcommand{\coqDef}{\code{\coqkwstyle{Definition}}}
\newcommand{\coqEval}{\code{\coqkwstyle{Eval}}}
\newcommand{\coqOpaque}{\code{\coqkwstyle{Opaque}}}
\newcommand{\zooidTy}[1]{\code{wt\_proc} \; #1}
\newcommand{\azooid}{\code{typed\_proc}}

\newcommand{\runProc}{\code{extract\_proc}}

\newcommand{\TwoBuyer}{\code{two\_buyer}}
\newcommand{\ProcB}{\code{buyer}_{\code{B}}}
\newcommand{\BuyerA}{\code{A}}
\newcommand{\BuyerB}{\code{B}}
\newcommand{\BuyerBLT}{\code{B}_{\code{lt}}}
\newcommand{\Seller}{\code{S}}
\newcommand{\ItemId}{\code{ItemId}}
\newcommand{\Quote}{\code{Quote}}
\newcommand{\Propose}{\code{Propose}}
\newcommand{\Accept}{\code{Accept}}
\newcommand{\Reject}{\code{Reject}}
\newcommand{\Date}{\code{Date}}

\newcommand{\dte}[1]{\withcolor{\colorte}{#1}} 

\newcommand{\tbool}{\dte{\code{bool}}}

\newcommand{\tnat}{\dte{\code{nat}}}
\newcommand{\tunit}{\dte{\code{unit}}}
\newcommand{\tint}{\dte{\code{int}}}
\newcommand{\tseq}{\dte{\code{seq}}}
\newcommand{\tS}{\dte{\code{S}}}
\newcommand{\tplus}[2]{\dte{#1\code{+}#2}}
\newcommand{\tpair}[2]{\dte{#1\code{*}#2}}

\newcommand{\role}{\code{role}}
\newcommand{\rel}[2]{\code{rel}\ \ {#1}\ \ {#2}}
\newcommand{\lty}{\code{l\_ty}}
\newcommand{\gty}{\code{g\_ty}}
\newcommand{\rlty}{\code{rl\_ty}}
\newcommand{\rgty}{\code{rg\_ty}}
\newcommand{\igty}{\code{ig\_ty}}
\newcommand{\colty}{\code{l\_ty}^{\code c}}
\newcommand{\cogty}{\code{g\_ty}^{\code c}}
\newcommand{\mty}{\code{mty}}
\newcommand{\lbl}{\code{label}}
\newcommand{\seq}{\code{seq}}
\newcommand{\fmap}{\code{fmap}}

\newcommand{\guarded}{\code{guarded}}
\newcommand{\gfv}{\code{fv}}
\newcommand{\lfv}{\code{fv}}
\newcommand{\closed}{\code{closed}}
\newcommand{\env}{\code{renv}}

\newcommand{\subject}[1]{\ensuremath{\code{subj}\ {#1}}}

\newcommand{\newDfrac}[2]{%
  \ooalign{%
    $\dfrac{#1}{\phantom{#2}}$\cr
    \raisebox{-2.0pt}{$\dfrac{\phantom{#1}}{#2}$}}%
}

\newcommand{\p}{{\sf p}}
\newcommand{\q}{{\sf q}}
\newcommand{\pr}{{\sf r}}
\newcommand{\parti}{\code{prts}}
\newcommand{\partof}[2]{\code{part\_of}\ {#1} \ {#2}}
\newcommand{\pof}{\code{part\_of}}

\newcommand{\Alice}{{{\sf Alice}}}
\newcommand{\Bob}{{{\sf Bob}}}
\newcommand{\Carol}{{{\sf Carol}}}

\newcommand{\dlt}[1]{\withcolor{\colorlp}{#1}} 

\newcommand{\lT}{\dlt{ \ensuremath{{\sf L}}}}
\newcommand{\lsend}[4]{\dlt{![{#1}];\{#2_i({#3}_i). #4_i \}_{i \in I}}}
\newcommand{\lsendni}[4]{\dlt{![{#1}];\{#2_i({#3}_i). #4 \}_{i \in I}}}
\newcommand{\lrecv}[4]{\dlt{?[{#1}];\{#2_i({#3}_i). #4_i \}_{i \in I}}}
\newcommand{\lrecvni}[4]{\dlt{?[{#1}];\{#2_i({#3}_i). #4 \}_{i \in I}}}
\newcommand{\lrcv}[1]{\dlt{?[{#1}];}}
\newcommand{\lsnd}[1]{\dlt{![{#1}];}}
\newcommand{\lrec}[2]{\dlt{\mu #1 . #2}}
\newcommand{\lX}{\dlt{X}}
\newcommand{\lend}{\dlt{\mathtt{end}}}

\newcommand{\lunroll}[2]{#1\;\Re\;#2}

\newcommand{\dcol}[1]{\withcolor{\colorcol}{#1}}

\newcommand{\colT}{\dcol{ \ensuremath{{\sf L^c}}}}
\newcommand{\colrecv}[4]{\dcol{?^{\sf c}[{#1}];\{#2_i({#3}_i). #4_i \}_{i \in I}}}
\newcommand{\colsend}[4]{\dcol{!^{\sf c}[{#1}];\{#2_i({#3}_i). #4_i \}_{i \in I}}}
\newcommand{\colrcv}[1]{\dcol{?^{\sf c}[{#1}];}}
\newcommand{\colsnd}[1]{\dcol{!^{\sf c}[{#1}];}}
\newcommand{\colend}{\dcol{\mathtt{end}^{\sf c}}}

\newcommand{\dgt}[1]{\withcolor{\colorgt}{#1}} 
\ifdefined\G
\renewcommand{\G}{\dgt{ \ensuremath{{\sf G}}}}
\else
\newcommand{\G}{\dgt{ \ensuremath{{\sf G}}}}
\fi
\newcommand{\gmu}{\dgt{\mu}}
\newcommand{\msg}[2]{\dgt{ \ensuremath{#1\to#2:}}}
\newcommand{\msgi}[5]{\dgt{ \ensuremath{#1\to#2:\{#3_i({#4}_i). #5_i \}_{i \in I}}}}

\newcommand{\gX}{\dgt{X}}
\newcommand{\grec}[2]{\dgt{\mu #1 . #2}}
\newcommand{\gend}{\dgt{\mathtt{end}}}

\newcommand{\dcog}[1]{\withcolor{\colorcog}{#1}}

\newcounter{sarrow}
\newcommand\xrsa[1]{%
\stepcounter{sarrow}%
\mathrel{\begin{tikzpicture}[baseline= {( $ (current bounding box.south) + (0,0
) $ )}]
\node[inner sep=.5ex] (\thesarrow) {$\scriptstyle #1$};
\path[draw,<-,decorate,
  decoration={zigzag,amplitude=0.7pt,segment length=1.2mm,pre=lineto,pre length=4pt}]
    (\thesarrow.south east) -- (\thesarrow.south west);
\end{tikzpicture}}%
}

\newcommand{\coG}{\dcog{ \ensuremath{{\sf G}^{\sf c}}}}

\newcommand{\cogend}{\dcog{\mathtt{end}^{\sf c}}}
\newcommand{\comsgni}[5]{\dcog{ \ensuremath{#1\to#2:\{#3_i({#4}_i). #5_i \}_{i \in I}}}}

\newcommand{\comsgsi}[6]{\dcog{ \ensuremath{#1\xrsa{#3} #2:\{#4_i({#5}_i). #6_i \}_{i \in I}}}}

\newcommand{\comsgn}[2]{\dcog{ \ensuremath{#1\to#2:}}}
\newcommand{\comsgs}[3]{\dcog{ \ensuremath{#1\xrsa{#3} #2:}}}

\newcommand{\gunroll}[2]{#1\;\Re\;#2}
\newcommand{\step}{\xrightarrow{\text{step}}}
\newcommand{\stepa}[1]{\xrightarrow{#1}}

\newcommand{\preG}{\dpre{ \ensuremath{{\sf G}^{\sf p}}}}

\newcommand{\dpre}[1]{\withcolor{\colorpre}{#1}}

\newcommand{\pgend}{\dpre{\mathtt{inj}^{\sf p}}}
\newcommand{\pmsgni}[5]{\dpre{ \ensuremath{#1\to#2:\{#3_i({#4}_i). #5_i \}_{i \in I}}}}

\newcommand{\pmsgsi}[6]{\dpre{ \ensuremath{#1\xrsa{#3} #2:\{#4_i({#5}_i). #6_i \}_{i \in I}}}}

\newcommand{\pmsgs}[3]{\dpre{ \ensuremath{#1\xrsa{#3} #2:}}}

\newcommand{\projt}[2]{{#1}{\upharpoonright}{#2}}
\newcommand{\coproj}[3]{ #2\ {\upharpoonright^{\code c}} {#1}\ #3}
\newcommand{\osproj}{\upharpoonright\!\upharpoonright}

\newcommand{\qenv}{\code{qenv}}
\newcommand{\enq}{\code{enq}}
\newcommand{\deq}{\code{deq}}
\newcommand{\qproj}[2]{{#1}{\upharpoonright^{\code q}}{#2}}

\newcommand{\trend}{[]}
\newcommand{\trnext}[2]{{#1}\#{#2}}
\newcommand{\glts}[2]{\code{tr}^{\code{g}}\;{#1}\;{#2}}
\newcommand{\llts}[2]{\code{tr}^{\code{l}}\;{#1}\;{#2}}
\newcommand{\plts}[2]{\code{tr}^{\code{p}}\;{#1}\;{#2}}

\ifdefined\NOCOLOR
  \newcommand{\chcolor}[1]{black} 
\else
  \newcommand{\chcolor}[1]{#1} 
\fi

\newcommand{\myparagraph}[1]{\paragraph{\textbf{#1}}}
\renewcommand{\sec}{\S\xspace}

\newcommand{\ocamlcommentstyle}{\color{blue}}

\lstdefinelanguage{ocaml}[Objective]{Caml}{
  deletekeywords={ref},
  morekeywords={module}
  flexiblecolumns=false,
  showstringspaces=false,
  framesep=5pt,
  commentstyle=\ocamlcommentstyle,
  basicstyle=\tt\footnotesize,
  numberstyle=\scriptsize,
  escapeinside={$}{$},
}

\newcommand{\appref}[1]{%
  \ifthenelse{\isundefined{\appendixIncluded}}{\ref{#1}}{\ref{APP-#1}}}

\iftoggle{full}{
\newcommand{\revised}[1]{#1}
}
         {
\newcommand{\revised}[1]{{\color{red}#1}}
}

\usepackage{listings}
\lstdefinelanguage{Coq}{
mathescape=true,
texcl=false,
morekeywords=[1]{Section, Module, End, Require, Import, Export,
  Variable, Variables, Parameter, Parameters, Axiom, Hypothesis,
  Hypotheses, Notation, Local, Tactic, Reserved, Scope, Open, Close,
  Bind, Delimit, Definition, Let, Ltac, Fixpoint, CoFixpoint, Add,
  Morphism, Relation, Implicit, Arguments, Unset, Contextual,
  Strict, Prenex, Implicits, Inductive, CoInductive, Record,
  Structure, Canonical, Coercion, Context, Class, Global, Instance,
  Program, Infix, Theorem, Lemma, Corollary, Proposition, Fact,
  Remark, Example, Proof, Goal, Save, Qed, Defined, Hint, Resolve,
  Rewrite, View, Search, Show, Print, Printing, All, Eval, Check,
  Projections, inside, outside, Def},
morekeywords=[2]{forall, exists, exists2, fun, fix, cofix, struct,
  match, with, end, as, in, return, let, if, is, then, else, for, of,
  nosimpl, when},
morekeywords=[3]{Type, Prop, Set, true, false, option},
morekeywords=[4]{pose, set, move, case, elim, apply, clear, hnf,
  intro, intros, generalize, rename, pattern, after, destruct,
  induction, using, refine, inversion, injection, rewrite, congr,
  unlock, compute, ring, field, fourier, replace, fold, unfold,
  change, cutrewrite, simpl, have, suff, wlog, suffices, without,
  loss, nat_norm, assert, cut, trivial, revert, bool_congr, nat_congr,
  symmetry, transitivity, auto, split, left, right, autorewrite},
morekeywords=[5]{by, done, exact, reflexivity, tauto, romega, omega,
  assumption, solve, contradiction, discriminate},
morekeywords=[6]{do, last, first, try, idtac, repeat},
morecomment=[s]{(*}{*)},
showstringspaces=false,
morestring=[b]",
morestring=[d],
tabsize=3,
extendedchars=false,
sensitive=true,
breaklines=false,
basicstyle=\small,
captionpos=b,
columns=[l]flexible,
identifierstyle={\ttfamily\color{black}},
keywordstyle=[1]{\ttfamily\color{dkviolet}},
keywordstyle=[2]{\ttfamily\color{dkgreen}},
keywordstyle=[3]{\ttfamily\color{ltblue}},
keywordstyle=[4]{\ttfamily\color{dkblue}},
keywordstyle=[5]{\ttfamily\color{dkred}},
stringstyle=\ttfamily,
commentstyle={\ttfamily\color{dkgreen}},
literate=
    {\\forall}{{\color{dkgreen}{$\forall\;$}}}1
    {\\exists}{{$\exists\;$}}1
    {<-}{{$\leftarrow\;$}}1
    {=>}{{$\Rightarrow\;$}}1
    {==}{{\code{==}\;}}1
    {==>}{{\code{==>}\;}}1
    {->}{{$\rightarrow\;$}}1
    {<->}{{$\leftrightarrow\;$}}1
    {<==}{{$\leq\;$}}1
    {\#}{{$^\star$}}1
    {\\o}{{$\circ\;$}}1
    {\@}{{$\cdot$}}1
    {\/\\}{{$\wedge\;$}}1
    {\\\/}{{$\vee\;$}}1
    {++}{{\code{++}}}1
    {~}{{\ }}1
    {\@\@}{{$@$}}1
    {\\mapsto}{{$\mapsto\;$}}1
    {\\hline}{{\rule{\linewidth}{0.5pt}}}1
}[keywords,comments,strings]

\lstnewenvironment{coq}{\lstset{language=Coq}}{}

\begin{document}

\theoremstyle{acmdefinition}
\newtheorem{remark}[theorem]{Remark}

\title[Certified Multiparty Computation]
  {\dslName: a DSL for Certified Multiparty Computation}         
\subtitle{From Mechanised Metatheory to Certified Multiparty Processes\\(long version)}


\author{David Castro-Perez}
\affiliation{
  \department{Department of Computer Science}              
  \institution{Imperial College London, UK}            
}

 \affiliation{
   \department{School of Computing}             
   \institution{University of Kent, UK}           
 }
\email{d.castro-perez@kent.ac.uk}

\author{Francisco Ferreira}
\affiliation{
  \department{Department of Computer Science}              
  \institution{Imperial College London, UK}            
}
\email{f.ferreira@imperial.ac.uk}          

\author{Lorenzo Gheri}
\affiliation{
  \department{Department of Computer Science}              
  \institution{Imperial College London, UK}            
}
\email{l.gheri@imperial.ac.uk}          

\author{Nobuko Yoshida}
\affiliation{
  \department{Department of Computer Science}              
  \institution{Imperial College London, UK}            
}
\email{n.yoshida@imperial.ac.uk}          

\begin{abstract}

We design and implement $\textsf{Zooid}$, a domain specific language for
certified multiparty communication, embedded in Coq and
implemented atop our mechanisation framework of asynchronous
multiparty session types (the first of its kind). $\textsf{Zooid}$ provides a fully
mechanised metatheory for the semantics of global and local types, and
a fully verified end-point process language that faithfully reflects
the type-level behaviours and thus inherits the global types
properties such as deadlock freedom, protocol compliance, and liveness
guarantees.


\end{abstract}

\ccsdesc[500]{Software and its engineering~General programming languages}
\ccsdesc[300]{Social and professional topics~History of programming languages}

\keywords{multiparty session types, mechanisation, Coq, asynchronous message passing,
  concurrent processes}  

\maketitle

\section{Introduction}\label{sec:intro}
\emph{Concurrent behavioural type
  systems}~\cite{Huttel:2016:FST:2911992.2873052} accurately simulate and abstract
the behaviour of interactive \emph{processes}, as opposed to sequential types for
programs that simply describe values.  The \emph{session types
system}~\cite{Honda:1993,Honda:1998,Takeuchi:1994} is one of such behavioural
type systems, which can determine protocol
compliance for processes. Session types consist of 
actions for sending and receiving, 
sequencing, choices, and
recursion. In session types, when a typed process communicates, its
type also evolves, thus reflecting the progression
of the state of the protocol (type) after
performing an action. This rich behavioural aspect of session types
has opened new areas of study, such as a connection with communicating
automata~\cite{cfsm83} and concurrent game semantics~\cite{lics11} by
linking actions of session types to transitions of state machines
\cite{DenielouY12} and events of games~\cite{POPL19}.

Originally, binary session types (BST) provide deadlock-freedom for a pair of
processes, but not when more than two \emph{participants} (often also called
\emph{roles}) are involved. \revised{For more than two processes,
  ensuring deadlock-freedom in BST requires either complicated
  additional causality-based typing systems on top of plain BST,
  e.g. \cite{dezani-deliguoro-yoshida,balzer-toninho-pfenning} or
  limitation to deterministic, strongly-normalising session types
\cite{TY2018,TY2018b}.}

\emph{Multiparty session types}
(MPST,~\cite{Honda2008Multiparty,HYC2016}) solve this limitation, by
defining \emph{global types} as an overall specification of all the
communications by every participant involved. The essence of the MPST
theory (depicted in Figure~\ref{fig:mpst}) is \emph{end-point projection} where a global type $\G$ is projected
into one \emph{local type} $\dlt{\lT_i}$ for each participant, so that the
participant $\dproc{\proc_i}$ can be implemented following an abstract behaviour
represented by the local type.
To ensure correctness, the collection of behaviours of the local types projected
from a global type need to mirror the behaviour of that global type.

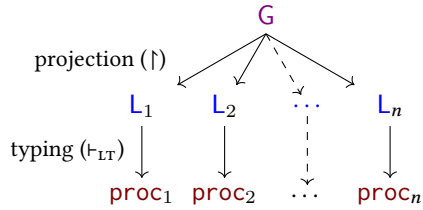
\begin{figure}
  \centering
  \vspace{3mm}
  \begin{tikzpicture}
    [x=1.1cm, y=-1.5cm, align=center, font=\normalsize, draw,
    every node/.style={minimum width=1.8cm}]
		\node [] (G) at (0,0) {$\G$};
		\node [] (L1) at (-1.5,.8 ) {$\dlt{\lT_1}$};
		\node [] (L2) at (-.50,.8 ) {$\dlt{\lT_2}$};
		\node [] (LL) at (.50 ,.8 ) {$\dlt{\dots}$};
		\node [] (Ln) at ( 1.5,.8 ) {$\dlt{\lT_n}$};
		\node [] (P1) at (-1.5,1.6) {$\dproc{\proc_1}$};
		\node [] (P2) at (-.50,1.6) {$\dproc{\proc_2}$};
		\node [] (PP) at ( .50,1.6) {$\dots$};
		\node [] (Pn) at ( 1.5,1.6) {$\dproc{\proc_n}$};
		\draw [->] (G.south) to (L1);
		\draw [->] (G.south) to (L2);
		\draw [->, dashed] (G.south) to (LL);
		\draw [->] (G.south) to (Ln);
		\draw [->] (L1) to (P1);
		\draw [->] (L2) to (P2);
		\draw [->, dashed] (LL) to (PP);
		\draw [->] (Ln) to (Pn);
	    \node  at (-2,0.4) {\small projection ($\upharpoonright$)};
	    \node  at (-2.4,1.2) {\small  typing ($\ofLt$)};
	\end{tikzpicture}
	\caption{MPST in a nutshell}
\vspace{-6mm}
\label{fig:mpst}
\end{figure}

The behaviour of global and local types is defined by
(asynchronous) labelled transition systems (LTS) whose sound and
complete correspondence is key to provide: progress of
processes~\cite{HYC2016}, synthesis of global
protocols~\cite{DenielouYoshida2013,LTY2015}, and to establish bisimulation for
processes~\cite{KY2015}. Practically, type-level transition systems
are particularly useful for, e.g., dynamic monitoring of components
in distributed systems~\cite{DHHNY2015} and generating deadlock-free APIs
of various programming
languages, e.g.,
~\cite{Castro:2019:DPU:3302515.3290342,SDHY2017,HY2017,NHYA2018,ZFHNY2020}.

Unfortunately, the more complicated the behaviour is,
the more \emph{error-prone} the theory becomes.
The literature reveals broken proofs of subject reduction for several
MPST systems~\cite{SY2019}, and a flaw of the decidability of
subtyping~\cite{BravettiCZ16} for asynchronous MPST. All of which are
caused by an incorrect understanding of the (asynchronous) behaviour of
types.

Motivated by this experience, we design and implement
\dslName\footnote{A \href{https://en.wikipedia.org/wiki/Zooid}{zooid}
  is a single animal that is part of a colonial animal, akin to how an
  endpoint process is part of a distributed system.}, a certified
Domain Specific Language (DSL) to write \emph{well-typed by
  construction} communicating processes. \dslName's
implementation is embedded in the Coq proof
assistant~\cite{CoqManual}, so that it relies on solid and precise
foundations: in Coq we have formalised the metatheory for MPST, which
serves as the type system for \dslName{}. On one side, mechanising the
metatheory is immediately useful for documenting, clarifying, and
ensuring the validity of proofs, on the other it results in certified
specifications and implementations of the concepts in the theory.
\dslName{} exemplifies this for MPST, a complex and relevant theory
with many real-world applications. In this system, not only the theory
is validated in Coq, the actual implementation of projection, type
checking and validation of processes, is extracted from certified
proofs.

We provide the first \emph{fully mechanised} proof of sound and
complete correspondence between the labelled transition systems of
global and local types, in terms of equivalence of execution traces,
recapturing the original LTS provided in~\citep{DenielouYoshida2013}.
In this work, instead of trying to formalise existing proofs in the
literature, we approach the problem with a fresh look and use
tools that would allow for a successful and reusable mechanisation. On
the theory side, we use coinductive trees inspired
by~\cite{Li-yao:2019, GHILEZAN2019127}; on the tool side, we
depend on the Coq proof assistant~\cite{CoqManual}, taking advantage
of small scale reflection (SSReflect)~\cite{Gonthier:2010} to
structure our proofs, and PaCo~\cite{paco} to provide a powerful
parameterised coinduction library, which we use extensively.

To certify an MPST end-point process implementation,
we define a concurrent process language and an
LTS semantics for it. This guarantees that process traces
respect the ones from its local and global types.
Naturally, processes do not need to implement every aspect of the
protocol. Therefore, we define the notion of complete subtraces to
represent the fact that an implementation may choose not to implement
some aspects, but it still needs to match the global trace (we make precise
this concept in \sec~\ref{sec:zooidsem}). Our final result
is the design and implementation of \dslName{}, a Coq-embedded
DSL to write end-point processes that are
well-typed (hence deadlock-free and live) \emph{by construction}.
This development takes full
advantage of the metatheory to provide a certified validation,
projection, and type checking for \dslName processes.

The contributions of this work are fourfold:
\begin{description}[leftmargin=2mm]
\item[Fully mechanised transition systems] for global and local types,
  using \emph{asynchronous communications} and proofs of their sound and
  complete trace equivalence.
\item[Semantic representation] of behavioural
types based on \emph{coinductive trees},
proposing a novel approach to the proof of trace
equivalences.
\item[A concurrent process language] with an associated typing discipline and the
  notion of \emph{complete subtraces} to relate process traces to global
  traces, as processes may not fully implement a protocol and still be
  compliant.
\item[\dslName] a DSL embedded in Coq and framework that
  \emph{specifies} global protocols, performs \emph{projections}, and
  implements intrinsically well-typed processes, using code \emph{certified}
  by Coq proofs. The code of \dslName{} processes is \emph{extracted into
    OCaml} code for execution. \dslName{} uses the mechanisation to
  provide a framework for processes that enjoy
 deadlock freedom and liveness (with a type checker certified in Coq).
\end{description}
\textbf{\emph{Outline.}\ } In \sec~\ref{sec:overview}, we provide an overview
of the theory and the paper.  In
\sec~\ref{sec:MPST}, we present the theory of MPST together with the
soundness and completeness results. 
We describe the process language, its metatheory and the \dslName DSL
in~\sec~\ref{sec:proc}. In \sec~\ref{sec:zooid}, we present
\dslName{}'s workflow and  showcase its use with some examples.
In \sec~\ref{sec:related} we discuss related work and offer
some future work and conclusions.

\revised{The git repository of our development is publicly available: 
  \url{https://github.com/emtst/zooid-cmpst}; it contains all the
complete Coq definitions and proofs from the paper, together with the
examples and case studies implemented using \dslName.
\iftoggle{full}{In the Appendix, we present the proofs of our theorems (\sec\ref{appendix:meta-proofs}), and additional technical details of the toolchain (\sec\ref{appendix:extraction})
}{
We present the proofs of our theorems, and additional technical details of the
toolchain, in the appendix of the full version of the paper
(\url{https://arxiv.org/abs/2009.06541})}.
}


\vspace{-1mm}
\section{Overview}\label{sec:overview}
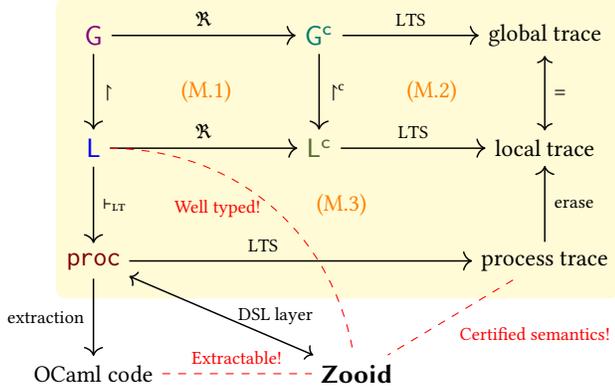
\begin{figure}[t] 

\pgfdeclarelayer{background}
\pgfsetlayers{background,main}

\begin{tikzpicture}[commutative diagrams/every diagram]

   \begin{pgfonlayer}{background}
   \fill[rounded corners, yellow!20!white] (-0.5,0.5) rectangle (7,-3.5);
   \end{pgfonlayer}

	\node(G0)at (0,0) {$\G$};
	\node(G1)at (3,0) {$\coG$};
	\node(G2)at (6,0) {\small{global trace}};
	\node(L0)at (0,-1.5) {$\lT$};
	\node(L1)at (3,-1.5) {$\colT$};
	\node(L2)at (6,-1.5) {\small{local trace}};
	\node(M1)at (1.5,-0.75) {\small{\textcolor{orange}{(M.1)}}};
	\node(M2)at (4.5,-0.75) {\small{\textcolor{orange}{(M.2)}}};
	\node(P0)at (0,-3) {$\proc$};
	\node(P2)at (6,-3) {\small{process trace}};
	\node(M3)at (3.3,-2.25) {\small{\textcolor{orange}{(M.3)}}};
	\node(OC)at (0,-4.5) {\small{OCaml code}};
	\node(ZZ)at (3.5,-4.5) {\bf \dslName};

	\path[commutative diagrams/.cd,every arrow,font=\scriptsize]
	(G0) edge node[above] {$\Re$} (G1)
	(G1) edge node[above] {LTS} (G2)
	(L0) edge node[above] {$\Re$} (L1)
	(L1) edge node[above] {LTS} (L2)
	(G0) edge node[right] {$\upharpoonright$} (L0)
	(G1) edge node[right] {$\upharpoonright^\textsf{c}$} (L1)
	(G2) edge[<->] node[right] {$=$} (L2)
	(L0) edge node[right] {$\ofLt$} (P0)
	(P0) edge node[above, xshift=-.5cm] {LTS} (P2)
	(P2) edge node[right] {erase} (L2)
	(P0) edge node[left] {extraction} (OC)
	(P0) edge[<->] node[right,xshift=.1cm] {DSL layer} (ZZ)
	;
	\path[dashed, red, commutative diagrams/.cd, font=\scriptsize]
	(L0) edge[bend left=40] node[below left,yshift=.2cm] {Well typed!} (ZZ)
	(P2) edge[] node[below right] {Certified semantics!} (ZZ)
	(OC) edge[] node[above] {Extractable!} (ZZ)
	;
\end{tikzpicture}
\vspace{-.2cm}

\caption{Our contribution at a glance.}\label{fig:dia}
\vspace{-.3cm}
\end{figure}
In this section, we present our formalised results and
the relationship that puts them together to build \dslName{}; and we show,
with an example, how our development allows to \emph{certify}
the implementation of a multiparty protocol.

\subsection{Results and Development}
\label{subsec:diagram}
Figure~\ref{fig:dia}
summarises our contribution.
The yellow
rectangle on the background encases the metatheory that we have
formalised for types and processes. On such solid basis,
we build \dslName, our language for specifying end-point processes.

\myparagraph{
  Types as Trees, Projection and Unravelling.}
We formalise in Coq the inductive syntaxes of global types
and local types.
Of these,
we give an alternative representation in terms of coinductive trees,
moving one step forward towards semantics. By defining
the unravelling relation $\Re$, of a type into a tree (\sec\ref{subsec:global}),
and projections $\upharpoonright$, from
global to local objects (\sec\ref{subsec:projection}), we
prove Theorem \ref{thm:unr-pro}:
projection is preserved by unravelling (square \textcolor{orange}{(M.1)}
in Figure \ref{fig:dia}).

\myparagraph{
  Trace Semantics.} Moving further to
the right, we define labelled
transition systems for trees
(\sec\ref{subsec:async_projection} and
\ref{subsec:step}). Exploiting their tree representation,
we give an asynchronous semantics in terms of execution traces
to global and local types (\sec\ref{subsec:trace-eq}).
\emph{Soundness and completeness} come together
in the \emph{trace equivalence theorem} for global
and local types, Theorem \ref{thm:trace-equiv}, thus closing square
\textcolor{orange}{(M.2)} in Figure \ref{fig:dia}.

\myparagraph{
  Process Language and Typing.}
We formalise the syntax for specifying (core) processes, $\proc$ in
Figure~\ref{fig:dia} (\sec\ref{sec:core-procs}). 
We define a typing relation between local types and
processes, then we give semantics to processes (\sec\ref{sec:zooidsem}), again in terms of an LTS and execution traces,
and finally we prove type preservation, Theorem
\ref{thm:preservation}. We conclude the metatheory part
with Theorem
\ref{thm:zooid}, (thus closing square \textcolor{orange}{(M.3)} of Figure~\ref{fig:dia}):
we show that \emph{process traces are global traces}.

\subsection{Process Language: \dslName}
\label{subsec:dsl:overview}
On the foundations of a formalised metatheory, we build a domain
specific language embedded in Coq, \dslName, as presented in \sec~\ref{sec:proc} and \ref{sec:zooid}. Processes specified in \dslName
are well-typed by construction. \dslName{} terms are dependent pairs
of a core process $\proc$, 
and a proof that it is
well-typed with respect to a given local type $\lT$, obtained via
projection of the global type $\G$ given for the protocol. \dslName{}
terms are built using a collection of \emph{smart constructors}: we
make sure that the local type of any smart constructor is fully
determined by its inputs, so that we can use Coq to infer the local
type for every \dslName{} process.

To summarise, our end product \dslName{} is a DSL embedded in Coq. The user specifies as inputs:
\begin{enumerate}[leftmargin=15pt, nosep]
\item the general discipline of the protocol as a global type;
\item the communicating process they are interested in, as a \dslName{} term.
\end{enumerate}
From this the user will obtain:
\begin{enumerate}[label=(\alph*), leftmargin=15pt, nosep]
\item a collection of local types inferred by projection from the given global type;
\item that their process is well-typed by construction;
\item a certified semantics for their process, namely the guarantee that the behaviour of their process adheres to the semantics of the global protocol.
\end{enumerate}
Moreover the user's process is easily translated to an OCaml program, thanks to Coq code-extraction.

\subsection{\dslName{} at Work}
\label{subsec:first-zooid}

We briefly illustrate how \dslName{} works with a simple example, a
ring protocol. We want to write a certified process for $\Alice$ that
sends a message to $\Bob$ and then receives a message from $\Carol$,
but only after $\Bob$ and $\Carol$ have exchanged a message themselves.
In what follows, all the considered messages are natural numbers
of type $\tnat$.

First, we provide \dslName{} with the intended disciplining protocol,
a global type $\G$:\\[0.5mm]
{\small
\centerline{
  $\begin{array}{@{}l@{\;}l@{}}
    \G = \msg{\Alice}{\Bob} \dgt{\ell(\tnat) .}
     & \msg{\Bob}{\Carol} \dgt{\ell(\tnat) .} \\
    &
    \msg{\Carol}{\Alice} \dgt{\ell(\tnat) .}
    \gend
  \end{array}
$}
}\\[0.5mm]
The global type $\G$ prescribes the full protocol, where $\Alice$
sends a message containing a $\tnat$ number to $\Bob$ (with a
generic label $\ell$), $\Bob$ receives it and sends another number
to $\Carol$, who receives and can send the last message to $\Alice$.
$\Alice$ receives and the protocol terminates ($\gend$).

Taking the point of view of $\Alice$, we automatically obtain
a local type $\lT$, \emph{projection} of $\G$ onto the role $\Alice$:
\\[0.5mm]
{\small
\centerline{
$\lT = \lsnd{\Bob} \dlt{\ell(\tnat).}
\lrcv{\Carol}\dlt{\ell(\tnat).\lend}$,
}}\\[0.5mm]
which prescribes for $\Alice$ that she will send a number to $\Bob$,
receive a number from $\Carol$ and terminate.

A \dslName{} implementation for $\Alice$'s process,
respecting $\lT$, is ($\Alice$ sends $x$ to $\Bob$ and
gets $y$ from $\Carol$):
\\[0.5mm]
{\small
\centerline{
  $\begin{array}{@{}l@{\;}l@{}}
    \proc = & \zsend{\Bob}{\ell}{x : \tnat} \\
    &\zrecv{\Carol}{\ell}{\dcoq{y} : \tnat} \zend
  \end{array}
$
}}\\[0.5mm]
Thanks to \dslName{}'s smart constructors, we obtain that
$\proc$ is \emph{well-typed} with respect to the local type $\lT$.
Additionally, the underlying metatheory certifies, by Coq
proofs, that the 
behaviour of $\proc$ conforms to the semantics of protocol $\G$.


\section{Sound and Complete Asynchronous Multiparty Session Types}
\label{sec:MPST}
In this section, we describe the first layer of \dslName{}'s certified
development: a mechanisation of the metatheory of multiparty session types. We
focus on the design, main concepts and results, while for
a more in-detail presentation with pointers to the Coq mechanisation,
we refer to \iftoggle{full}{Appendix \appref{appendix:meta-proofs}}{\cite{fullversion}}.

\subsection{Global and Local Types}
\label{subsec:global}\label{subsec:loc}

A \emph{global type} describes the communication protocol in its entirety,
recording all the interactions between the different participants. Each
participant has a \emph{local type} specifying its intended behaviour within the
protocol. The literature offers a wide variety of presentations of global and
local types \cite{Honda2008Multiparty,HYC2016,Scalas:2019,CDPY2015}:
here, building on
\cite{DenielouYoshida2013}, we formalise full \emph{asynchronous multiparty
  session types} (MPST), which captures asynchronous communication,
with choice and recursion.
\begin{definition}[Sorts, global and local types]
  \label{def:global-types}\label{def:local-types}%
        {\em Sorts} ({\small \code{mty}} in {\small \code{Common/AtomSets.v}}),
        {\em global types} ({\small \gty} in {\small \code{Global/Syntax.v}}),  and
        {\em Local types} ({\small \code{l\_ty}} in {\small \code{Local/Syntax.v}}), ranged over by $\tS$, $\G$, and $\lT$ respectively,
        are generated by:\\[0.5mm]
\centerline{
\small
  $\begin{array}{@{}r@{\;}l@{}}
    \tS & \Coloneqq \tnat \SEP \tint \SEP\tbool \SEP \tplus \tS \tS \SEP \tpair \tS \tS 
    \\[1mm]
    \G & \Coloneqq %
              \gend \;\SEP\;%
              \gX \;\SEP\;%
              \grec \gX \G \;\SEP\;
              \msgi \p\q \ell {\tS} {\G}%
              \\[1mm]
    \lT & \Coloneqq \lend \SEP \lX \SEP \lrec \lX \lT
               \SEP  \lsend \q \ell {\tS} {\lT}
               \SEP \lrecv \p \ell {\tS} {\lT}
  \end{array}
$}\\[0.5mm]
with $\p\neq \q$, $I\neq \emptyset,$ and
$\ell_i \neq \ell_j$ when $i \neq j,$ for all $i,j\in I$.
\end{definition}
\noindent Above, \emph{sorts} refer to the
types of supported message payloads.
We are interested in types such that (1) bound variables are
\emph{guarded}---e.g.,
$\gmu\gX.\allowbreak \msg \p\q \dgt{\ell (\tnat). \G}$ is a valid
global type, whereas $\gmu\gX\dgt{.}\gX$ is not---and (2) types are
\emph{closed}, i.e., all 
variables are bound by $\gmu\gX$
(\iftoggle{full}
         {Appendix \appref{appendix:meta-proofs}, Definitions \appref{def:guarded-A} and \appref{def:closed-A}}
         {\cite[Appendix \appref{appendix:meta-proofs}, Definitions \appref{def:guarded-A} and \appref{def:closed-A}]{fullversion}}).

In the literature, it is common 
to adopt the
\emph{equi-recursive viewpoint} \cite{PierceBook}, i.e., to identify
$\dgt{\gmu\gX.\G}$ and $\dgt{\G \{ \gmu\gX.\G / \gX}\dgt{ \}}$, given that
their intended behaviour is the same. Such unravelling of
recursion can be performed infinitely many times, thus
obtaining possibly infinite trees\footnote{
  Formally, in Coq, a coinductively defined datatype (codatatype)
  of finitely branching trees with possible infinite depth.},
whose structure derives from the
syntax of global and local types \cite{GHILEZAN2019127}.

\begin{definition}[Semantic global and local trees]
  \label{def:global trees} \label{def:local-trees}
  {\em Semantic global trees} ({\small \rgty} and {\small \igty} in
  {\small \code{ Global/Tree.v}}, \iftoggle{full}{see also
Appendix \appref{appendix:meta-proofs}, Remark
\appref{remark:prefixes-A}}{see also
\cite[Appendix \appref{appendix:meta-proofs}, Remark \appref{remark:prefixes-A}]{fullversion}}),
  ranged over by $\coG$, and
  {\em semantic local trees} ({\small \code{rl\_ty}} in
  \code{\small Local/Tree.v}), ranged over by $\lT$,
  are generated \emph{coinductively} by:\\[0.5mm]
\centerline{
  \small
  $\begin{array}{@{}r@{\;}l@{}}
    \coG &  ::=   \cogend \SEP
    \comsgni \p\q \ell {\tS} {\coG} \SEP
    \comsgsi  \p\q {\ell_j} \ell {\tS} {\coG}\\[1mm]
    \colT  & ::=  \colend\;\SEP\;\colsend \p \ell {\tS} {\colT}
    \;\SEP\;   \colrecv \q \ell {\tS} {\colT}
  \end{array}$
}\\[0.5mm]
  with $\p\neq \q$, $I\neq \emptyset,$ and
  $\ell_i \neq \ell_j$ when $i \neq j,$ for all $i,j\in I$. %
\end{definition}

Global and local objects share the
type for a terminated protocol $\mathtt{end}$,
the injection of a variable $X$,
and the recursion construct $\mu X.\dots$; semantic global/local
trees do not include the last two constructs,
since recursion is captured by infinite depth
\iftoggle{full}{(Appendix \appref{subsec:global-A} and \appref{subsec:loc-A})}{(\cite[Appendix \appref{subsec:global-A} and \appref{subsec:loc-A}]{fullversion})}.
\textbf{Global messages:}  
$\msgi \p\q \ell {\tS} {\G}$ describes a protocol where %
participant $\p$ sends to $\q$ one message %
with label $\ell_i$ and a value of sort $\dte{\tS_i}$ as payload, for
some $i \in I$; %
then, depending on which $\ell_i$ was sent by $\p$, %
the protocol continues as $\G_i$. With trees, we make explicit
the two asynchronous stages of the communication of a
message: $\comsgni \p\q \ell {\tS} {\coG}$ represents the status where a
message from $\p$ to $\q$ has yet to
be sent; $\comsgsi \p\q {\ell_j} \ell {\tS} {\coG}$ represents the next
status: the label $\ell_j$ has been selected,
$\p$ has sent the message, with payload $\tS_j$,
but $\q$ has not received it yet. \textbf{Local messages:}
\emph{send type} $\lsend \q \ell {\tS} {\lT}$:
the participant 
sends a message to $\q$; if the participant chooses the label $\ell_i$, %
then the sent payload value must be of sort
$\tS_i$, and it continues as prescribed by $\dlt{\lT_i}$. %
\emph{Receive type} $\lrecv \p \ell {\tS} {\lT}$: %
the participant waits %
to receive from $\p$ a value of sort $\tS_i$, for some $i \in I$, %
via a message with label $\ell_i$; then the protocol continues
as prescribed by $\dlt{\lT_i}$. The same intuition holds, mutatis mutandis, for trees.

We define the function $\parti$ to return the set of participants (or
\emph{roles}) of a global type; e.g. $\p$ and $\q$ above.  For global trees, we
define the predicate \pof{}.  The formal definitions can be found in
\iftoggle{full}{Appendix \appref{subsec:global-A}}{\cite[Appendix \appref{subsec:global-A}]{fullversion}}.

We formalise equi-recursion by relating
types with their representation as trees, as follows:
\begin{definition}[Unravelling]
  \label{def:gunroll}%
  \label{def:lunroll}%
 {\em Unravelling of global types types} (\code{\small GUnroll} in
  \code{\small Global/Unravel.v}) and {\em unravelling of local types} (\code{\small LUnroll} in
  \code{\small Local/Unravel.v}) are the relations between global/local types and semantic
  global/local trees coinductively defined by: 
  \[
  \small
  \hspace{-2mm}\begin{array}{l}
    \begin{array}{ll}
      \rulename{g-unr-end}\ \newDfrac{ }{\gunroll \gend \cogend}\ &
      \rulename{g-unr-rec}\ \newDfrac{\gunroll {\dgt{\G \{ \gmu\gX.\G / \gX \}}} \coG}{\gunroll{\grec \gX \G} \coG}
  \end{array}\\[2mm]
\begin{array}{l}
     \rulename{g-unr-msg}\ \newDfrac{\forall i\in I. \gunroll {\dgt{\G_i}} {\dcog{\coG_i}} }{\gunroll{\msgi \p\q \ell {\tS} {\G}}{\comsgni \p\q \ell {\tS} {\coG}} }
\end{array}\\[3mm]
%
    \begin{array}{ll}
   \rulename{l-unr-end} & \hspace{-2mm}\rulename{l-unr-send}\\
   \newDfrac{ }{\lunroll \lend \colend} &
       \hspace{-2mm}\newDfrac{\forall i\in I. \lunroll {\dlt{\lT_i}} {\dcol{\colT_i}}}
             {\lunroll{  \lsend \q \ell {\tS} {\lT} }{\colsend \q \ell {\tS} {\colT}}  }

  \\[3mm]
  \rulename{l-unr-rec} &  \hspace{-2mm}\rulename{l-unr-recv}\\
  \newDfrac{\lunroll {\dlt{ \lT \{ \lrec{\lX}{\lT} / \lX \}}} \colT}{\lunroll{\lrec \lX \lT} \colT}
  &\hspace{-2mm}\newDfrac{\forall i\in I. \lunroll {\dlt{\lT_i}} {\dcol{\colT_i}} }{\lunroll{  \lrecv \p \ell {\tS} {\lT}}{\colrecv \p \ell {\tS} {\colT}} }
  \end{array}
  \end{array}
\]
\end{definition}
Representing types in terms of trees
allows for a smoother mechanisation of the semantics.
The unravelling operation formally relates the two representations.

\subsection{Projections, or how to discipline communication}
\label{subsec:projection}

\emph{Projection} is the key operation of
multiparty session types: it extracts a local perspective of the protocol, from
the point of view of a single participant, from the global bird's-eye
perspective offered by global types. We define both inductive and
coinductive projections.

\begin{figure*}
  \begin{footnotesize}
  \begin{subfigure}{\textwidth}
  \[
  \begin{array}{@{}c@{}}
    \begin{array}{@{}lll@{}}
      \rulename{proj-end} &\rulename{proj-send} & \rulename{proj-recv}\\
      \projt{\gend} \pr = \lend  &
        \pr=\p\ \ \text{implies}\ \ \projt{\msgi \p\q \ell {\tS} {\G}} \pr=\lsendni \q \ell {\tS} { {\color{black}\projt{\G_{\dgt i}} \pr} }
      &
        \pr=\q\ \ \text{implies}\ \ \projt{\msgi \p\q \ell {\tS} {\G}} \pr=\lrecvni \p \ell {\tS} { {\color{black}\projt{\G_{\dgt i}} \pr} }
    \end{array}\\[2mm]
    \begin{array}{lll}
      \rulename{proj-var} & \rulename{proj-rec} &\rulename{proj-cont}\\
      \projt{\gX}{\pr}=\lX
      & \projt{(\grec \gX \G)}{\pr}=\lrec \lX {} (\projt{\G}{\pr})\ \text{if}\ \ \guarded  (\projt{\G}{\pr}) &
         \text{$\pr\neq \p$, $\pr \neq \q$ and $\forall i,j\in I$, $\projt{\G_{\dgt{i}}}\pr=\projt{\G_{\dgt{j}}}\pr$  implies}\ \ \projt{\msgi \p\q \ell {\tS} {\G}} \pr = \projt {\G_{\dgt{\bar{\imath}}}} \pr\ \ \text{(with $\bar{\imath}\in I$)}
    \end{array}
    \end{array}
    \]
  \vspace{-1mm}
  \subcaption{Rules for recursive projection, Definition \ref{def:projection}}
  \label{subfig:proj}
  \end{subfigure}
  \begin{subfigure}{\textwidth}
    \[
  \begin{array}{@{}c@{}}
  \begin{array}{@{}lll@{}}
    \rulename{co-proj-send-1} &     \rulename{co-proj-send-2} &    \rulename{co-proj-recv-1}\\
    \newDfrac
        {\pr=\p\ \ \ \ \ \ \ \ \forall i\in I. \coproj {\pr} {\coG_{\dcog i}} {\colT_{\dcol i}} }
 	{\coproj {\pr} {\comsgni \p\q \ell {\tS} {\coG}} {\colsend \q \ell {\tS} {\colT}}}  &
    \newDfrac{\pr\neq\q\ \ \ \ \ \ \ \ \forall i\in I. \coproj {\pr} {\coG_{\dcog i}} {\colT_{\dcol i}} }
             { \coproj {\pr} {\comsgsi  \p\q {\ell_j} \ell {\tS} {\coG} } {\colT_{\dcol j } } }
             &
             \newDfrac{\pr=\q\ \ \ \ \ \ \ \ \forall i\in I. \coproj {\pr} {\coG_{\dcog i}} {\colT_{\dcol i}} }
             { 	\coproj {\pr} { \comsgni \p\q \ell {\tS} {\coG}} {\colrecv \p \ell {\tS} {\colT} } }
  \end{array}\\[4mm]
  \begin{array}{@{}lll@{}}
          \rulename{co-proj-recv-2} & \rulename{co-proj-cont}  &   \rulename{co-proj-end}\\

           \newDfrac{\pr=\q\ \ \ \ \ \ \ \ \forall i\in I. \coproj {\pr} {\coG_{\dcog i}} {\colT_{\dcol i}} }
          {\coproj {\pr} {\comsgsi  \p\q {\ell_j} \ell {\tS} {\coG} } {\colrecv \p \ell {\tS} {\colT} } }&
 \newDfrac{
          \pr\neq\p\ \ \ \pr\neq\q\ \ \ \forall i\in I. \coproj {\pr} {\coG_{\dcog i}} {\colT_{\dcol i}} \ \ \ \forall i,j\in I.  \colT_{\dcol i}=\colT_{\dcol j}\ \ \
          \forall i\in I. \partof \pr {\coG_{\dcog i}} }
          { \coproj {\pr} { \comsgni \p\q \ell {\tS} {\coG} } {\colT_{\dcol{\bar{\imath}}} } \ \ \text{(with $\bar{\imath}\in I$)}}
          &
    \newDfrac{\neg\ \partof \pr \coG}{\coproj {\pr} {\coG}  {\colend}}
  \end{array}
  \end{array}
  \]
  \vspace{-1mm}
  \subcaption{Rules for coinductive projection, Definition \ref{def:co-projection}}
  \label{subfig:coproj}
  \end{subfigure}
  \vspace{-3mm}\caption{Projection rules}
  \label{fig:projection}
\end{footnotesize}\vspace{-3mm}
\end{figure*}

\begin{definition}  \label{pro}%
\label{def:projection}%
The inductive \emph{projection 
  of a global type onto a participant $\pr$} (\code{\small project} in
\code{\small Projection/IProject.v}) is a partial function
${\small \projt{\_} \pr: \gty \nrightarrow \lty}$
defined by recursion on $\G$ whenever one of the clauses in
Figure~\ref{subfig:proj} applies and the recursive call is defined;
%
%
  \label{co-pro}%
  \label{def:co-projection}%
  the coinductive \emph{projection 
    of a global tree onto a participant $\pr$} (\iftoggle{full}{definitions \code{\small Project} and \code{\small IProj} in \code{\small Projection/CProject.v}}{\code{\small Project}
  and \code{\small IProj} in \code{\small Projection/}\\ \code{\small CProject.v}}) is a relation
  ${\small \coproj {\pr} {\_} {\_}:\rel{\hspace{-0.5mm}\cogty}{\hspace{-0.5mm}\colty}}$ coinductively
  defined in Figure \ref{subfig:coproj}.
\end{definition}

In rules \rulename{co-proj-end} and \rulename{co-proj-cont} we have
added explicit conditions on participants. By factoring in the
predicate \pof , Definition \ref{co-pro} ensures (1) that the
projection of a global tree on a participant outside the protocol is
$\colend$ (rule \rulename{co-proj-end}) and (2) that this discipline
is preserved in the continuations (rule \rulename{co-proj-cont}). We see
that the clauses for projecting of types and trees follow the same
intuition: projecting a global object onto a sending (resp. receiving) role gives a
sending (resp. receiving) local object, provided that the local continuations
are also projections of the corresponding global continuations. As expected,
the tree projection takes care explicitly of
asynchronicity (rules \rulename{co-proj-send-2} and \rulename{co-proj-recv-2}).
This is an adaptation to our coinductive setting of the
definition in \cite[Appendix A.1]{DenielouYoshida2013}. Below we give an example
to clarify the meaning of \rulename{proj-cont}.

%
\begin{example}[Projection]\label{ex:projection}
  About rule \rulename{proj-cont}, we observe that the type
  \iftoggle{full}{
    \[
    \begin{small}
    \begin{array}{lll}
       \dgt{\G'} = \msg{\Alice}{\Bob} \dgt{\{}&\dgt{\ell_1(\tnat) .}
       \msg{\Bob}{\Carol}\dgt{ \ell(\tnat) .\gend ,}&\\
    &\dgt{\ell_2(\tnat) .
    \msg{\Alice}{\Carol}} \allowbreak \dgt{\ell(\tnat) .\allowbreak\gend} &\dgt{\}}
    \end{array}
    \end{small}
    \]
  }{
  {\small $ \dgt{\G'} = \allowbreak \msg{\Alice}{\Bob} \allowbreak \dgt{\{ \ell_1(\tnat) .}
    \msg{\Bob}{\Carol}\allowbreak\dgt{ \ell(\tnat)}\allowbreak
  \dgt{ .\gend ,\ell_2(\tnat) .
    \msg{\Alice}{\Carol}} \allowbreak \dgt{\ell(\tnat) .\allowbreak\gend \}}
  $} }
  is not projectable onto \Carol{}, since, after skipping the first interaction between
  \Alice{} and \Bob{}, it would not be clear whether \Carol{} should
  expect a message
  from  \Alice{} or from \Bob{}. If we take instead
  \iftoggle{full}{
    \[
    \begin{small}
      \begin{array}{lll}
            \dgt{\G} = \msg{\Alice}{\Bob}
            \dgt{\{}& \dgt{\ell_1(\tnat) .\msg{\Bob}{\Carol} \ell(\tnat) .\gend ,}&\\
          &  \dgt{ \ell_2}\dgt{(\tbool)}\dgt{ .
    \msg{\Bob}{\Carol}}
  \dgt{\ell(\tnat).\gend}&\dgt{ \}}\text{},

    \end{array}
    \end{small}
    \]
  }{\small $
    \dgt{\G} = \msg{\Alice}{\Bob}
    \dgt{ \{ \ell_1(\tnat) .
      \msg{\Bob}{\Carol}}\allowbreak
    \dgt{ \ell(\tnat) .}\allowbreak \dgt{\gend , \ell_2}\allowbreak\dgt{(\tbool)}
        \allowbreak \dgt{ .
    \msg{\Bob}{\Carol}}
  \dgt{\ell(\tnat).\gend \}}
  $,}
  the projection $\projt\G \Carol$ is well defined as the local type
  {\small $\lT = \lrcv{\Bob}\allowbreak\dlt{\ell(\tnat).\lend}$}. Following common
  practice, we use an option type to encode projection as
  a partial function in Coq.
\end{example}
%

Coinductive projection is more permissive than its inductive
counterpart, since it removes the technical issues related to formally
dealing with (equi)recursion, \iftoggle{full}{thus }{}allowing for a smoother development in Coq 
(\iftoggle{full}{Appendix \appref{subsec:projection-A}}{\cite[Appendix \appref{subsec:projection-A}]{fullversion}}
and \citeN[Definition 3.6 and Remark 3.14]{GHILEZAN2019127}).

If, when reasoning about semantics, coinductive trees are
more convenient objects to work with, we still want to rely
on session types for imposing a typing discipline on
the communication. The followng theorem allows us to do so.

\begin{restatable}[Unravelling preserves projections]{theorem}{unrpro}\label{thm:unr-pro}(\code{\small ic\_proj} in \code{\small Projection/Correctness.v}.)
  Given a global type $\G$, such that $\guarded\ \G$ and $\closed\ \G$, if
\begin{enumerate*}[label=(\alph*)]
\item there exists a local type $\lT$ such that $\projt{\G} \pr = \lT\ $,
\item there exists a global tree $\coG$ such that $\gunroll \G \coG$ and,
\item there exists a local tree $\colT$ such that $\lunroll \lT \colT$,
\end{enumerate*}
then $\coproj \pr \coG \colT$.
\end{restatable}

This first central result closes the first metatheory square \textcolor{orange}{(M.1)}
of \theDiagram\ in Figure \ref{fig:dia}. For a sketch of its proof
see
\iftoggle{full}{Appendix \appref{appendix:meta-proofs}, Theorem \appref{thm:unr-pro-A}}{\cite[Appendix \appref{appendix:meta-proofs}, Theorem \appref{thm:unr-pro-A}]{fullversion}}.

\subsection{Projection Environments for Asynchronous Communication}
\label{subsec:async_projection}

In this subsection, we introduce key concepts for building an
asynchronous operational semantics for MPST.
In \cite{DenielouYoshida2013} a precise correspondence is drawn
between communicating finite-state
automata and MPST. We do not formalise an explicit
syntax for automata, but develop labelled transition systems for
global and local trees with automata in mind.

Consider the following scenario:
$\p$ sends a message to $\q$
with label $\dcol{\ell}$ and payload of sort $\tS$ and continues on
$\colT$, and dually $\q$ receives from $\p$ the message, with same
label and payload, and then continues on $\dcol{\colT'}$.
For $\q$ to receive the message, it is necessary that $\p$ has first
sent it. To model this asynchronous behaviour, we use FIFO
queues: in the designated queue $Q(\p,\q)$ (empty at first) we
enqueue the message sent from $\p$, until the
message is received by $\q$ and removed from the queue.
%
We use one queue for each ordered
  pair of participants $(\p,\q)$ to store in-transit messages sent
  from $\p$ to $\q$, and we collect such queues in \emph{queue environments}.
\begin{definition}[Queue environments]
\label{def:qenv}
We call \emph{queue environment} (notation {\small \qenv} in \code{\small Local/Semantics.v}) any finitely supported function that
maps a pair of participants into a finite sequence (\emph{queue}) of
pairs of labels and sorts.
\end{definition}


We define the operations of enqueuing and
dequeuing on queue environments:
\\[0.5mm]
\centerline{
$  \small
\begin{array}{lll}
\enq\ Q\ (\p,\q)\ (\ell,\tS) &=& Q[(\p,\q)\mapsfrom Q(\p,\q)@(\ell,\tS)]\\
\deq\ Q\ (\p,\q) &=& \code{if}\ Q(\p,\q)=(\ell,\tS)\#s\\
                  && \code{then}\ ((\ell,\tS),Q[(\p,\q)\mapsfrom s])\ \ \ \code{else}\ \code{None}
\end{array}
$}\\[0.5mm]
We use $\#$ as the ``cons'' constructor for lists and
$@$ as the ``append'' operation; $f[x\mapsfrom y]$ denotes the
updating of a function $f$ in $x$ with $y$, namely
$f[x\mapsfrom y]\ x'=f\ x'$ for all $x\neq x'$ and
$f[x\mapsfrom y]\ x=y$.  We use option types for
partial functions, with $\code{None}$ as the standard returned
value where the function is undefined.
In case the sequence $Q(\p,\q)$ is empty \deq~will
not perform any operation on it, but return $\code{None}$; in case the
sequence is not empty it will return both its head and its tail (as a
pair). We denote the empty queue environment by $\epsilon$, namely
$\epsilon\ (\p,\q)=\code{None}$ for all $(\p,\q)$.

Global trees can represent stages of the execution, where
a participant has already sent a message, but it has not yet been received.
We adapt the ``queue projection'' from \citep[Appendix
A.1]{DenielouYoshida2013} to our coinductive setting, to associate
global trees to the queue contents of a system.

\begin{definition}[Queue projection](Definition \code{\small qProject} in
  \code{\small Projection/QProject.v})
  \label{qpro}%
  \label{def:q-projection}%
  \emph{Projection on queue environments
  of a global tree} (\emph{queue projection} for short) %
  is the relation \iftoggle{full}{\\}{}${\small \qproj {\_}  {\_}\ :\ \rel{\cogty}{\qenv} }$ coinductively specified by:
  \[
  \small
\begin{array}{ll}
    \rulename{q-proj-send}
    \newDfrac{\forall i\in I. \qproj {\coG_{\dcog i}} {Q}\ \ \ \ \ \ \ \ Q(\p,\q)=\code{None}
    }{\qproj {\comsgni \p\q \ell {\tS} {\coG}} {Q}} &
  \newDfrac{\rulename{q-proj-end}}{\qproj \cogend  {\epsilon}} \\[4mm]
    \multicolumn{2}{c}{    \rulename{q-proj-recv} \
      \newDfrac{\qproj {\coG_{\dcog{j}}} {Q}\ \ \ \ \ \ \ \ \deq\ Q'(\p,\q)= ((\ell_j, \tS_j),Q)}
               {\qproj {\comsgsi \p\q {\ell_j} \ell {\tS} {\coG}} {Q'}}
    }
  \end{array}
  \]
See \iftoggle{full}{Appendix \appref{subsec:async_projection-A}}{\cite[Appendix \appref{subsec:async_projection-A}]{fullversion}} for more details.
\end{definition}

Analogously to queue environments,
we consider all the local
types of the protocol at once. 
\begin{definition}[Local environments]
\label{def:env}
We call \emph{local environment}, or simply \emph{environment}, any
finitely supported function $\dcol{E}$ that maps participants into
local types.
\end{definition}

We are interested in those environments that are defined on the participants
of a global protocol $\coG$ and that map each
participant $\p$ to the projection of $\coG$ onto such $\p$.
\begin{definition}[Environment projection](Definition\\ \code{\small eProject} in \code{\small Projection/CProject.v}.)
\label{def:epro}
We say that $\dcol{E}$ is an environment projection for $\coG$,
notation $\coG\upharpoonright\dcol{E}$, if it holds that
\ $\forall \p.\ \coproj \p  \coG {(\dcol{E\ \p})}$.
\end{definition}
%

We define the semantics on
a set of local types together with queue environments. We therefore consider the
projection of a global tree both on local environments and on queue
environments, together in one shot.
\begin{definition}[One-shot projection](Definition\\ \code{\small Projection} in \code{\small Projection.v})
\label{def:ospro}
We say that the pair of a local environment and of a queue environment $(\dcol{E},Q)$ is a (one-shot) projection for the global tree $\coG$, notation $\coG\osproj(\dcol{E},Q)$ if it holds that:
$\coG\upharpoonright\dcol{E} \quad \text{and} \quad \qproj \coG Q$.
\end{definition}
\begin{example}
  Let us consider the global tree:
  $\coG = \comsgs \p \q \ell \allowbreak\dcog{\ell(\tS).}\comsgn \q \p
  \dcog{\ell(\tS).}\comsgn \q \p \allowbreak\dcog{\ell(\tS).\dots}$ .
  Participant $\p$ has sent a message to $\q$, $\q$ will receive it
  next (but has not yet) and then the protocol continues indefinitely
  with $\q$ sending a message to $\p$ after the other. We define
  $\dcol{E}$ 
  such that:
  $ \ \ \dcol{E\ \p}=\colrcv \q \dcol{\ell(\tS).}\colrcv \q
  \dcol{\ell(\tS).\dots}\ \ $ and
  $ \ \ \dcol{E\ \q}=\colrcv \p \dcol{\ell(\tS).}\colsnd \p
  \dcol{\ell(\tS).}\allowbreak \colsnd \p \dcol{\ell(\tS).\dots}\ $ . We then
  define $Q$
  such that: $\ \ Q(\p,\q)=[(\ell,\tS)]\ \ $ and
  $\ \ Q(\q,\p)=\code{None}$. It is easy to verify
  that $\coG\osproj(\dcol{E},Q)$; observe
  that the only ``message'' enqueued in $Q$ is $(\ell,\tS)$, since
  this is the only one sent, but not yet received (at this stage of
  the execution).
\end{example}

\subsection{Labelled Transition Relations for Tree Types}
\label{subsec:step}

At the core of the \emph{trace semantics} for session types
lies a labelled transition system (LTS) defined on trees,
with regard to \emph{actions}. The basic \emph{actions}
(datatype \code{\small act} in \code{\small Common/Actions.v}) of
our asynchronous communication are objects, ranged over by $a$, of the
shape either:
$!\p\q(\ell,\tS)$: send $!$ action, from participant $\p$ to
  participant $\q$, of label $\ell$ and payload type $\tS$, or
$?\q\p(\ell,\tS)$: receive $?$ action, from participant $\p$ at
  participant $\q$, of label $\ell$ and payload type $\tS$.
We define the \emph{subject} of an action $a$ (definition
\code{\small subject} in \code{\small Common/Actions.v}), $\subject a$, as $\p$ if
$a=!\p\q(\ell ,\tS)$ and as $\q$ if $a=?\q\p(\ell,\tS)$.%
\footnote{The representation of actions is directly taken from
  \cite{DenielouYoshida2013}, however we have swapped the order of
  $\p$ and $\q$ in the receive action, so that the subject of an
  action always occurs in first position.}
Given an action, our types (represented as trees) can perform a
reduction step.

\begin{definition}[LTS for global trees]\label{def:gstep}\ \\(\code{\small step} in \code{\small Global/Semantics.v})
  The \emph{labelled transition relation for global trees}
  (\emph{global reduction} or \emph{global step} for short) %
  is, for each action $a$, the relation
  $\_\ \stepa{a}\ \_\ :\ \rel{\cogty}{\cogty} $ inductively specified
  by the following clauses:\\[1mm]
\centerline{
$\small
\begin{array}{c}
  \rulename{g-step-send}
  \dfrac{a=!\p\q(\ell_j,\tS_j)}{\comsgni \p \q \ell \tS \coG \stepa{a} \comsgsi \p \q {\ell_j} \ell \tS \coG} \\[5mm]
  \rulename{g-step-recv}
  \dfrac{a=?\q\p(\ell_j,\tS_j)}{  \comsgsi \p \q {\ell_j} \ell \tS \coG \stepa{a} \coG_{\dcog{j}} }\\
\rulename{g-step-str1}
\dfrac{\subject a \neq \p \quad \subject a \neq \q \quad \forall i\in I. \coG_{\dcog{i}}\stepa{a}\dcog{{\coG}_i'}}{ \comsgni \p \q \ell \tS \coG  \stepa{a} \comsgni \p \q \ell \tS {\coG'} }\\[2mm]
 \rulename{g-step-str2}
\dfrac{\subject a \neq \q \quad \coG_{\dcog{j}}\stepa{a}\dcog{{\coG}_j'}\quad \forall i\in I\backslash\{ j \}. \coG_{\dcog{i}}=\dcog{{\coG}'_i}}{ \comsgsi \p \q {\ell_j} \ell \tS \coG  \stepa{a} \comsgsi \p \q {\ell_j} \ell \tS {\coG'} }
\end{array}
$}
\end{definition}
The step relation describes a labelled transition system for global
trees with the following intuition: \rulename{g-step-send}  \emph{sending base case}:
with the sending action $!\p\q(\ell_j,\tS_j)$,
a message with label $\ell_j$ and payload type
  $\tS_j$ is sent by $\p$, but not yet received by $\q$;
  \rulename{g-step-recv}  \emph{receiving base case}: with the receiving
  action $?\q\p(\ell_j,\tS_j)$,
  a message with label $\ell_j$ and payload type $\tS_j$,
  previously sent by $\p$, is now received by $\q$;
  in \rulename{g-step-str1}, a step is allowed
  to be performed under a sending constructor
  $\dcog{\p \to \q}$: each time that the subject of that
  action is different from $\p$ and from $\q$ and each
  continuation steps;
  \rulename{g-step-str2} with an action $a$ a step is allowed
  to be performed under a receiving constructor:
each time that the subject of
  that action is different from $\q$ ($\p$ has already sent the
  message and the label $\ell_j$ has already been selected), the
  continuation corresponding to $\ell_j$ steps
  and others stay as the same.

This semantics allows for some degree of non-determinism. For instance,
$\comsgni \p \q \ell \tS \coG$ could perform a step according to both
rules \rulename{g-step-send} and \rulename{g-step-str1} (depending
on the subject of the action).


Below we 
define a transition system for environments of local trees, together
with environments of queues.
\begin{definition}[LTS for
  environments]\label{lstep}\label{def:lstep}\
  (\code{\small l\_step} in\\ \code{\small Local/Semantics.v})\
  The \emph{labelled transition relation for environments}
  (\emph{local reduction} or \emph{local step} for short) is, for each
  $a$, the relation
  $\_\ \stepa{a}\ \_\ :\ \rel{(\env * \qenv)}{(\env * \qenv)} $
  inductively specified by the following clauses:\\[1mm]
\centerline{
 $\small
\begin{array}{c}
  \rulename{l-step-send}\
  \dfrac{a=!\p\q(\ell_j,\tS_j)\quad \dcol{E\ \p} = \colsend \q \ell \tS \colT}
        {(\dcol{E},Q) \stepa{a} (\dcol{E[\p\mapsfrom \colT_j]},\enq\ Q\ (\p,\q)\ (\ell_j,\tS_j))}\\[5mm]
  \rulename{l-step-recv}\
  \dfrac{
\begin{array}{l}
a=?\q\p(\ell_j,\tS_j)\\
\dcol{E\ \q} = \colrecv \p \ell \tS \colT \quad
Q(\p,\q)=(\ell_j,\tS_j)\#s
\end{array}}
        {(\dcol{E},Q) \stepa{a} (\dcol{E[\q\mapsfrom \colT_j]},Q[(\p,\q)\mapsfrom s])}
\end{array}
$}
\end{definition}\vspace{-2mm}

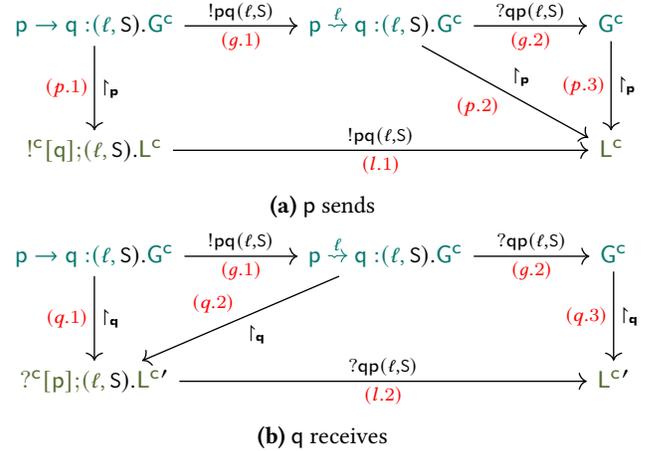
\begin{figure}
\begin{small}
\begin{subfigure}{\columnwidth}
  \[
\begin{tikzcd}
\comsgn \p \q \dcog{(\ell, \tS).} \coG \arrow[rr, "{!\p\q (\ell,\tS)}","(g.1)"' red]\arrow[dd,"\upharpoonright_{\p}","(p.1)"' red]
& &
\comsgs \p \q {\ell}  \dcog{(\ell, \tS).} \coG\arrow[rr, "{?\q\p (\ell,\tS)}","(g.2)"' red]\arrow[rrdd, "\upharpoonright_{\p}","(p.2)"' red]
& &
\coG\arrow[dd, "\upharpoonright_{\p}","(p.3)"' red] \\
& & & &\\
\colsnd \q \dcol{(\ell, \tS).} \colT \arrow[rrrr, "{!\p\q (\ell,\tS)}","(l.1)"' red]& & & &\colT
\end{tikzcd}
\]
\vspace*{-2mm}
\subcaption{$\p$ sends}\label{subfig:send}
\end{subfigure}\newline\vspace*{-2mm}
\begin{subfigure}{\columnwidth}
  \[
\begin{tikzcd}
\comsgn \p \q \dcog{(\ell, \tS).} \coG\arrow[rr, "{!\p\q(\ell,\tS)}","(g.1)"' red]\arrow[dd,"\upharpoonright_{\q}","(q.1)"' red]
& &
\comsgs \p\q \ell \dcog{(\ell, \tS).} \coG\arrow[rr, "{?\q\p(\ell,\tS)}","(g.2)"' red]\arrow[lldd, "\upharpoonright_{\q}","(q.2)"' red]
& &
\coG\arrow[dd, "\upharpoonright_{\q}","(q.3)"' red] \\
& & & &\\
\colrcv \p \dcol{(\ell, \tS).} \dcol{\colT'} \arrow[rrrr, "{?\q\p(\ell,\tS)}","(l.2)"' red]& & & &\dcol{\colT'}
\end{tikzcd}
\]
\vspace*{-2mm}
\subcaption{$\q$ receives}\label{subfig:recv}
\end{subfigure}
\vspace*{-3mm}
\caption{Basic send/receive steps for global and local trees}\label{fig:steps}
\vspace*{-4mm}
\end{small}
\end{figure}

\begin{example}[Basic 
    steps for global and local trees]
  Figure \ref{subfig:send} shows the transitions for a global tree,
  regulating the sending of a message from $\p$ to $\q$, and
  the local transition for its projection on $\p$. The
  asynchronicity of our system is witnessed by the two different
  steps: $\color{red}{(g.1)}$, for the sending action
  $!\p\q(\ell,\tS)$, and $\color{red}{(g.2)}$, for the receiving one
  $?\q\p(\ell,\tS)$. Projecting
  $\comsgn \p \q \dcog{(\ell, \tS).} \coG$ on $\p$ (arrow
  $\color{red}{(g.1)}$) gives us a local tree that performs a sending
  step $\color{red}{(l.1)}$ corresponding to $\color{red}{(g.1)}$, and
  projection is preserved (arrow $\color{red}{(p.2)}$). However this
  does not happen for the receiving step $\color{red}{(g.2)}$: here
  the projections on $\p$ of
  $\comsgs \p\q \ell \dcog{(\ell, \tS).} \coG$ along
  $\color{red}{(p.2)}$ and of $\coG$ along $\color{red}{(p.3)}$ are
  the same. Dually if we consider the projection on the
  receiving participant $\q$, Figure \ref{subfig:recv}. Here the
  projections along $\color{red}{(q.1)}$ and $\color{red}{(q.2)}$,
  corresponding to the global tree performing a sending action, result
  in the same local tree. We have instead a local step
  $\color{red}{(l.2)}$ preserving the local projections on $\q$ along
  $\color{red}{(q.2)}$ and $\color{red}{(q.3)}$ for the receiving
  action along $\color{red}{(g.2)}$.
\end{example}

Figure \ref{fig:steps} confirms our intuition: when the global tree
performs one step, \emph{there is one local tree} (namely, one projection of the
global tree) such that it performs a corresponding
step. We have indeed defined semantics for collections of local trees,
as opposed to single local trees. The formal relation of the
small-step reductions with respect to projection is established with
soundness and completeness results (see
\iftoggle{full}{Appendix \appref{appendix:meta-proofs}}{\cite[Appendix \appref{appendix:meta-proofs}]{fullversion}}
for proof outlines).

\begin{restatable}[Step Soundness]{theorem}{stepsound}(Theorem \code{\small Project\_step} in \code{\small TraceEquiv.v})\label{thm:step-sound}\
  If $\coG \stepa{a} \dcog{\coG'}$ and $\coG \osproj (\dcol{E},Q)$,
  there exist $\dcol{E'}$ and $Q'$ such that
  $\dcog{\coG'} \osproj (\dcol{E'},Q')$ and
  $(\dcol{E},Q) \stepa{a} (\dcol{E'},Q')$.
\end{restatable}
\vspace{-3mm}
\begin{restatable}[Step Completeness]{theorem}{stepcompl}(Theorem \\\code{\small Project\_lstep} in \code{\small TraceEquiv.v})\label{thm:step-compl}\
  If $(\dcol{E},Q) \stepa{a} (\dcol{E'},Q')$ and
  $\coG \osproj (\dcol{E},Q)$, there exist $\dcog{\coG'}$ such that
  $\dcog{\coG'} \osproj (\dcol{E'},Q')$ and
  $\coG \stepa{a} \dcog{\coG'}$.
\end{restatable}

\subsection{Trace Semantics and Trace Equivalence}
\label{subsec:trace-eq}
We finally show trace equivalence for global and
local types with
our Coq development of semantics for coinductive trees.


\begin{definition}[Traces]
  \label{def:traces}%
  (Codatatype \code{\small trace} in \code{\small Action.v}), ranged
  over by $t$, are terms generated \emph{coinductively} by
  $\ \ t\ \ \  ::=\ \trend \;\SEP\; \trnext a t\ \ $ where $a$
  is any action, as defined in
  \sec\ref{subsec:step}\footnote{For traces, we use the same notation
  as for lists, however we bear in mind that this definiton is
  coinductive: it generates possibly infinite streams.}.
\end{definition}

We associate traces to the execution of global trees and local environments.

\begin{definition}[Admissible traces for a global
  tree]\label{def:g-traces}
  We say that a trace is admissible for a global tree if the
  coinductive relation $\glts \_ \_$ (definition \code{\small g\_lts} in
  \code{\small Global/Semantics.v}) holds:
   \vspace*{-1mm}
   \[
   \small
    \newDfrac{}{\glts  \trend \cogend} \quad %
    \newDfrac{\coG \stepa{a}  {\dcog{\coG'}}\quad \glts t {\dcog{\coG'}} }{ \glts {\trnext a t} \coG}
  \]
\end{definition}


\begin{definition}[Admissible traces for environments]\label{l-traces}
  We say that a trace is admissible for a pair of a local environment
  and a queue environment if the coinductive relation $\llts \_ \_$
  (definition \code{\small l\_lts} in \code{\small Local/Semantics.v}) holds:
  \vspace*{1.5mm}
  \[
  \small
\newDfrac{\forall \p. \dcol{E}\ \p=\code{None} 
}{\llts  \trend (\dcol{E},\epsilon)} \quad %
\newDfrac
{(\dcol{E},Q) \stepa{a}  (\dcol{E'},Q')\quad \llts t {(\dcol{E'},Q')}}
{ \llts {\trnext a t} {(\dcol{E},Q)}}
\]
\end{definition}

Observe that generally more than one execution trace
are admissible for a global tree or for an environment\footnote{
  About non-determinism in our semantics, see
  \iftoggle{full}
  {Remark \appref{remark:nondet-A}, Appendix \appref{appendix:meta-proofs}}
  {\cite[Remark \appref{remark:nondet-A}]{fullversion}}.}.

We can now state the \emph{trace equivalence} theorem, our final result
for multiparty session types. We sketch an outline of the proof in \iftoggle{full}
{Appendix \appref{appendix:meta-proofs}, Theorem \appref{thm:trace-equiv-A}}{\cite[Theorem \appref{thm:trace-equiv-A}]{fullversion}}.

\begin{theorem}[Trace equivalence](Theorem \\\code{\small TraceEquivalence} in \code{\small TraceEquiv.v}.)\label{thm:trace-equiv}\ \\
  If $\coG\osproj(\dcol{E},Q)$, then
$\ \glts t \coG\ $
if and only if $\ \llts t (\dcol{E},Q)\ $.
\end{theorem}




Trace equivalence for global and local types (trees)
concludes our formalisation of the metatheory
of multiparty session types: squares
\textcolor{orange}{(M.1)} and \textcolor{orange}{(M.2)} of \theDiagram in Figure
\ref{fig:dia}. In the next section we specify a language
for communicating systems inside Coq and extend extend the trace equivalence
result to well-typed processes.


\section{A Certified Process Language}\label{sec:proc}

This section defines \dslName{}, an embedded domain specific language
in Coq for specifying certified multiparty processes. \dslName{}
combines shallow and deep embedding: on one hand process actions are
deeply embedded, represented as an inductive type; on the other, the
exchanged values, and computations applied to them are a shallow
embedding expressed as Gallina terms. The core process calculus of
\dslName{} is session-typed, where the typing derivation is described
as a Coq inductive predicate. The constructs of \dslName{} are
\emph{smart constructors} that build both a process, and a proof that
this is well-typed with respect to a given local type. \revised{Each
  process is single threaded and the concurrent semantics occurs due
  to the asynchronous nature of the channels.}

\subsection{Core Processes}\label{sec:core-procs}

\begin{figure}
\begin{small}
$
\begin{array}{l} 
\hspace{-5mm}
    \begin{array}{c} 
      \begin{array}{@{}l@{}}
        \rulename{p-ty-end} \\
        \dfrac{}{\Gamma \ofLt \pend : \lend }
      \end{array} \
      \begin{array}{@{}l@{}}
        \rulename{p-ty-jump} \\
        \dfrac{}{\Gamma \ofLt \pjump{\lX} : \lX}
      \end{array} \
      \begin{array}{@{}l@{}}
        \rulename{p-ty-loop} \\
        \dfrac{%
        \Gamma \vdash \expr : \Proc \quad
        \Gamma \ofLt \expr : \lT
        }{%
        \Gamma \ofLt \ploop{\lX}{\expr} : \lrec{\lX}{\lT}
        }
      \end{array} \\
      \begin{array}{@{}l@{}}
        \rulename{p-ty-read} \\
        \dfrac{%
        \Gamma \vdash \pactr : \code{unit} \to \llbracket \tS \rrbracket \quad
        \Gamma,x:\llbracket \tS \rrbracket \ofLt \expr : \lT
        }{%
        \Gamma \ofLt \pread \pactr \expr : \lT
        }
      \end{array}
    \end{array}\\[2mm]
    \begin{array}{cc} 
      \begin{array}{@{}l@{}}
        \rulename{p-ty-write} \\
        \dfrac{%
        \Gamma \vdash \pactw : \llbracket \tS \rrbracket \to \code{unit} \quad
        \Gamma \vdash {\expr_v} : \llbracket \tS \rrbracket \quad
        \Gamma \ofLt \expr : \lT
        }{%
        \Gamma \ofLt \pwrite \pactw {\expr_v} \expr : \lT
        }
      \end{array}
\\
      \begin{array}{@{}l@{}}
        \rulename{p-ty-send} \\
        \dfrac{%
        \Gamma \vdash \expr_1 : \llbracket \tS_j \rrbracket \quad
        \Gamma \vdash \expr_2 : \Proc \quad
        \Gamma \ofLt \expr_2 : \lT_j \quad
        j \in I
        }{%
        \Gamma \ofLt \psend{\p}{\ell_j, \expr_1}{\expr_2} : \lsend{\p}{\ell}{\tS}{\lT}
        }
      \end{array}\\
      \begin{array}{@{}l@{}}
        \rulename{p-ty-recv} \\
        \dfrac{%
        \forall i \in I \quad
        \Gamma \vdash \expr_i : \llbracket \tS_i \rrbracket \to \Proc \quad
        \Gamma, x : \llbracket \tS_i \rrbracket \ofLt \expr_i \; x : \lT_i
        }{%
        \Gamma \ofLt \precv{\p}{\ell_i. \expr_i}_{i \in I} : \lrecv{\p}{\ell}{\tS}{\lT}
        }
      \end{array} \\
      \begin{array}{@{}l@{}}
        \rulename{p-ty-interact} \\
        \dfrac{%
        \Gamma \vdash \pacti :  \llbracket \tS \rrbracket \to \llbracket \tS' \rrbracket \quad
        \Gamma \vdash {\expr_v} : \llbracket \tS \rrbracket \quad
        \Gamma,x:\llbracket \tS' \rrbracket \ofLt \expr : \lT
        }{%
        \Gamma \ofLt \pinteract \pacti {\expr_v} \expr : \lT
        }
      \end{array}
    \end{array}\\
  \end{array}
$
  \end{small}
\vspace*{-2mm}
\caption{Process Typing System}\label{fig:typing}
\vspace*{-4mm}
\end{figure}

The core process calculus of \dslName{} differs to those generally
used in the session-types literature in several aspects. First, the
combination of shallow and deep embedding implies that a process may
be defined in terms of a larger expression of the ambient calculus.
Secondly, the process calculus does not include parallel composition.
Just as \emph{``zooid''}, in biology, is used to refer to the single
individual in a colonial organism, a process $\proc$ implements the
behaviour of a single participant in the distributed system:
we are interested in certifying processes in
isolation to the larger system
. This approach plays well with the usual MPST
methodology and it admits heterogenous development, as in one can use
\dslName{} for the critical roles and other roles can be implemented in
different languages, using different frameworks.

\begin{definition}[Syntax of untyped processes]\label{def:proc-syntax}
  Processes, $\proc$ (definition \code{\small Proc} in \code{\small Proc.v}), are
  embedded in an ambient calculus $\expr$. In our implementation,
  $\proc$ is the inductive type of processes, of type $\Proc$, and the
  ambient calculus is Gallina, the specification language of Coq.
  \begin{displaymath}
  \begin{array}{rl}
    \expr
      \Coloneqq &\proc
      \mid \expr + \expr
      \mid \eif{\expr}\; \ethen{\expr}\; \eelse{\expr}\\
      \mid & \tfun{\vx}{\expr}
      \mid \ldots
      \hspace{.5cm}
      \dproc{n} \in \mathbb{N} 
      \hspace{.5cm}
      \dproc{\ell_i} \in \mathbb{N}
    \\
    \proc \in \Proc
    \Coloneqq & \pend
    \mid \pjump{\lX}
    \mid \ploop{\lX}{\expr}\\
    \mid &\precv{\p}{\ell_i. \expr_i}_{i \in I}
    \mid \psend{\p}{\ell, \expr}{\expr}\\
    \mid & \pread \pactr \expr
    \mid \pwrite \pactw {\expr_v} \expr \\
    \mid & \pinteract \pacti {\expr_v} \expr
  \end{array}
  \end{displaymath}
\end{definition}

The constructs of $\Proc$ mirror those of local types: $\pend$ is the
\textbf{\emph{ended}} process; $\pjump{\lX}$ is a \textbf{\emph{jump}}
to recursion variable $\lX$; $\ploop{\lX}{\expr}$ is a
\textbf{\emph{recursive process}}, built by expression $\expr$, that
introduces a new recursion variable $\lX$;
$\precv{\p}{\ell_i. \expr_i}_{i \in I}$ is the process
\textbf{\emph{receiving}} from $\p$ a message with label $\ell_i$, a
value $x$, and continues as $(\expr_i \; x)$; and
$\psend{\p}{\ell, \expr_1}{\expr_2}$ is the \textbf{\emph{sending}}
process with label $\ell$ and expression $\expr_1$ to participant
$\p$, and then continues as $\expr_2$. Our calculus does not include
parallel composition: we assume that the system is implemented as the
parallel composition of all the participants. For example, the
following is a process that receives requests from a participant $\p$
and replies increasing the received number by $m$, until $\p$ chooses
to finish:
\[
\small
\begin{array}{ll}
  \proc_\q = &
  \dproc{\code{loop}\;\lX\{}
    \dproc{\code{recv}\;\p\{}
    \dproc{\ell_1.}
    \tfun{\vx}{\psend{\p}{\ell_1, \dcoq{x + m}}{}}
    \\
    & \hspace{2cm}\pjump{\lX}\dproc{;}
      \dproc{\ell_2.\tfun{\vx}{\pend}}
    \dproc{\}}
    \dproc{\}}
\end{array}
\]
\noindent
A process can be defined mixing Gallina terms and $\proc$. For example, in the
process above, the term $\dcoq{x + m}$ is a term in Gallina. These Gallina
terms can be used to specify branching in the control flow of the process.
The process below is one possible implementation for $\p$ that loops
until the value received is greater than some threshold $n$:
\begin{displaymath}
  \begin{small}
  \begin{array}{rcl}
    \expr_\p & = &
                   \tfun{x}{} \;
                   \eif{x > n} \;
                   \ethen{%
                   \psend{\q}{\ell_2, \dcoq{\code{tt}}}{\pend}
                   }\\& & \eelse{%
                   \psend{\q}{\ell_1,\dcoq{x}}{\pjump{\lX}}
                   }

    \\
    \proc_\p & = & \psend{\q}{\ell_1, \dcoq{0}}{\ploop{\lX}{\precv{\q}{\ell_1. \expr_\p}}}
  \end{array}
  \end{small}
\end{displaymath}

\dslName processes interact with their environment by calling
functions written in the language of the runtime (OCaml in this case).
These functions exchange information between \dslName and the
environment in a safe way by not exposing channels or the transport
API. The interaction happens by calling an external function: \pactr,
\pactw, and \pacti for reading, writing or interacting with the
environment. \pactr is a function that takes a unit and returns a
value of payload type (i.e.: a $\code{coq\_ty T}$ for some type
$\code{T}$). \pactw is a function that takes a parameter of payload
type and returns unit, allowing the process to call OCaml to print on
the screen or write to file or similar things. Finally \pacti is the
action function that passes data to the OCaml runtime and receives
some response, thus combining the two other environment interaction
functions. These functions \emph{do not affect the communication structure}
of the process: they are internal actions and do not appear
in the trace of the process.

\begin{definition}[Process typing system]\label{def:proc-typing}
  We define 
  typing for processes
  $\Gamma \ofLt \expr : \lT$ in Figure \ref{fig:typing}, as an inductive predicate in Coq
  (definition \code{\small of\_lt} in \code{\small Proc.v}). Since $\proc$ is
  embedded in Coq, we assume the standard typing judgement for Gallina
  terms, of the form $\Gamma \vdash \expr : T$. We assume a set of
  sorts $\tS_j$, and an encoding as a Coq type
  $\llbracket \tS_j \rrbracket$ (see Definition \ref{def:global-types}).
\end{definition}
\noindent Rules \rulename{p-ty-end}, \rulename{p-ty-jump}, and
\rulename{p-ty-loop} state that the local type of the ended process, a
jump to $\lX$, and recursion are $\lend$, $\lX$, and a recursive
type respectively. Rule $\rulename{p-ty-send}$ specifies that a
send process with label $\ell$ has a send type, if $\ell$ is in
the set of accepted labels. Rule
$\rulename{p-ty-recv}$ specifies that a receive process 
has a receive type, if all the alternatives
have the correct local type for all possible payloads
$x : \llbracket \tS_i \rrbracket$. Any expression $\expr$ that does
not match any of these rules must be proven to be of the correct
type for all of its possible reductions. For example, it is
straightforward to prove that if $\Gamma \ofLt \dcoq{\expr_t} : \lT$
and $\Gamma \ofLt \dcoq{\expr_f} : \lT$ then
$\Gamma \ofLt \eif{\expr}\; \ethen{\expr_t}\; \eelse{\expr_f} : \lT$
by case analysis on $\expr$. Finally, rules \rulename{p-ty-read},
\rulename{p-ty-write}, and \rulename{p-ty-interact}, have no impact on
the local type, so they simply check that the actions are well typed,
and that the continuation process has the expected type.

\subsection{\dslName}

In the Coq library \code{Zooid.v}, \dslName{} terms (ranged over by
$\zooid$) are dependent pairs of a $\proc$, and a proof that it is
well-typed with respect to a given local type $\lT$.
\vspace{-1mm}
\begin{lstlisting}[language=Coq,  basicstyle=\tt\footnotesize, numberstyle=\scriptsize]
Definition wt_proc L := { P : Proc  | of_lt P L }.
\end{lstlisting}
\vspace{-1mm}
\revised{They are built using smart constructors, helper functions and
notations to define processes that are well-typed by construction
(i.e.: a process and a witness of its type derivation). Moreover, we
take care that the local type of each smart constructor is fully
determined by their inputs, so we can use Coq to \emph{infer} the
local type of each of these processes.} Given a \dslName{} expression
$\zooid$, we can project the \emph{first component} to extract the
underlying $\proc$ term. Since the behaviour of alternatives in
$\zooid$ terms is fully specified, we can infer its local type. By
construction, if a term $\zooid$ can be defined, then its underlying
$\proc$ is well-typed with respect to some local type $\lT$,
\emph{second component} of the dependent pair.

\revised{The simplest example is the \lstinline!finish! term for
  inactive processes of type \lstinline!l_end!. Coq infers most parameters.
 }\vspace{-1mm}
\begin{lstlisting}[language=Coq,  basicstyle=\tt\footnotesize, numberstyle=\scriptsize]
Definition wt_end : wt_proc l_end :=  exist _ _ t_Finish.
Notation finish := wt_end.
\end{lstlisting}
\vspace{-1mm}

\revised{
  On the other hand, the notation \lstinline!\\send! is defined in the
  same way, but the definition of the dependent pair requires a simple
  proof (i.e.: \lstinline!wt_send!). The send command is implemented
  using a singleton choice, and this proof simply says that this label
  is the one in the singleton choice. The definition is as follows:
}
\vspace{-1mm}
\begin{lstlisting}[language=Coq,  basicstyle=\tt\footnotesize, numberstyle=\scriptsize]
Definition wt_send p l T (pl : coq_ty T) L (P : wt_proc L)
  : wt_proc (l_msg l_send p [::(l, (T, L))])
  := exist _ _ (t_Send p pl (of_wt_proc P)
                            (find_cont_sing l T L)).
Notation "\send" := wt_send.
\end{lstlisting}
\vspace{-1mm}

Despite not being directly encoded as a Coq datatype,
Figure~\ref{fig:zooidsyntax} presents the syntax for \dslName{}
terms in BNF notation.
\begin{figure}
\begin{small}
\begin{definition}[\dslName{} syntax]\label{fig:zooid}
  \small 
\ \\[0.5mm]
$\begin{array}{rcl}
    \zooidb
    & \Coloneqq & \zbalt{\ell}{x : \tS} \zooid \\

      \zooids & \Coloneqq & \zcase{\expr}{\ell}{\expr : \tS} \zooid
      \mid \zskip{\ell}{\tS} \lT \\ & &
      \mid \zdflt{\ell}{\expr : \tS} \zooid
    \\[.1cm]
    \zooid
    & \Coloneqq &
      \zjump \lX
      \mid \zloop \lX {\zooid}
      \mid \zif{\expr} \; \zthen{\zooid} \; \zelse{\zooid} \\ & &
      \mid \zsend{\p}{\ell}{\expr : \tS} \; \zooid
      \mid \zrecv{\p}{\ell}{\vx : \tS} \; \zooid
    \\
    & &
      \mid \zend
      \mid \zbranch{\p} \; \zalts{\zooidb_1 \mid \ldots \mid \zooidb_n} \\ & &
      \mid \zselect{\p} \; \zalts{\zooids_1 \mid \ldots \mid \zooids_n} \\ & &
      \mid \pread \pactr \zooid
      \mid \pwrite \pactw {\expr} \zooid \\ & &
      \mid \pinteract \pacti {\expr} \zooid
    \end{array}$
\end{definition}
  \end{small}
\vspace*{-5mm}
\caption{\dslName Syntax}\label{fig:zooidsyntax}
\vspace*{-6mm}
\end{figure}
\noindent 
The syntactic constructs are the expected, with only a few
differences:
\begin{enumerate*}[label=(\alph*)]
\item $\zif{}\zthen{}\zelse{}$ is a \dslName{} construct since it needs to carry the
proof that the underlying proc is well-typed;
\item branch and select must take a list of alternatives ($\zooidb$ and $\zooids$ respectively),
and send/receive are defined as branch/select with a singleton alternative.
\end{enumerate*}
The alternatives for branch, $\zooidb$, are pairs of labels and
continuations. The alternatives for $\dproc{\code{select}}$, $\zooids$ are:\\
\begin{enumerate*}[label=(\arabic*)]
\item $\zcase{\expr_1}{\ell}{\expr_2 : \tS} \zooid$, specifies
to send $\ell$ and $\expr_2 : \llbracket \tS \rrbracket$
and then continue as $\zooid$,
when
$\expr_1$ evaluates to $\dcoq{\code{true}}$;\\
\item
$\zdflt{\ell}{\expr_2 : \tS} \zooid$, specifies that the default alternative
is to send  $\ell$ and $\expr_2$, and then continue as $\zooid$;  and  \iftoggle{full}{\\}{}
\item
$\zskip{\ell}{\tS} \lT$, specifies the unimplemented alternative of sending
$\ell$ and a value of sort $\tS$, and then continuing as $\lT$.
\end{enumerate*}
We require $\dproc{\code{skip}}$ to enforce a unique local type: since
Definition~\ref{def:proc-typing} does not include subtyping,
\dslName{} requires that all the possible behaviours in the local type
must be either implemented or declared. We impose a syntactic
condition on $\dproc{\code{select}}$: there must be exactly one default
case, which must occur after the last $\dproc{\code{case}}$. The three
constructs to interact with external code ($\dproc{\code{read}}$,
$\dproc{\code{write}}$, and $\dproc{\code{interact}}$) are similar to
their untyped counterparts from \sec\ref{sec:core-procs}. 
These actions do not impact the traces nor the local types, so
they simply sport the local type of their continuations.

\vspace{-2mm}
\subsection{Semantics of \dslName{}} \label{sec:zooidsem}
The semantics of \dslName{} is defined as a labelled transition system of the
underlying $\proc$ terms, analogously to that of local type
trees in Definition \ref{lstep}\footnote{For the sake of uniformity, here we present
  the LTS for processes as a relation,
  however in Coq we define it, equivalently,
  as a recursive function: \code{do\_step\_proc} in \code{Proc.v}.
}, but with values instead of sorts in the trace,
and explicitly unfolding recursion.

\begin{definition}[LTS for processes]\label{proc:lts}
  The \emph{LTS for processes}
  is, for each action $a$, defined as:\\[0.5mm]
  \centerline{
    \small
    $\begin{array}{c}
      \begin{array}{lll}
        \rulename{p-step-send} &
        \rulename{p-step-recv} \\
        \dfrac{%
        a=!\p\q(\ell,\expr_1)\quad
        }{%
        \psend{\q}{\ell, \expr_1}{\expr_2}
        \stepa{a}
        \expr_2
        } \qquad
        &
          \dfrac{%
          a=?\q\p(\ell_{\bar{\imath}},\expr)\quad
          }{%
          \precv{\p}{\ell_i. \expr_i}_{i \in I}
          \stepa{a}
          (\expr_{\bar{\imath}} \; \expr)
          }
        &
      \end{array} \\[5mm]
      \begin{array}{l}
        \rulename{p-step-loop} 
          \dfrac{%
          [(\ploop{\lX}{\expr})/(\pjump \lX)]\expr \stepa{a} \expr'
          }{%
          (\ploop{\lX}{\expr})
          \stepa{a}
          \expr'
          }
      \end{array}
    \end{array}
    $}
\end{definition}
\vspace*{-1mm} The steps of the LTS are: \rulename{p-step-send} states that a send
process transitions to the continuation $\expr_2$ with the action that sends a
label $\ell$ and value $\expr_1$; \rulename{p-step-recv} states that a receive
process transitions to $(\expr_{\bar{\imath}} \; \expr)$ with the receive action
from participant $\p$; and \rulename{p-step-loop} unfolds recursion once to
perform a step on a recursive process.

We prove the type preservation for $\ofLt$. To show this, we need to
relate process actions with local/global type actions. This is done by a simple
erasure that removes the values, but preserves the types in an action, denoted
by $|a|$. For example, if $a=!\p\q(\ell,\expr)$ and $\expr : \llbracket \tS
\rrbracket$, then $|a| = !\p\q(\ell, \tS)$.

\begin{theorem}[Type preservation](Theorem \code{\small preservation} in the file \code{\small Proc.v}.)\label{thm:preservation}
If $\Gamma \ofLt \expr : \lT$ and $\expr \stepa{a} \expr'$, then there exists
$\lT'$ such that $\lT \stepa{|a|} \lT'$, and $\Gamma \ofLt \expr' : \lT'$.
\end{theorem}

We write $\plts t \expr$ to express that a trace $t$ is admissible
by process $\expr$. The formal definition goes analogously to Definition
\ref{l-traces} for $\llts \_ \_$; 
note, however, that the admission of a trace by
process is checked \emph{in isolation} to other processes.
To relate process traces to global/local type traces we need to define the
notion of a \emph{complete subtrace}.

\begin{definition}[Complete subtrace]\label{def:completesubtrace}
  We say that $t_1$ is a complete subtrace of $t_2$ for participant $\p$ (definition \code{\small subtrace} in \code{\small Local.v}), if all
actions in $t_2$ that have $\p$ as a subject occur in $t_1$ in the same
relative position (i.e.\ the $n$-th action of $\p$ in $t_2$ must be the $n$-th action
of $t_1$). We write $t_1 \subtrace_\p t_2$ as the greatest relation
satisfying:
\begin{center}
\begin{small}
$
  \newDfrac{
    \subject{a} \neq \p \quad
    t_1 \subtrace_\p t_2
  }{
    t_1 \subtrace_\p (\trnext a t_2)
  }
  \qquad
  \newDfrac{
    \subject{a} = \p \quad
    t_1 \subtrace_\p t_2
  }{
    (\trnext a t_1) \subtrace_\p (\trnext a t_2)
  }
  \qquad
  \newDfrac{
  }{
    \trend \subtrace_\p \trend
  }
  $
\end{small}
\end{center}
\end{definition}

The main result for \dslName{} states that for all admissible traces
for a well-typed process, there exists at least a trace in the larger
system that is a complete supertrace of that of the process. We state
this formally as Theorem \ref{thm:zooid} \\({\small \zooidTheoremCoq{}} in
\code{\small Proc.v}). 
Thus, well-typed processes inherit the global type properties of
protocol compliance, deadlock freedom and liveness.

\begin{theorem}[Process and global type traces]\label{thm:zooid}
Let $\coG \osproj (\dcol{E}, \epsilon)$ and $\Gamma \ofLt \expr : \lT$
such that $\lunroll{\lT}{(\dcol{E} \; \p)}$.  Then, for all traces $t_\p$ such
that $\plts{t_\p}{\expr}$ there exists a trace $t$ such that
$\glts{t}{\coG}$, and $|t_\p| \subtrace_\p t$.
\end{theorem}
Figure~\ref{fig:zooid-theorem} presents the meaning of the above theorem graphically.
\begin{figure}
\begin{center}
  \small
    \tikzset{
      ncbar angle/.initial=90,
      ncbar/.style={
        to path=(\tikztostart)
        -- ($(\tikztostart)!#1!\pgfkeysvalueof{/tikz/ncbar angle}:(\tikztotarget)$)
        -- ($(\tikztotarget)!($(\tikztostart)!#1!\pgfkeysvalueof{/tikz/ncbar angle}:(\tikztotarget)$)!\pgfkeysvalueof{/tikz/ncbar angle}:(\tikztostart)$)
        -- (\tikztotarget)
      },
            ncbar/.default=0.3cm,
    }
    \tikzset{square left brace/.style={ncbar=0.1cm}}

    \begin{tikzpicture}
      \node[text=\colorgt] (a1) {\small $\coG_1$ };
      \node[text=\colorgt  , right = .35cm of a1] (a2) {\small $\coG_2$ };
      \node[text=\colorgt  , right = .35cm of a2] (a3) {\small $\coG_3$ };
      \node[                 right = .35cm of a3] (d1) {\small $\cdots$ };
      \node[text=\colorgt  , right = .35cm of d1] (a5) {\small $\coG_5$ };
      \node[text=\colorgt  , right = .35cm of a5] (a6) {\small $\coG_6$ };
      \node[                 right = .35cm of a6] (d2) {\small $\cdots$ };

      \path[draw, ->] (a1) -- (a2) node[midway, above, text = \colorgt] {\footnotesize $a_1$ };
      \path[draw, ->] (a2) -- (a3) node[midway, above, text = \colorproc] (g2) {\footnotesize $a_2$ };
      \path[draw, ->] (a3) -- (d1) node[midway, above, text = \colorgt] { };
      \path[draw, ->] (d1) -- (a5) node[midway, above, text = \colorgt] {\footnotesize $a_4$ };
      \path[draw, ->] (a5) -- (a6) node[midway, above, text = \colorproc] (g6) {\footnotesize $a_5$ };
      \path[draw, ->] (a6) -- (d2) node[midway, above, text = \colorgt] { };

      \node[text=\colorproc, above = .5cm of a1] (a2') {\small $\proc_1$ };
      \node[text=\colorproc, right = 1.0cm of a2'] (a5') {\small $\proc_2$ };
      \node[text=\colorproc, right = 1.0cm of a5'] (a6') {\small $\proc_3$ };
      \node[                 right = 1.0cm of a6'] (d2') {\small $\cdots$ };
      \path[draw, ->] (a2') -- (a5') node[above, midway, text = \colorproc] (p2) {\footnotesize $a_2$};
      \path[draw, ->] (a5') -- (a6') node[above, midway, text = \colorproc] (p6) {\footnotesize $a_5$};
      \path[draw, ->] (a6') -- (d2');

      \path[draw, thick, dotted, -, red] (g2) -- (p2);
      \path[draw, thick, dotted, -, red] (g6) -- (p6);

      \draw ($(a2')-(.5cm,.2cm)$) to [square left brace] ($(a2')+(-.5cm, .2cm)$)
       node[text=\colorproc, left, yshift=-.2cm, xshift=-.1cm] {$\forall t_\p$} ;

      \draw ($(a1)-(.5cm,.2cm)$) to [square left brace] ($(a1)+(-.5cm, .2cm)$)
       node[text=\colorgt, left, yshift=-.2cm, xshift=-.2cm] (t) {$\exists t$} ;


    \end{tikzpicture}
  \end{center}\vspace*{-3mm}
\caption{Theorem \ref{thm:zooid}, visually.}\vspace*{-2mm}
\label{fig:zooid-theorem}
\end{figure}
Any trace $t_\p = \dproc{\trnext{a_2}{\trnext{a_5}{\ldots}}} $ of a process
$\proc$ is contained within a larger system trace $t =
\dgt{\trnext{a_1}{\trnext{\dproc{a_2}}{\trnext{a_3}{\ldots}}}}$ of $\coG$, given
that $\proc$ behaves as some participant $\p$ in $\coG$.
Namely,
if a process $\expr$ is well typed with a local type $\lT$, which is equal
up to unravelling to that of participant $\p$ in $\coG$, then the
behaviour of $\expr$ is that of $\p$ in $\coG$.

\subsection{Extraction}\label{sec:extraction}
\label{subsec:extraction}
Terms of type $\Proc$, in Coq, can be easily extracted to executable OCaml code,
following an approach similar to that of Interaction Trees \cite{Li-yao:2019}:
we can substitute the occurrences of $\proc$ terms by a suitable OCaml handler.
Figure~\ref{fig:processmonad} shows the declaration of a module for that purpose.
\begin{figure}
\begin{center}
\begin{lstlisting}[language=Coq,  basicstyle=\tt\footnotesize, numberstyle=\scriptsize]
Module ProcessMonad.
  Parameter t : Type -> Type.
  (* monadic bind and pure values *)
  Parameter bind : forall T1 T2, t T1 ->
    (T1 -> t T2) -> t T2.
  Parameter pure : forall T1, T1 -> t T1.
  (* actions to send and receive *)
  Parameter send : forall T, role -> lbl -> T -> t unit.
  Parameter recv : (lbl -> t unit) -> t unit.
  Parameter recv_one : forall T, role -> t T.
  (* actions for setting up a loop and jumping *)
  Parameter loop : forall T1, nat -> t T1 -> t T1.
  Parameter set_current: nat -> t unit.
  (* function to run the monad *)
  Parameter run : forall A, t A -> A.
End ProcessMonad.
\end{lstlisting}
  \end{center}\vspace*{-4mm}
\caption{The Process Monad.}\vspace*{-4mm}
\label{fig:processmonad}
\end{figure}

Module \code{ProcessMonad} specifies a monadic type \code{t}, that supports the
standard \code{bind} and \code{pure} operations, as well as constructs
for adding the required effects, in this case network communication
and looping (with potential non-termination). During extraction this
module becomes the ambient monad for the extracted code. In order to
run the code the user 
instantiates the monad
to provide a low
level implementation, which fills in the details about the network
transport. 
\dslName{}
processes are translated into the monad using the function
\code{\small extract\_proc} from \code{\small Proc.v}.
\iftoggle{full}{Appendix~\appref{appendix:extraction}}{\cite[Appendix~\appref{appendix:extraction}]{fullversion}}
shows the function in its entirety.

\subsection{Runtime}

The code for an endpoint process is extracted as a value inside of the
process monad from \sec\ref{sec:extraction}. \dslName's runtime
provides an implementation of \code{ProcessMonad}. 
The endpoint process is independent of the transport and network
protocols; the exact specification of those is deferred to the
implementation of the monad.
The runtime implements the monad relying on
the monad provided by OCaml's Lwt
library\footnote{https://ocsigen.org/lwt/5.2.0/manual/manual},
as well as its asynchronous communication primitives.
The transport uses TCP/IP and the payloads are encoded and decoded
using the 'Marshal' module in OCaml's standard
library\footnote{\url{https://ocaml.org/releases/4.11/htmlman/libref/Marshal.html}}.
This design prioritises OCaml based technologies to implement
asynchronous I/O and data encoding. Other transports are possible 
(e.g., web services over HTTP).

\subsubsection{Implementation}

In \dslName, the user implements their processes in the DSL, then uses 
Coq to produce OCaml code for the monad's module type and 
for the process, using extraction. The runtime implements a means to run that
code. Concretely it provides the transport and serialization.

A runnable process amounts to an instance of the functor type in
Figure~\ref{fig:processfunctor}, in which we provide the process monad
instance together with the extracted process.
\begin{figure}
\begin{center}
\begin{lstlisting}[language=ocaml, basicstyle=\tt\footnotesize, numberstyle=\scriptsize, flexiblecolumns=true]
module type PROCESS_FUNCTOR =
  functor (MP : ProcessMonad) -> sig
    module PM : sig
      type 'x t = 'x MP.t
      val run : 'a1 t -> 'a1
      val send : role -> lbl -> 'a1 -> unit t
      val recv :
        role -> (lbl -> unit t) -> unit t
      val recv_one : role -> 'a1 t
      val bind : 'a1 t -> ('a1 -> 'a2 t) -> 'a2 t
      val pure : 'a1 -> 'a1 t
      val loop : var -> (unit -> 'a1 t) -> 'a1 t
      val set_current : var -> unit t
    end
    val proc : unit MP.t
  end
\end{lstlisting}
  \end{center}\vspace*{-6mm}
\caption{The Process Functor.}\vspace*{-6mm}
\label{fig:processfunctor}
\end{figure}

Communication primitives in processes are unaware of transport or
other networking issues, they simply expect to be able to communicate
with the other roles involved in the protocol. The runtime
implementation requires the user to provide for each role a list of
channels to communicate with the other roles. It is specified as:
\vspace{-1mm}
\begin{lstlisting}[language=ocaml, basicstyle=\tt\footnotesize, flexiblecolumns=true]
type connection_spec
  = Server of sockaddr | Client of sockaddr
type conn_desc =
  { role_to : role; spec : connection_spec }
\end{lstlisting}
\vspace{-1mm}
where each process needs to specify a \lstinline!conn_desc list!
detailing a channel to each role where it either starts a connection
(using the \lstinline!Client! connector and specifiying IP and port in
the \lstinline!sockaddr! datatype) or waits for a connection (in a
similar way using the \lstinline!Server! constructor).




So finally, the runtime is invoked by calling the function:
\vspace{-1mm}
\begin{lstlisting}[language=ocaml, basicstyle=\tt\footnotesize, flexiblecolumns=true]
val execute_extracted_process
  : conn_desc list -> (module PROCESS_FUNCTOR) -> unit
\end{lstlisting}
\vspace{-1mm}
\noindent which connects a participant to all the roles as specified in the connection
list and executes extracted process passed as first-class module value
to the function. If the extracted process interacts with OCaml code, the library that
implements all the external functions has to be compiled into the
executable. 

With the addition of the runtime \dslName processes become certified
code that can be readily executed to implement distributed multiparty
services. 


\section{Evaluation: Certified Processes} \label{sec:zooid}
This section displays several common use cases in the MPST literature,
implemented and certified using \dslName: (1) several implementations of a
recursive ping-pong protocol; (2) a recursive pipeline; and (3) the two-buyer
protocol from \cite{Honda2008Multiparty}. \revised{We conclude the section with
a summary evaluating our mechanisation effort.}

\myparagraph{A Common Workflow}
\label{sec:zooid-workflow}
Our workflow consists of the following
steps:
\begin{enumerate*}[label=(\arabic*)]
  \item\label{lst:spec} specify the global type for the protocol;
  \item\label{lst:proj} project the global type into the set of local types;
  \item\label{lst:impl} implement a process using \dslName;
  \item\label{lst:prove} (if necessary) prove that the local type of the process
is equal up to unravelling to the projection of some participant;
  \item\label{lst:extr} use extraction to OCaml; and
  \item\label{lst:ocaml} implement external OCaml actions (if any).
\end{enumerate*}

Steps \ref{lst:spec}, \ref{lst:impl}, and \ref{lst:ocaml} are the
necessary inputs for implementing a certified process. Steps
\ref{lst:proj} and \ref{lst:extr} are fully automated, and step
\ref{lst:prove} is often automated too, although it may require a
simple manual proof. Finally, while step
\ref{lst:extr} is fully automated, it is possible to control the
result by using common Coq commands (e.g.~marking some definitions
opaque to avoid inlining them).
\subsection{Examples of Certified Processes}
\noindent\textbf{Pipeline. }
We start with a recursive variant of the example in
\sec~\ref{subsec:first-zooid}. The first step is to specify the global type.
We write its inductive representation:

\vskip.1cm
\noindent
$
\small
\begin{array}{ll}
\coqDef \; \Rpipe :=  \grec{\gX}{} & \hspace{-2mm}
  \msg{\Alice}{\Bob} \ell(\tnat) .\\
  & \hspace{-2mm} \msg{\Bob}{\Carol} \ell(\tnat) .
  \gX.
\end{array}
$
\vskip.1cm

\noindent
The next step is to project $\Rpipe$ into all of its participants.
There are two reasons to apply the projection at this step: (1) only
well-formed protocols are projectable; and (2) we obtain the local
types that will guide the implementation: the local types
will need to typecheck the implemented processes. If the global type
is not projectable, or the processes do not implement the resulting
local types (or one of their unrollings), then we cannot guarantee
anything about a \dslName{} implementation of any participant. We
define a notation for performing the projection of all participants:

\vskip.1cm
\noindent
$\small
\coqDef \; \RpipeLT := \coqProj \; \Rpipe.
$
\vskip.1cm

\noindent
If $\iftoggle{full}{\ }{}\Rpipe\iftoggle{full}{\ }{}$ is not well-formed, then $\iftoggle{full}{\ }{}\coqProj \iftoggle{full}{\ }{}$ will not \iftoggle{full}{\\}{}
typecheck. Otherwise, $\RpipeLT$ will be a list of pairs of
participants and local types. This list will contain an entry for
$\Alice$, $\Bob$ and $\Carol$. We get local type for $\Bob$ with:
\vskip.1cm
\noindent
$\small
\coqDef \; \BobLT := \coqGet \; \textsf{Bob} \; \RpipeLT.
$
\vskip.1cm

\noindent
The notation $\coqGet$ expands into a lookup in $\RpipeLT$ that
requires a proof that $\Bob$ is in $\RpipeLT$. If we write
$\coqGet \; \p \; \allowbreak\RpipeLT$ with some $\p \not\in\RpipeLT$, then the
command will fail to typecheck. There are now two possibilities for
using $\BobLT$ to implement $\Bob$: (1) providing $\BobLT$ as a type
index; or (2) omitting $\BobLT$, inferring the local type, and then
proving that the inferred local type is equal to $\BobLT$ up to
unravelling. Here we use (1), but sometimes the process actually
implements an unrolling of the local type. We will show examples of
(2) in the next section.

\vskip.1cm
\noindent
$\small
\begin{array}{l}
  \coqDef \; \BobProc : \zooidTy{\BobLT}
  \\ \quad
  :=
  \dproc{\code{loop}} \; {\lX} \; %
  \begin{array}[t]{@{}l@{}}
    ( \zrecv{\Alice}{\ell}{x : \tnat} \\
    \zinteract {\withcolor{dkgreen}{\code{compute}}} x \ %
    (\zfun {res} \\ \ \ %
    {\zsend{\Carol}{\ell}{\dcoq{res} : \tnat} \zjump \lX}) ).
  \end{array}
\end{array}
$
\vskip.1cm

With \dslName's
\withcolor{\colorproc}{\code{interact}}\ command we can call
the \withcolor{dkgreen}{\code{compute}}\ function, which is implemented in
OCaml, allowing any arbitrary computation safely because the runtime
hides the communication channels to prevent errors.

Finally, to do extraction to OCaml, we call $\runProc : \Proc \to
\code{MP\!.\!t}$. The user has options for code extraction: (1)
since $\Proc$ is defined inductively, use Coq's $\coqEval \;
\code{compute}$ to first replace all occurrences of $\Proc$ to
$\code{MP\!.\!t}$; (2) extract the inductive representation, as well as
$\runProc$. The former may evaluate and unfold more terms than desired. To
control this, we use Coq's command $\coqOpaque$ to specify any function or
definition that we do not wish to be unfolded.\\[1mm]
\noindent\textbf{Ping-Pong. }
In the anonymous supplement, we present
several implementations of the clients of a ping-pong server.
The global protocol is:

\vskip.1cm
\noindent
$\small
\begin{array}{l}
\coqDef \; \Rpingpong := \grec{\gX}{}
  \msg{\Alice}{\Bob}\ \{\\\quad\quad\quad \ell_1(\tunit). \; \gend; \;\ \ell_2(\tnat) .
  \msg{\Bob}{\Alice} \ell_3(\tnat) . \gX\}.
\end{array}
$
\vskip.1cm

\noindent
Here, $\Alice$ acts as the client for $\Bob$, which is the ping-pong server.
$\Alice$ can send zero or more \emph{ping} messages (label $\ell_2$), and
finally quitting (label $\ell_1$). $\Bob$, for each ping received, will reply a
\emph{pong} message (label $\ell_3$). In particular, we
wish to implement a client that sends an undefined number of pings,
stopping when the server replies with a natural number greater than
some $k$. We show below the \dslName{} specification:

\vskip.1cm
\noindent
$\small
\hspace{1mm}\begin{array}{@{}l@{}}
  \coqDef \; \AliceProc: \azooid
  := \dproc{[\code{proc}} \\
  \begin{array}[t]{@{}l@{}}
    \zselect{\Bob}{
    \hspace{-3mm}\begin{array}[t]{@{}l@{}}
      \hspace{3mm}[\zskip{\ell_1}{\tunit} \lend
      \\ \mid
      \zdflt{\ell_2}{0 : \tnat}\\\ \ %
      \dproc{\code{loop}} \; \lX \; (
      \zrecv{\Bob}{\ell_3}{x : \tnat}
      \\
      \zselect{\Bob}{%
      \begin{array}[t]{@{}l@{}}
        [ \zcase{x \geq k}{\ell_1}{\code{tt} : \tunit}
          \zend
        \\ \mid
        \zdflt{\ell_2}{x : \tnat}
        \zjump{\lX}
        ])
      ]]
      \end{array}
      }
    \end{array}
    }
  \end{array}
\end{array}
$
\vskip.1cm

\noindent
We project $\Rpingpong$ and get the expected
local type for $\Alice$: $\AliceLT$.
We observe that here the local type for
$\AliceProc$ is not syntactically equal to $\AliceLT$:

\vskip.1cm

\noindent
\hspace{-1.5mm}$\small
\begin{array}{l}
\AliceLT =
\dlt{\mu \lX}. \lsnd{\Bob} \{
   \ell_1(\tunit). \lend; 
   \ell_2(\tnat). \lrcv{\Bob} \ell_3(\tnat). \lX\}

  \\[0.5em]

\code{projT1} \; \AliceProc = \lsnd{\Bob} \{\ %
   \ell_1(\tunit). \lend;
   \; \ell_2(\tnat). \dlt{\mu \lX}.\\\quad \quad \quad \lrcv{\Bob} \ell_3(\tnat).
   \lsnd{\Bob} \{
   \ell_1(\tunit). \lend; \;
   \ell_2(\tnat). \lX.
    \} \}

\end{array}
$\\[1mm]
\noindent
This is not a problem since a simple proof by coinduction can show
that both types unravel to the same local tree. This gains 
the flexibility to have processes that implement any unrolling of
their local type, and the proofs are mostly simple as they follow the
way the types were unrolled. See
\iftoggle{full}{Appendix \appref{app:pingpong}}{\cite[Appendix \appref{app:pingpong}]{fullversion}}
for
more details on how to
construct gradually this client, showing how to iteratively
program using \dslName.

\subsection{A Certified Two Buyer Protocol}
We conclude this section presenting an implementation of the
\emph{two-buyer protocol} \cite{Honda2008Multiparty}, a common
benchmark of MPST. This is a protocol for an online purchase service
that enables customers to split the cost of an item among two
participants, as long as they agree on their shares. First, buyer
$\BuyerA$ queries the seller $\Seller$ for an item. Then,
$\Seller$ sends the item cost first to $\BuyerA$, then to
$\BuyerB$. Then,
$\BuyerA$ sends a proposed share for the item.
$\BuyerB$ then either accepts the proposal, and receives the delivery
date from $\Seller$, or rejects the proposal.

\begin{figure}
  \centering
  \begin{displaymath}\footnotesize
  \begin{array}{l}
    \begin{array}{@{}l@{}}
      \coqDef \; \TwoBuyer := \ %
      \msg{\BuyerA}{\Seller} \ItemId(\tnat).\\ \quad
      \msg{\Seller}{\BuyerA} \Quote(\tnat).
      \msg{\Seller}{\BuyerB} \Quote(\tnat).\\ \quad

      \msg{\BuyerA}{\BuyerB} \Propose(\tnat).
      \msg{\BuyerB}{\Seller}
      \{\ %
      \Accept(\tnat). \\ \quad \msg{\Seller}{\BuyerB}\Date(\tnat). \gend; \ %
      \Reject(\tunit). \; \gend
      \}
    \end{array}
    \\\multicolumn{1}{c}{\rule{.4\columnwidth}{0.4pt}}\\

    \begin{array}{@{}l@{}}
      \BuyerBLT := \lrcv{\Seller} \Quote(\tnat).
      \lrcv{\BuyerA} \Propose(\tnat).
      \lsnd{\Seller}
      \{
      \Accept(\tnat).\\  \lrcv{\Seller} \Date(\tnat). \lend;
      \Reject(\tunit).  \lend
      \}
    \end{array}
    \\\multicolumn{1}{c}{\rule{.4\columnwidth}{0.4pt}}\\

    \begin{array}{@{}l@{}}
      \coqDef \; \ProcB : \zooidTy{\BuyerBLT} :=\\
      \zrecv{\Seller}{\Quote}{x : \tnat}
      \zrecv{\BuyerA}{\Propose}{y : \tnat}
      \\ \quad
      \zselect{\Seller}{%
      \begin{array}[t]{@{}l@{}}
        [\zcase{y >= \code{divn} \; x \; 3}{\Accept}{y - x : \tnat}\\ \ \ %
        \zrecv{\Seller}{\Date}{d : \tnat}.
        \zend
        \\
        \mid
        \zdflt{\Reject}{\code{tt} : \tunit} \zend
        ]
      \end{array}
      }
    \end{array}
  \end{array}
\end{displaymath}
\vspace*{-2mm}
\caption{The Two Buyer protocol}\vspace*{-.6cm}
\label{fig:twobuyer}
\end{figure}

Figure~\ref{fig:twobuyer} shows the protocol as a global type,
the local type, $\BuyerBLT$, that results from the projection on $\BuyerB$,
and a possible implementation of the role of $\BuyerB$ in \dslName{}.
Different implementations of the local type
will differ in how the choice is made, but 
the local type
will always need to be syntactically equal to the projected $\BuyerBLT$,
due to the absence of recursion. In the implementation chosen in
Figure~\ref{fig:twobuyer}, the participant $\BuyerB$
will reject any proposal where $\BuyerB$ pays
more than one third of the cost of the item.
This implementation is \emph{guaranteed} to behave as $\BuyerB$ in
the protocol $\TwoBuyer$, hence deadlock-free.  Our workflow
preserves the ability to define and implement each participant independently:
$\BuyerA$ and $\Seller$ could be implemented in any language, as long as they
are implemented using a compatible transport to that of the OCaml implementation
of $\code{MP\!.\!t}$. The code that checks the types and performs the
projections is certified, as it is exactly the same code about which the
properties were established.

\subsection{Mechanisation Effort}
\revised{The development is 7.3KLOC of Coq code, and 1.7KLOC of OCaml
  for the runtime (including examples). The certified code consists of
  269 definitions, including functions and (co)inductive definitions
  and 396 proved lemmas and theorems. The most challenging part was
  working out the right definitions: the finite syntax object/infinite
  unrolling correspondence felt like a convoluted approach at first,
  but it greatly accelerated our progress afterwards.}


\section{Related Work and Conclusion} \label{sec:related}\label{sec:conclusion}

In the concurrency and behavioural types communities, there is growing
interest in mechanisation and the use of proof assistants to validate
research. As a recent example,
\citet{Krebbers:2020} explore the notion of semantic typing using a
concurrent separation logic as a semantic domain to build on top a
language to describe binary session types. On the same vein,
SteelCore~\cite{Swamy:2020} allows DSLs to take advantage of solid the
semantic foundations provided by a proof assistant. Where their
works use separation logic as a foundation, \dslName uses MPST and
their coinductive expansion.

The ambition of mechanisation in behavioural types is increasing and
collaborative projects that explore the space of available solutions
are an important tool for the community, where they explore different
representations of binders (names, de~Bruijn indices/levels, nominals
respectively), see \citeN[Discussion]{VEST:2020}. In this work we
sidetrack the question by designing \dslName to use a shallow
embedding of its binders (thus avoiding to need an explicit
representation for variables). In our experience, this is a simple and
valuable technique for the situations where it is applicable.

Other works also explore ideas on \emph{binary} session types using
proof assistants and mechanised proofs. For example,
\revised{\citet{brady:2017} develops a methodology to describe safely
  communicating programs and implements DSLs, embedded in Idris,
  relying on the Idris type checker}. \citet{Thiemann:2019} develops
an intrinsically typed semantics in Agda that provides preservation
and a notion of progress for binary session types. \Citet{GTV20}
explore the interaction between duality and recursive types and how
they take advantage of mechanisation to formalise some of their
results. \citet{Tassarotti:2017} show the correctness (in the Coq
proof assistant) of a compiler that uses an intermediate language
based on a simplified version of the GV system~\cite{Gay:2010} to add
session types to a functional programming language. And
\citet{Orchard:2016} discuss the relation between session types and
effect systems, and implement their code in the Agda proof assistant.
Their formalisation concentrates on translating between effect systems
and session types in a type preserving manner. \citet{emtst:2019}
present a type preservation of binary session types~\cite{Honda:1998,
  Yoshida:2007} as a case study of using their
tool~\cite{Castro:2020}. Furthermore, \citet{Goto:2016} present a
session types system with session polymorphism and use Coq to prove
type soundness of their system. Note that none of the above works on
session types treats \emph{multiparty session types} -- they are
limited to \emph{binary} session types.

Our work on MPST uses mechanisation to both give a fresh look at trace
equivalence~\cite{DenielouYoshida2013} in MPST and to further explore
its relation to a process calculus. At the same time our aim is to
provide a bedrock for future projects dealing with the MPST theories.
And crucially, this is the first work that tackles \emph{a full syntax
  of asynchronous multiparty session types} that type the whole
interaction, as opposed to binary session types, which only type
individual channels.

\revised{ Furthermore, in this work, we present not only \dslName{} as
  a certified process language, but also the methodology to design a
  certified language like this. Zooid’s design starts with the theory,
  then the mechanised metatheory, and, finally, 
  implementing a deeply embedded process language (deeply embedded in
  two ways: as a DSL and in the library of definitions and lemmas
  provided in the proof mechanisation). We propose \dslName{} as an
  alternative to writing an implementation that is proved correct post
  facto. There is no tension between proofs and implementation, since
  the proofs enable the
  implementation. 
  An important feature of our design is the \emph{correspondence of
    syntactic objects and their infinite tree representation}.
  Coinductive trees allow us to deal smoothly with semantics and avoid
  bindings: such a technique applies to languages with equi-recursion,
  a widespread construct \cite
  {GHILEZAN2019127,PierceBook,SeveriDezani:2019}. On the other hand we
  have kept an inductive type system for processes, so that we have
  finite, easy-to-inspect, structures, on which we can make
  computations. Our novel design takes advantage of the infinite-tree
  representation of syntactic objects, thus providing us with
  syntactic types for Zooid and coinductive representation for the
  proofs.}

Regarding the choice of tool and inspiration in this work, we point
out that the first objective is to mechanise trace equivalence between
global and local types. 
For that, we
took inspiration from more semantic representations of session
types~\cite{GHILEZAN2019127, Li-yao:2019}. The choice of the Coq proof
assistant~\cite{CoqManual} was motivated by its stability, rich
support for coinduction, and good support for the extraction of
certified code. Stability is important since this is a codebase that
we expect to work on and expand for future projects. The proofs take
advantage of small scale reflection~\cite{Gonthier:2010} using
Ssreflect to structure our development. And given the pervasive need
for greatest fixed points in MPST, we extensively use the PaCo
library~\cite{paco} for the proofs that depend on coinduction.

To conclude, we design and implement a certified language for
concurrent processes supporting MPST. We start by mechanising the
meta-theory of asynchronous MPST, and prove the soundness and
completness theorems of trace semantics of global and local types. We
then build \dslName, a process language on top of that. Using code
extraction, we interface with OCaml code to produce running
implementations of the processes specified in \dslName.

This work on mechanising MPST and \dslName is a founding stone, there
are many exciting opportunities for future work. On top of our
framework, we plan to explore new ideas and extensions of the theory
of session types. The immediate next step is to make the proofs
extensible, for example by allowing easy integration of custom merge
strategies, adding advanced features such as indexed dependent session
types~\cite{Castro:2019:DPU:3302515.3290342}, timed
specifications~\cite{BYY2014,BLY2015}, \revised{or session/channel
  delegation~\cite{Honda2008Multiparty}.} Moreover, we intend to apply
the work in this paper (and its extensions) to implement a certified
toolchain for the Scribble protocol description language
(\iftoggle{full}{available at }{}{\url{http://www.scribble.org}), also known as ``the practical
  incarnation of multiparty session types''
  \cite{10.1007/978-3-642-19056-8_4,FeatherweightScribble}. To this
  aim we plan to translate from Scribble to MPST style global types,
  following the Featherweight Scribble
  formalisation~\cite{FeatherweightScribble}.


\begin{acks}                            

  We thank the PLDI reviewers for their careful reviews and
  suggestions. We thank Fangyi Zhou for their comments and testing the
  artifact. The work is supported by
  \grantsponsor{}{EPSRC}{https://epsrc.ukri.org/}, grants
  \grantnum{EPSRC}{EP/T006544/1}, \grantnum{EPSRC}{EP/K011715/1},
  \grantnum{EPSRC}{EP/K034413/1}, \grantnum{EPSRC}{EP/L00058X/1},
  \grantnum{EPSRC}{EP/N027833/1}, \grantnum{EPSRC}{EP/N028201/1},
  \grantnum{EPSRC}{EP/T014709/1}, and \grantnum{EPSRC}{EP/V000462/1}
  and by \grantsponsor{}{NCSS/EPSRC}{} \grantnum{NCSS/EPSRC}{VeTSS}.

\end{acks}

\bibliography{fullversion}

\appendix
\onecolumn
\theoremstyle{acmplain}
\newtheorem*{theorem*}{Theorem}

\section{Multiparty Session Types in Coq}
\label{appendix:meta-proofs}

One of the contributions of our work is a mechanisation of multiparty session types;
this section is dedicated to their metatheory,
as we have formalised it in Coq: we follow the structure of the Coq
development, and indicate with precise pointers where definitions and
results can be found in the Coq formalisation associated to this
paper. We will present proof outlines for some of the main results, while the full proofs are to be found in the Coq formalisation associated with this paper; such outlines are meant to guide the interested reader through our formalisation.

This appendix is structured as a formalised journey towards semantics.
We first present the classic syntax of global and local session types.
We then introduce coinductive global and local trees as a ``more
semantic version'' of types (Sections \ref{subsec:global-A} and
\ref{subsec:loc-A}). We formalise a precise relation to associate trees
to types, and then we show that it preserves projections (Section
\ref{subsec:projection-A}). This closes the square
\textcolor{orange}{(M.1)} of Figure \ref{fig:dia}.
We introduce buffers to deal with asynchronicity (Section
\ref{subsec:async_projection-A}), we then define small-step semantics
for global and local types, via labelled transition systems on trees
(Section \ref{subsec:step-A}). Finally, in Section \ref{subsec:trace-eq-A}
we present our main result: execution trace equivalence for global and
local types, thus closing the square \textcolor{orange}{(M.2)} of
Figure \ref{fig:dia}.

\subsection{Global Types}
\label{subsec:global-A}
This subsection gives definitions for global types, following our Coq
development. The literature offers a wide variety of presentations of
global types \cite{Honda2008Multiparty,HYC2016,Scalas:2019,CDPY2015},
each exploring different aspects of communication. Building on
\cite{DenielouYoshida2013}, we formalise \emph{asynchronous multiparty
  session types} (MPST), which allow us to capture the essential
behaviour of asynchronous message exchange, where messages are
transmitted via FIFO queues, and treat the key features of MPST
including selection, branching and recursion.

We first discuss the formalisation of \emph{global types}, covered in
the Coq files of the folder \code{Global}. We use \emph{sorts} to refer to the
types of supported message payloads, covered in the Coq file: \code{Common/AtomSets.v}.
\begin{definition}[Sorts and global types]
  \label{def:global-types-A}%
        {\em Sorts} (datatype \mty{} in \code{Common/AtomSets.v}), ranged over by $\tS$,
and  {\em global types} (datatype \gty{} in \code{Global/Syntax.v}), ranged over by $\G$, are generated by:
\[
  \begin{array}{rcl}
    \tS & ::=& \tnat \SEP \tint \SEP\tbool \SEP \tplus \tS \tS \SEP \tpair \tS \tS \SEP \tseq\ \tS\\[1mm]
    \G & ::=& %
              \gend \;\SEP\;%
              \gX \;\SEP\;%
              \grec \gX \G \;\SEP\;
              \msgi \p\q \ell {\tS} {\G}%
  \end{array}
\]
We require that $\p\neq \q$, $I\neq \emptyset,$ and
$\ell_i \neq \ell_j$ whenever $i \neq j,$ for all $i,j\in I$.
\end{definition}
In Definition~\ref{def:global-types-A}, the type
$\msgi \p\q \ell {\tS} {\G}$ describes a protocol where %
participant $\p$ must send to $\q$ one message %
with label $\ell_i$ and a value of sort $\dte{\tS_i}$ as payload, for
some $i \in I$; %
then, depending on which $\ell_i$ was sent by $\p$, %
the protocol continues as $\G_i$. Sorts can be basic types such as
natural numbers ($\tnat$), integers ($\tint$), booleans ($\tbool$) or
recursive combinations of these, as sums ($\dte{+}$), pairs
($\dte{*}$) or lists ($\tseq$). %
The type $\gend$ represents a terminated protocol. %
Recursive protocol is modelled as $\dgt{\gmu\gX.\G}$, %
where recursion variable $\gX$ is bound.

The representation of the syntax above as inductive types is standard.
In Coq we represent the recursion binder using \dbj
indices~\cite{debruijn:1972, Gordon:1993,
  McKinna:1999,McBride:2004}.
To ease the presentation throughout the paper,
we keep using explicit names for variables.

As customary in the literature, we are interested in global types such
that (1) bound variables are \emph{guarded}---e.g.,

\[\gmu\gX.\msg \p\q \dgt{\ell (\tnat). \G}\]

is a valid global type,
whereas $\gmu\gX\dgt{.}\gX$ is not---and (2) types are \emph{closed},
i.e., all recursion variables are bound by $\gmu\gX$. We will mostly
leave these conditions implicit, however we provide below the two
formal definitions.

\begin{definition}[Guardedness for global types]
  \label{def:guarded-A}%
  We say that a global type is \emph{guarded} (\code{guarded} in
  \code{Global/Syntax.v}) according to the following definition:
  \[
    \small
    \begin{array}{c}
      \dfrac{}{\guarded\ \gend}\quad
      \dfrac{}{\guarded\ \gX}\quad
      \dfrac{\code{not\_pure\_rec}\;\gX\;\G \quad \guarded\ \G}{\guarded\ (\grec \gX \G)}\\[3mm]
      \dfrac{\forall i\in I. \guarded\ {\dgt{\G_i}}}{\guarded\ (\msgi \p\q \ell {\tS} {\G})}
    \end{array}
  \]
  where $\code{not\_pure\_rec}\;\gX\;\G$ means that $\G$ is different
  from $\grec {\dgt{Y_1}} {\dots\grec {\dgt{Y_n}} \gX}$ and also
  $\G\neq\gX$.
\end{definition}

\begin{definition}[Free variables and closure (global types)]
\label{def:closed-A}%
The set of \emph{free variables} of a global type $\G$, $\gfv(\G)$,
(\code{g\_fidx} in \code{Global/Syntax.v}) is defined as follows.
\[
\small
\begin{array}{c}
\gfv(\gend) = \emptyset, \
\gfv(\gX) = \{\gX\}, \
\gfv(\grec \gX \G) = (\gfv(\G))\backslash\{\gX\}, \\
\gfv(\msgi \p\q \ell {\tS} {\G}) = \bigcup\limits_{i\in I} \gfv({\dgt{\G_i}})
\end{array}
\]
We say that $\G$ is \emph{closed} (\code{g\_closed} in
\code{Global/Syntax.v}) if it does not contain free variables, namely:
$\closed\ \G \leftrightarrow (\gfv(\G)=\emptyset)$.
\end{definition}
We define the \emph{set of participants of a global type $\G$} (\code{participants} in \code{Global/Syntax.v}), %
by structural induction on $\G$, as follows:
\[
\small
\begin{array}{c}
  \parti\ \gend = \parti\ \gX = \emptyset \quad
  \parti\ \dgt{\mu \gX.\G} = \parti\ {\G} \\
  \parti\ {\msgi \p \q \ell \tS {\G}}=\{\p,\q\}\cup \bigcup\limits_{i\in I} \parti\ {\dgt{G_i}} %
\end{array}
\]
Participants of a global type are those roles that are involved in the
communication.

In session types, it is common practice to adopt the
\emph{equi-recursive viewpoint} \cite{PierceBook}, i.e., to identify
$\dgt{\gmu\gX.\G}$ and $\dgt{\G \{ \gmu\gX.\G / \gX \}}$, given that
their intended behaviour is the same. Such unravelling of the
recursion constructor can be performed infinitely many times, thus
obtaining possibly infinite trees, whose structure derives from the
above syntax of global types \cite{GHILEZAN2019127}.
%
%

In the formalisation we provide a \emph{coinductively defined
  datatype} (codatatype) of \emph{finitely branching trees} with
\emph{possible infinite depth}. Their branching mirrors the branching
of global types into their continuations and the infinite depth allows
us to indefinitely unravel recursion.

We have a similar coinductive representation for local types
(\sec\ref{subsec:loc-A}). For the definitions of, and proofs about,
coinductive objects in Coq, we have taken advantage of the Paco
library \cite{paco} for parametrised coinduction, which allows for a
more compositional reasoning in the formalisation than the standard
\code{cofix} construction.

\begin{definition}[Semantic global trees]
  \label{def:global trees-A}%
  {\em Semantic global trees} (datatypes \rgty\ and \igty\ in
  \code{Global/Tree.v}, see Remark \ref{remark:prefixes-A}), %
  ranged over by $\coG$, %
  are terms generated \emph{coinductively}
  by: 
  \[
    \begin{array}{rll}
    \coG & ::= & \\
    & & \cogend \;\SEP\;
    \comsgni \p\q \ell {\tS} {\coG}\\
      && \;\SEP
         \comsgsi  \p\q {\ell_j} \ell {\tS} {\coG}
    \end{array}
  \]
  We require that $\p\neq \q$, $I\neq \emptyset,$ and
  $\ell_i \neq \ell_j$ whenever $i \neq j,$ for all $i,j\in I$. %
\end{definition}
The above codatatype represents the bridge between the syntax and the
semantics for global types.
Here 
we make explicit the two asynchronous stages of the communication of a
message:
\begin{itemize}
\item $\comsgni \p\q \ell {\tS} {\coG}$ represents the status where a
  message from the participant $\p$ to the participant $\q$ has yet to
  be sent;
\item $\comsgsi \p\q {\ell_j} \ell {\tS} {\coG}$ represents the status
  immediately after the above: the label $\ell_j$ has been picked
  among the $\ell_i$, $\p$ has sent the message, with payload $\tS_j$,
  but $\q$ has not received it yet.
\end{itemize}

We can now give the central definition of coinductive unravelling for
global types, which relates a global type with its semantic tree.
\begin{definition}[Global unravelling]
  \label{def:gunroll-A}%
  {\em Unravelling of global types} (definition \code{GUnroll} in
  \code{Global/Unravel.v}) is the relation between global types and semantic
  global trees coinductively defined by: 
  \[
    \begin{array}{ccc}
    \rulename{g-unr-end} &\rulename{g-unr-rec}\\
    \newDfrac{ }{\gunroll \gend \cogend} &
    \newDfrac{\gunroll {\dgt{\G \{ \gmu\gX.\G / \gX \}}} \coG}{\gunroll{\grec \gX \G} \coG} &
  \end{array}
  \]
  \[
\begin{array}{l}
    \rulename{g-unr-msg} \\
    \newDfrac{\forall i\in I. \gunroll {\dgt{\G_i}} {\dcog{\coG_i}} }{\gunroll{\msgi \p\q \ell {\tS} {\G}}{\comsgni \p\q \ell {\tS} {\coG}} }
  \end{array}
  \]
\end{definition}
The unravelling operation gives us a snapshot of all possible
executions of the global type. As infinite trees, both
$\dgt{\gmu\gX.\G}$ and $\dgt{\G \{ \gmu\gX.\G / \gX \}}$ have the same
representation: we are able to identify global types that are the same
up to unfolding, thus offering a rigorous behavioural
characterisation.
Moreover, we have obtained a
binding-free syntax for global types and therefore removed one of the
most notoriously tedious features of formal reasoning.

\begin{remark} \label{remark:prefixes-A} In Coq the coinductive datatype
  for global trees is defined in a slightly different way, in
  particular we formally split Definition \ref{def:global trees-A} in
  two parts. First, we define its coinductive core as follows
  (datatype \rgty\ in \code{Global.v}):
\begin{definition}[Semantic global trees, alternative definition]
  \label{def:g-trees-alt-A}%
  {\em Semantic global trees} (datatype \rgty\ in \code{Global/Trees.v}), ranged over by $\coG$, %
  are terms generated \emph{coinductively} by the following grammar:
  \[
  \coG \quad ::=\quad %
  \cogend \;\SEP\;%
  \comsgni \p\q \ell {\tS} {\coG}
  \]
  We require that $\p\neq \q$, $I\neq \emptyset,$ and
  $\ell_i \neq \ell_j$ whenever $i \neq j,$ for all $i,j\in I$. %
\end{definition}
Note that at this point we have \emph{not} introduced the asynchronous
bit of separating ``send'' and ``receive'' messages yet. We introduce
the ``receive'' constructor $\comsgs \_ \_ \_$ with the next
\emph{inductive} datatype, defined on top of trees. In particular,
this receive constructor is only present in the prefixes because at
any given time only finitely many messages have been sent.
\begin{definition}[Prefixes for global trees]
  \label{def:prefixes-A}%
  {\em Prefixes for global trees} (datatype \igty\ in \code{Global/Tree.v}), ranged over by $\preG$, %
  are terms generated \emph{inductively} by the following grammar:
  \[
    \begin{array}{rll}
    \preG & ::= & \\
    & & \pgend\ \coG \;\SEP\;
    \pmsgni \p\q \ell {\tS} {\preG}\;\SEP\;\\
      && \;\SEP
         \pmsgsi  \p\q {\ell_j} \ell {\tS} {\preG}
    \end{array}
  \]
  We require that $\p\neq \q$, $I\neq \emptyset,$ and
  $\ell_i \neq \ell_j$ whenever $i \neq j,$ for all $i,j\in I$. %
\end{definition}
The intuition for each construct in the last two definitions follow
exactly the one for Definition \ref{def:global trees-A}; we inject the
codatatype \rgty\ into the datatype \igty\ using a dedicated
constructor $\pgend$ and, as anticipated, the receive-message
$\pmsgs \p \q \ell$ is now part only of the syntax for inductive
prefixes.

In Coq, distinguishing the two (co)datatypes and injecting one into
the other, has allowed us to perform ``induction on global trees'' (on
their prefixes), by considering each time the \emph{finite number of
  unevaluated steps}, namely the number of messages that have been
sent and not yet received. About this, we notice that formally there
is no isomorphism between the codatatype of infinite trees in
Definition~\ref{def:global trees-A} and the one we obtain by composing
prefixes of Definition~\ref{def:prefixes-A} and trees of
Definition~\ref{def:g-trees-alt-A}: the constructor
$\pmsgs \p \q \ell \ \dpre{\dots}$ can now appear only inside the
inductive prefix. However, this does not affect the unravelling
operation $\Re$ (Definition~\ref{def:gunroll-A}), since types are
unravelled in trees without any construct
$\pmsgs \p \q \ell \ \dpre{\dots}$, neither will affect any further
semantic description, since we will consider only a finite number of
semantic steps: only a finite number of messages will have been sent
after each semantic step of the system.

For simplicity and to stay closer to the intuition, throughout the
paper we will stick to the presentation of trees as a single
codatatype (Definition~\ref{def:global trees-A}). Where we need to
perform induction on prefixes we will explicitly mention it.
\end{remark}

\subsection{Local Types} \label{subsec:loc-A} For \emph{local types} (or
\emph{local session types}), we take the same approach as global
types: we formalise their inductive syntax and then we coinductively
unravel recursion to obtain possibly infinite trees.
\begin{definition}[Local types]
  \label{def:local-types-A}%
  {\em Local types} (datatype \code{l\_ty} in \code{Local/Syntax.v}), ranged over by $\lT$, %
  are generated by the following grammar:
  \[
    \begin{array}{rll}
    \lT & ::= & \\
    & & \lend\;\SEP\;\lX \;\SEP\;\lrec \lX \lT \;\SEP\;\\
      && \;\SEP\;  \lsend \q \ell {\tS} {\lT}\;\SEP\;%
  \lrecv \p \ell {\tS} {\lT}
    \end{array}
  \]
We require that $I\neq \emptyset,$ and $\ell_i \neq \ell_j$ whenever $i \neq j,$ for all $i,j\in I$.
\end{definition}
The session type $\lend$ says that no further communication is
possible and the protocol is completed. %
Recursion is modelled by the session type $\dlt{\mu\lX.\lT}$. The
\emph{send type} 
$\lsend \q \ell {\tS} {\lT}$ %
says that the participant implementing the type %
must choose a labelled message to send to $\q$; %
if the participant chooses the message $\ell_i$, for some $i\in I,$ %
it must include in the message to $\q$ a payload value of sort
$\tS_i$, %
and continue as prescribed by $\dlt{\lT_i}$. %
The \emph{receive
  type} 
$\lrecv \p \ell {\tS} {\lT}$ %
requires to wait to receive %
a value of sort $\tS_i$ (for some $i \in I$) %
from the participant $\p$, via a message with label $\ell_i$; %
if the received message has label $\ell_i$, %
the protocol will continue as prescribed by $\dlt{\lT_i}$. %

We restrict ourselves to closed local types and we require recursion
to be guarded. In the text we will mostly implicitly assume those. We
define analogous predicates to the ones for global types.
\begin{definition}[Guardedness for local types]
  \label{def:loc-guarded-A}%
  We say that a local type is \emph{guarded} (\code{lguarded} in
  \code{Local/Syntax.v}) according to the following definition:
\[
\small
  \begin{array}{c}
  \dfrac{}{\guarded\ \lend} \quad
  \dfrac{}{\guarded\ \lX} \quad
  \dfrac{\code{not\_pure\_rec}\;\lX\;\lT\quad \guarded\ \lT}{\guarded\ (\lrec \lX \lT)}\\[4mm]
  \dfrac{\forall i\in I. \guarded\ {\dlt{\lT_i}}}{\guarded\ (\lsend\q \ell {\tS} {\lT})}\quad
  \dfrac{\forall i\in I. \guarded\ {\dlt{\lT_i}}}{\guarded\ (\lrecv \p \ell {\tS} {\lT})}
  \end{array}
  \]
  where $\code{not\_pure\_rec}\;\lX\;\lT$ means that $\lT$ is
  different from $\lrec {\dlt{Y_1}} {\dots\lrec {\dlt{Y_n}} \lX}$ and
  also $\lT\neq\lX$
\end{definition}
\begin{definition}[Free variables and closure (local types)]
  \label{def:loc-closed-A}%
The set of \emph{free variables} of a local type $\lT$, $\lfv(\lT)$, (\code{l\_fidx} in \code{Local/Syntax.v}) is defined as follows.
\[
\small
\begin{array}{c}
\lfv(\lend) = \emptyset\quad
\lfv(\lX) = \{\lX\}\quad
\lfv(\lrec \lX \lT) = (\lfv(\lT))\backslash\{\lX\} \vspace*{2mm}\\
\gfv(\lsend \q \ell {\tS} {\lT}) = \bigcup\limits_{i\in I} \lfv({\dlt{\lT_i}})\\
\lfv(\lrecv \p \ell {\tS} {\lT}) = \bigcup\limits_{i\in I} \lfv({\dlt{\lT_i}})
\end{array}
\]
We say that $\lT$ is \emph{closed} (\code{l\_closed} in
\code{Local/Syntax.v}) if it does not contain free variables, namely
$\closed\ \lT \leftrightarrow (\lfv(\lT)=\emptyset)$.
\end{definition}
\noindent We provide a binding-free codatatype for local trees, whose structure derives from their syntax.

\begin{definition}[Semantic local trees]
  \label{def:local-trees-A}%
  {\em Semantic local trees} (datatype \code{rl\_ty} in
  \code{Local/Tree.v}), ranged over by $\lT$, %
  are terms generated \emph{coinductively} by the following grammar:
  \[
    \begin{array}{rll}
    \colT & ::= & \\
    & & \colend\;\SEP\;\colsend \p \ell {\tS} {\colT} \;\SEP\;\\
      && \;\SEP\;   \colrecv \q \ell {\tS} {\colT}
    \end{array}
  \]
 We require that $\p\neq \q$, $I\neq \emptyset,$ and $\ell_i \neq \ell_j$ whenever $i \neq j,$ for all $i,j\in I$.
\end{definition}
As is done for global types, we define the unravelling of a local type
into a local tree.
\begin{definition}[Local unravelling]
  \label{def:lunroll-A}%
  {\em Unravelling of local types} (definition \code{LUnroll} in
  \code{Local/Unravel.v}) is the relation between local types and semantic
  trees coinductively specified by the following rules:
  \[
    \begin{array}{ll}
   \rulename{l-unr-end} & \rulename{l-unr-rec}\\
  \newDfrac{ }{\lunroll \lend \colend} &
  \newDfrac{\lunroll {\dlt{ \lT \{ \lrec{\lX}{\lT} / \lX \}}} \colT}{\lunroll{\lrec \lX \lT} \colT}
  \\[3mm]
  \multicolumn{2}{l}{\rulename{l-unr-send}}\\
  \multicolumn{2}{l}{
    \newDfrac{\forall i\in I. \lunroll {\dlt{\lT_i}} {\dcol{\colT_i}}}
             {\lunroll{  \lsend \q \ell {\tS} {\lT} }{\colsend \q \ell {\tS} {\colT}}  }
  }\\[3mm]
  \multicolumn{2}{l}{\rulename{l-unr-recv}}\\
  \multicolumn{2}{l}{
    \newDfrac{\forall i\in I. \lunroll {\dlt{\lT_i}} {\dcol{\colT_i}} }{\lunroll{  \lrecv \p \ell {\tS} {\lT}}{\colrecv \p \ell {\tS} {\colT}} }
  }
  \end{array}
  \]
\end{definition}

\begin{remark}
We have required several ``well-formedness'' properties to types.
\begin{enumerate*}
\item $I\neq\emptyset$ in Definitions \ref{def:global-types-A},
  \ref{def:global trees-A}, \ref{def:local-types-A} and
  \ref{def:local-trees-A}, namely the continuations for global/local
  types/trees are not allowed to be empty.
\item Every recursion constructor in global/local types must be
  guarded (Definitions \ref{def:guarded-A} and \ref{def:loc-guarded-A}).
\item We only consider closed global/local types (Definitions
  \ref{def:closed-A} and \ref{def:loc-closed-A}).
\end{enumerate*}

In the rest of the paper we continue to implicitly assume those for
each object we consider; however in the Coq development such
conditions must be made explicit in definitions and statements. We
have formalised them with (co)inductive predicates. In particular for
global types we have defined \code{g\_precond} (in \code{Global/Syntax.v})
exactly as the conjunction of the three predicates listed above, while
for global trees we have defined \code{WF} (in \code{Projection/CProject.v}) to
ensure that the continuations in the tree are never empty.
\end{remark}

\subsection{Projections, or how to discipline communication}
\label{subsec:projection-A}
At the very core of the theory of multiparty session types, there is
the notion of \emph{projection}. We have laid down a setting, where
global types offer a bird's-eye perspective on communication and local
types take instead the point of view of a single participant. The
following definition is formalised to make sure that participants
respect what is globally prescribed for the protocol: each local type
$\lT$ protocol must be the projection, onto the respective
participant, of the global type $\G$.
\begin{definition}  \label{pro-A}%
\label{def:projection-A}%
The \emph{projection 
  of a global type onto a participant $\pr$} (\code{project} in
\code{Projection/IProject.v}) is a partial function
$\projt{\_} \pr: \gty \nrightarrow \lty $
defined by recursion on $\G$ whenever the recursive call is defined:
\[
\begin{array}{ll}
  \rulename{proj-end} & \rulename{proj-var}\\
  \projt{\gend} \pr = \lend & \projt{\gX}{\pr}=\lX \\[2mm]
  \rulename{proj-rec} & \\
  \multicolumn{2}{l}{
    \projt{(\grec \gX \G)}{\pr}=\lrec \lX {} (\projt{\G}{\pr})\ \text{if}\ \guarded  (\projt{\G}{\pr})
  }\\[2mm]
  \rulename{proj-send} & \\
  \multicolumn{2}{l}{
    \pr=\p\  \text{implies}\ \projt{\msgi \p\q \ell {\tS} {\G}} \pr=\lsendni \q \ell {\tS} { {\color{black}\projt{\G_{\dgt i}} \pr} }
  }\\[2mm]
  \rulename{proj-recv} & \\
  \multicolumn{2}{l}{
    \pr=\q\ \text{implies}\ \projt{\msgi \p\q \ell {\tS} {\G}} \pr=\lrecvni \p \ell {\tS} { {\color{black}\projt{\G_{\dgt i}} \pr} }
  }\\[2mm]
  \rulename{proj-cont} & \\
  \multicolumn{2}{l}{
    \begin{array}{l}
      \text{$\pr\neq \p$, $\pr \neq \q$ and $\forall i,j\in I$, $\projt{\G_{\dgt{i}}}\pr=\projt{\G_{\dgt{j}}}\pr$; implies}\\
      \projt{\msgi \p\q \ell {\tS} {\G}} \pr = \projt {\G_{\dgt 1}} \pr
    \end{array}
  }\\[4mm]
  \multicolumn{2}{l}{\text{undefined if none of the above applies.}}
\end{array}
\]
\end{definition}
We describe the clauses of Definition \ref{def:projection-A}:
\begin{description}
\item[\rulename{proj-end,proj-var}] give the projections for end-types
  and type variables; \item[\rulename{proj-rec}] gives the projection
  on recursive types;
\item[\rulename{proj-send} (resp.~\rulename{proj-recv})] %
  states that a global type starting with a communication %
  from $\pr$ to $\q$ (resp.~from $\q$ to $\pr$) %
  projects onto a sending (resp.~receiving) local type
  $\lsendni \q \ell {\tS} { {\color{black}\projt{\G_{\dgt i}} \pr}}$
  (resp.~$\lrecvni \p \ell {\tS} { {\color{black}\projt{\G_{\dgt i}}
      \pr} }$), provided that the continuations
  $\projt{\G_{\dgt{i}}}\pr$ are also projections of the corresponding
  global type continuations $\dgt{G_i}$;
\item[\rulename{proj-cont}] states that, if the projected global type
  starts with a communication between $\p$ and $\q$ and if we are
  projecting it onto a third participant $\pr$, then, for the
  projection to be defined, we need to make sure that continuation is
  \emph{the same} on all branches.%
\end{description}
To prove the main result of trace equivalence in Coq, we want to
conveniently work with coinductive trees, hence we also define the
projection of a global tree onto a participant.
\begin{example}[Projection]\label{ex:projection-A}
  We show the examples of projection of global types.
  First we would not be able to project
  \
  $ \dgt{\G'} = \msg{\Alice}{\Bob} \dgt{\{ \ell_1(\tnat) .
    \msg{\Bob}{\Carol} \ell(\tnat) .\gend ,}
  $\\
  $
    \dgt{\ell_2(\tnat) .
    \msg{\Alice}{\Carol} \ell(\tnat) .\gend \}}
  $
  onto \Carol{}, since, after skipping the first interaction between
  \Alice{} and \Bob{}, it would not be clear whether \Carol{} should
  expect a message
  from  \Alice{} or from \Bob{}. If instead we have
  \
  $
    \dgt{\G} = \msg{\Alice}{\Bob}
  $\\
  $
    \dgt{\{ \ell_1(\tnat) .
    \msg{\Bob}{\Carol} \ell(\tnat) .\gend , \ell_2(\tbool) .
    \msg{\Bob}{\Carol}}
  $\\
  $
  \dgt{\ell(\tnat).\gend \}}
  $
  the projection $\projt\G \Carol$ is well defined as the local type
  $\lT = \lrcv{\Bob}\dlt{\ell(\tnat).\lend}$. In Coq we have rendered
  this behaviour encoding the projection codomain as
  $\code{option}$  $\code{l\_ty}$, as is common practice when
  formalising partial functions.
\end{example}
\begin{definition}
  \label{co-pro-A}%
  \label{def:co-projection-A}%
  The \emph{projection 
    of a coinductive 
    global tree onto a participant $\pr$} (definitions \code{Project}
  and \code{IProj} in\\ \code{Projection/CProject.v}) is a relation
  $\coproj {\pr} {\_} {\_}\ :\ \rel{\cogty}{\colty} $ coinductively
  specified by the following clauses:
  \[
  \small
  \begin{array}{l}
    \rulename{co-proj-end}   \\
    \newDfrac{\neg\ \partof \pr \coG}{\coproj {\pr} {\coG}  {\colend}}\\[4mm]
    \rulename{co-proj-send-1}\\
    \newDfrac
        {\pr=\p\ \ \ \ \ \ \ \ \forall i\in I. \coproj {\pr} {\coG_{\dcog i}} {\colT_{\dcol i}} }
 	{\coproj {\pr} {\comsgni \p\q \ell {\tS} {\coG}} {\colsend \q \ell {\tS} {\colT}}}  \\[4mm]
    \rulename{co-proj-send-2}\\
    \newDfrac{\pr\neq\q\ \ \ \ \ \ \ \ \forall i\in I. \coproj {\pr} {\coG_{\dcog i}} {\colT_{\dcol i}} }
             { \coproj {\pr} {\comsgsi  \p\q {\ell_j} \ell {\tS} {\coG} } {\colT_{\dcol j } } } \\[6mm]
    \rulename{co-proj-recv-1}\\
    \newDfrac{\pr=\q\ \ \ \ \ \ \ \ \forall i\in I. \coproj {\pr} {\coG_{\dcog i}} {\colT_{\dcol i}} }
             { 	\coproj {\pr} { \comsgni \p\q \ell {\tS} {\coG}} {\colrecv \p \ell {\tS} {\colT} } } \\[4mm]
 \rulename{co-proj-recv-2}\\
 \newDfrac{\pr=\q\ \ \ \ \ \ \ \ \forall i\in I. \coproj {\pr} {\coG_{\dcog i}} {\colT_{\dcol i}} }
          {\coproj {\pr} {\comsgsi  \p\q {\ell_j} \ell {\tS} {\coG} } {\colrecv \p \ell {\tS} {\colT} } } \\[6mm]
 \rulename{co-proj-cont}\\
 \ \ \pr\neq\p\ \ \pr\neq\q\ \ \ \forall i\in I. \coproj {\pr} {\coG_{\dcog i}} {\colT_{\dcol i}} \\[1mm]
 \newDfrac{
          \ \ \ \ \ \ \forall i,j\in I.  \colT_{\dcol i}=\colT_{\dcol j}\ \ \ \ \ \ \
          \forall i\in I. \partof \pr {\coG_{\dcog i}}\ \ \ \ \ \ \ }
          {\coproj {\pr} { \comsgni \p\q \ell {\tS} {\coG} } {\colT_{\dcol 1} }}
        \end{array}
  \]
\end{definition}
The coinductive definition of projection follows the same intuition as
the recursive one (Definition \ref{def:projection-A}), but we have
extended this to the asynchronous construct
$\comsgs \p \q \ell\;\dots\;$, adapting to a coinductive setting the
definition in \cite[Appendix A.1]{DenielouYoshida2013}. %
In rules \rulename{co-proj-end} and \rulename{co-proj-cont} we have
added explicit conditions on participants being present in a global
tree.

\begin{definition}
\label{def:partof-A}
A role $\p$ is said to be \emph{participant of a global tree} $\coG$
(definition \pof{} in \code{Global/Tree.v}), when for $\p$ and $\coG$ the
following inductively defined predicate, $\pof{}\ \_\ \_$, holds:
\[
\small
\begin{array}{l}
  \dfrac{}{\partof \p {\comsgni \p\q \ell {\tS} {\coG}}} \\[3mm]
  \dfrac{}{\partof \q {\comsgni \p\q \ell {\tS} {\coG}}} \\[4mm]
  \dfrac{}{\partof \p {\comsgsi  \p\q {\ell_j} \ell {\tS} {\coG}}} \\[4mm]
  \dfrac{}{\partof \q {\comsgsi  \p\q {\ell_j} \ell {\tS} {\coG}}} \\[6mm]
  \dfrac{\exists i\in I.\ \partof{\pr}{\dcog{\coG_i}}}{\partof \pr {\comsgni  \p\q \ell {\tS} {\coG}}} \\[4mm]
  \dfrac{\exists i\in I.\ \partof{\pr}{\dcog{\coG_i}}}{\partof \pr {\comsgsi  \p\q {\ell_j} \ell {\tS} {\coG}}}
\end{array}
\]
\end{definition}
Definition~\ref{def:partof-A} captures the same concept as $\ \parti\ $
for global types. Such a predicate is inductive in its nature, even on
a coinductive datatype; the intuitive reason for this is that if $\p$
is a participant of $\coG$ it should be found, as a sending or
receiving role, within a \emph{finite}, albeit arbitrary, number of
steps in the branching structure of $\coG$. By factoring in the
predicate \pof , Definition \ref{co-pro-A} ensures (1) that the
projection of a global tree on a participant outside the protocol is
$\cogend$ (rule \rulename{co-proj-end}) and (2) that this discipline
is preserved in the continuations (rule \rulename{co-proj-cont}).

From a formalisation point of view, if we had tried to define the
above tree projection as a corecursive function, instead as
coinductive relation, we would have incurred problems for rules
\rulename{co-proj-send-2} and \rulename{co-proj-cont}: here, the
corecursive call of $\dcol{\colT_1}$ does not appear guarded by any
constructor. Also in rule \rulename{co-proj-cont}, we would need to
provide a coinductive proof for the hypothesis
$\forall i,j\in I. \colT_{\dcol i}=\colT_{\dcol j}$ and this would
lead to further complications. On the other side, working with
projection as a relation is common practice in the literature (see,
e.g.,\citeN[Definition 3.6]{GHILEZAN2019127}) and allowed us to have a
smoother development for trees in Coq.

\begin{example}
  The projection of coinductive trees is slightly more permissive than
  its inductive counterpart, as pointed out in \cite{GHILEZAN2019127},
  Remark 3.14. Let us consider for example:
\begin{equation}\label{eq:proj:one-A}
\G = \msg \p \q \dgt{\{\ \ell_{0}(\tnat).\G_{0},\ell_{1}(\tnat). \G_{1}\}}
\end{equation}
with
$\dgt{\G_{0}} = \grec \gX {\msg \p \pr \dgt{\ell(\tnat).}\gX}$ and
$\dgt{\G_{1}} = \msg \p \pr \dgt{\ell(\tnat).}$\\$\grec \gX {\msg \p \pr \ell(\tnat).\gX}$.
Then we have:
\begin{equation}\label{eq:proj:two-A}
  \begin{array}{c}
    \projt {\dgt{\G_{0}}} \pr = \lrec \lX {\lrcv \p \dlt{\ell(\tnat).\lX} }
    \\ \neq\\
    \lrcv \p \dlt{\ell(\tnat).} \lrec \lX {\lrcv \p \dlt{\ell(\tnat).\lX} } = \projt {\dgt{ \G_{1} }} \pr\  \text{,}
\end{array}
\end{equation}
then no rule from Definition \ref{def:projection-A} applies (in
particular \\ \rulename{proj-cont} does not), and the projection onto
$\pr$ for $\G$ is undefined. On the other end, it is clear that
$\dgt{G_{1}}$ is obtained by $\dgt{G_{0}}$ with ``one step of
unravelling'' or, formally, that the infinite tree associated by $\Re$
to both of
them 
is the same\\
$\dcog{\coG_{01}} = \comsgn \p \pr \dcog{\ell(\tnat).}\comsgn \p \pr
\dcog{\ell(\tnat).\;\dots}$. Indeed we observe that
$\dgt{\G_{1}} = \msg \p \pr \dgt{\ell(\tnat).}\dgt{\G_{0}}$, thus by
rule \rulename{g-unr-rec} of Definition \ref{def:gunroll-A}, we have
that from $\dgt{\G_{1}} \Re\dcog{\coG_{01}}$,
$\dgt{\G_{0}} \Re\dcog{\coG_{01}}$ must hold.
Thus we have:
\begin{equation}\label{eq:proj:three-A}
  \G\Re\coG\quad\text{with}\quad\coG = \comsgn \p \q \dcog{\{\ \ell_{0}(\tnat).\coG_{01},\ell_{1}(\tnat). \coG_{01}\}}
\end{equation}
By choosing
$\dcol{\colT_{01}} = \colrcv \p \dcol{\ell(\tnat).}\colrcv \p
\dcol{\ell(\tnat).\ \dots}$, we can indefinitely apply
\rulename{co-proj-recv-1} to get
$\coproj \pr {\dcog{\coG_{01}}} {\dcol{\colT_{01}}}$ and thus we
finally obtain $\coproj \pr \coG \colT$ by \rulename{co-proj-cont}:
namely while $\G$ does not admit any projection on $\pr$, its
unravelling $\coG$ does.
\end{example}

\begin{figure}

\begin{tikzpicture}
	\node(G0)at (0,0) {$\G$};
	\node(G1)at (3,0) {$\coG$};
	\node(L0)at (0,-1.5) {$\lT$};
	\node(L1)at (3,-1.5) {$\colT$};
	\node(M1)at (1.5,-0.75) {\small{\textcolor{orange}{(M.1)}}};
	\path[commutative diagrams/.cd,every arrow,font=\scriptsize]
	(G0) edge node[above] {$\Re$} (G1)
	(L0) edge node[above] {$\Re$} (L1)
	(G0) edge node[right] {$\upharpoonright$} (L0)
	(G1) edge node[right] {$\upharpoonright^\textsf{c}$} (L1)
	;

\end{tikzpicture}

\end{figure}

The above example shows that, inside a global type, when two branches
of a continuation are obtained by a different number of unravelling
steps of the same recursion type, syntactic projection (Definition
\ref{def:projection-A}) gets stuck. At the same time its coinductive
counterpart (Definition \ref{def:co-projection-A}) handles smoothly this
case, thanks to infinite unravelling that gives such recursion global
types the same representation.

To conclude this subsection, we state our first main result from the
formalisation, namely that unravelling preserves projections. This
completes the first metatheory
square 
\textcolor{orange}{(M.1)} of \theDiagram\ in Figure \ref{fig:dia}.

\begin{restatable}[Unravelling preserves projections]{theorem}{unrproA}\label{thm:unr-pro-A}(Theorem \code{ic\_proj} in \code{Projection/Correctness.v}.)\\ Given a global type $\G$, such that $\guarded\ \G$ and $\closed\ \G$, if
\begin{enumerate*}[label=(\alph*)]
\item there exists a local type $\lT$ such that $\projt{\G} \pr = \lT\ $,
\item there exists a global tree $\coG$ such that $\gunroll \G \coG$ and,
\item there exists a local tree $\colT$ such that $\lunroll \lT \colT$,
\end{enumerate*}
then $\coproj \pr \coG \colT$.
\end{restatable}

\begin{proof}[Proof Outline.]
The full proof is found in \code{Projection/Correctness.v} of our Coq formalisation. As an outline, before performing coinduction on the definition of the coinductive projection in $\coproj \pr \coG \lT$, we rule out the case in which $\G$ has the shape $\dgt{\grec \gX \G'}$. Formally this first step goes as follows.
\begin{itemize}
\item Since $\G=\dgt{\grec \gX \G'}$ and $\dgt{\G'\{\gX / \grec \gX \G'\}}$ have the same unravelling $\coG$ (Definition \ref{def:gunroll-A}), we can perform such operation (\emph{finite unravelling}, definition) on $\dgt{\grec \gX {\G'} }$ until we get $\dgt{\G_n}$ that is either a message-type or an end-type (we have as an hypothesis that $\G$ is closed and closure is preserved by finite unravelling).
\item We have proved that if $\projt \pr \G = \lT$, $\lunroll \lT \colT$ and $\projt \pr {\dgt{\G_n}} = \dlt{\lT_n}$ then $\lunroll {\dlt{\lT_n}} {\colT}$ (lemma \code{LUnroll\_ind} in \code{Local/Unravel.v}).
\end{itemize}
Thus, proving the theorem for every non-recursion global type $\dgt{\G_n}$, gives us the theorem for every $\G$ global type. We therefore assume that $\G$ is not a recursion-type and proceed by coinduction on $\coproj \pr \coG \lT$. We use the features of the PaCo \cite{paco} to modularise the proof: in a separate Coq lemma we can assume the coinductive hypothesis in the context and prove the statement with guardedness guaranteed by the PaCo features (lemma \code{project\_nonrec} in \code{Projection/Correctness.v}).
\end{proof}

\subsection{Projection Environments for Asynchronous Communication}
\label{subsec:async_projection-A}

In this subsection, we introduce key concepts for building an
asynchronous operational semantics for multiparty session types. We
define our semantics following \cite{DenielouYoshida2013}, where a
precise correspondence is drawn between communicating finite-state
automata and multiparty session types. We do not formalise an explicit
syntax for automata, but develop labelled transition systems for
global and local trees with automata in mind. We rely on \emph{queue
  environments} as communication buffers, shared between pairs of
local trees, which allow for asynchronicity of the execution, while
guaranteeing the disciplined behaviour of participants. Let us start
with a paradigmatic example: a simple message exchange between two
participants.

\begin{example}[Local trees for a simple message exchange]\label{ex:queue-A}
  Below we informally 
  use the notation $\colT_{\dcol{1}}\step\colT_{\dcol{2}}$ to indicate
  one semantic step between local trees. $\p$ sends a message to $\q$
  with label $\dcol{\ell}$ and payload of sort $\tS$ and continues on
  $\colT$, and dually $\q$ receives from $\p$ the message, with same
  label and payload, and then continues on $\dcol{\colT'}$:
  \[
    \colsnd \q \dcol{\ell(\tS).\colT} \step \colT
    \quad \mbox{and} \quad
    \colrcv \p \dcol{\ell(\tS).\colT'} \step \dcol{\colT'}
  \]
  For $\q$ to receive the message, it is necessary that $\p$ has first
  sent the message. To model this asynchronous behaviour, we use FIFO
  queues: in the designated queue $Q(\p,\q)$ (empty at first) we
  enqueue the message sent from $\p$ and we store it, until the
  message is received by $\q$ and dequeued.
  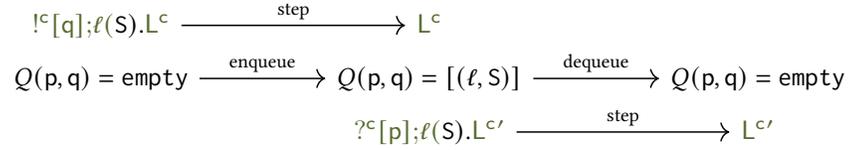
\begin{figure*}
  \[
  \hspace{32mm}
  \begin{tikzcd}[ampersand replacement=\&]
      \colsnd \q \dcol{\ell(\tS).} \colT \arrow[rr, "\text{step}"]\& \&\colT\&\&\\[-5mm]
      Q(\p,\q)=\code{empty}\arrow[rr,"\text{enqueue}"] \& \& Q(\p,\q)=[(\ell , \tS)]\arrow[rr,"\text{dequeue}"] \& \& Q(\p,\q)=\code{empty}\\[-5mm]
      \& \& \colrcv \p \dcol{\ell(\tS).} \colT{\color{dkolive}'}
      \arrow[rr, "\text{step}"]\& \& \colT{\color{dkolive}'}
  \end{tikzcd}
  \]
  \caption{Simple message exchange with a FIFO queue as buffer.} \label{fig-q-message-A}
  \end{figure*}
  Figure \ref{fig-q-message-A} summarises the intuition behind the semantics of a
  simple message exchange, starting from an empty queue: the message
  is received immediately after it has been sent. Such a FIFO queue
  allows for storing more than one message sent from $\p$ to $\q$,
  which $\q$ will receive in the same order they have been sent,
  according to the first-in-first-out discipline.
\end{example}

Considering our global protocols, we need one queue for each ordered
pair of participants $(\p,\q)$ to store the messages sent from $\p$ to
$\q$. We formally collect such queues in \emph{queue environments}.
\begin{definition}[Queue environments]
\label{def:qenv-A}
We call \emph{queue environment} any finitely supported function that
maps a pair of participants into a finite sequence (\emph{queue}) of
pairs of labels and sorts.
\end{definition}
In Coq we write the above type as
$\qenv = \{\fmap\ \ \ \role * \role \rightarrow \seq\; (\lbl *
\mty)\}$
(Notation \qenv~in \code{Local/Semantics.v}), where we have used
support from the Mathematical Components libraries
\cite{Gonthier:2010} for datatypes as finite function (\fmap) and
lists (or \emph{sequences}, \seq). Finite maps are formalised as
partial functions\footnote{More detail can be found at the
  Mathematical Components web page
  \url{https://math-comp.github.io/}.} with additional structure for
their finite domain. We use the Coq constructor $\code{None}$ for the
default return value of a partial function applied to an input value
outside its domain.


On queue environments we have defined the operation of enqueuing
\enq\ and dequeuing \deq\ as:
\[
\begin{array}{rcl}
\enq\ Q\ (\p,\q)\ (\ell,\tS) &=& Q[(\p,\q)\mapsfrom Q(\p,\q)@(\ell,\tS)]\\[1mm]
\deq\ Q\ (\p,\q) &=& \code{if}\ Q(\p,\q)=(\ell,\tS)\#s\\
                  && \code{then}\ ((\ell,\tS),Q[(\p,\q)\mapsfrom s])\\
		  && \code{else}\ \code{None}
\end{array}
\]
As for notation, we use $\#$ as the ``cons'' constructor for lists and
$@$ as the ``append'' operation; $f[x\mapsfrom y]$ denotes the
updating of a function $f$ in $x$ with $y$, namely
$f[x\mapsfrom y]\ x'=f\ x'$ for all $x\neq x'$ and
$f[x\mapsfrom y]\ x=y$, and $\code{None}$ is the default value for
partial functions provided by Coq. In case the sequence $Q(\p,\q)$ is
empty \deq~will not perform any operation on it, but return
$\code{None}$; in case the sequence is not empty it will return both
its head and its tail (as a pair). We denote the empty queue
environment by $\epsilon$, namely $\epsilon\ (\p,\q)=\code{None}$ for
all $(\p,\q)$.

Queue environments are used to regulate the asynchronous message
passing among participants for the whole protocol. We adapt the
projection of global types onto queue environment from \citep[Appendix
A.1]{DenielouYoshida2013}, to our coinductive setting.

\begin{definition}[Queue projection](Definition \code{qProject} in
  \code{Projection/QProject.v})
  \label{qpro-A}%
  \label{def:q-projection-A}%
  The \emph{projection on queue environments
  of a global tree} (\emph{queue projection} for short) %
  is the relation $\qproj {\_}  {\_}\ :\ \rel{\cogty}{\qenv} $ coinductively specified by the following clauses:
  \[
  \begin{array}{ll}
    \rulename{q-proj-end} &  \rulename{q-proj-send}\\
    \newDfrac{}{\qproj \cogend  {\epsilon}} \quad \quad \quad&
    \newDfrac{\forall i\in I. \qproj {\coG_{\dcog i}} {Q}\ \ \ \ \ \ \ \ Q(\p,\q)=\code{None}
    }{\qproj {\comsgni \p\q \ell {\tS} {\coG}} {Q}}\\[4mm]
    \rulename{q-proj-recv} & \\
    \multicolumn{2}{l}{
      \newDfrac{\qproj {\coG_{\dcog{j}}} {Q}\ \ \ \ \ \ \ \ \deq\ Q'(\p,\q)= ((\ell_j, \tS_j),Q)}
               {\qproj {\comsgsi \p\q {\ell_j} \ell {\tS} {\coG}} {Q'}}
    }
  \end{array}
  \]
\end{definition}
The projection of the tree $\cogend$ is $\epsilon$ as expected: once
the computation is terminated every queue is empty
(\rulename{q-proj-end}). Rule \rulename{q-proj-send} states that a
message $\comsgni \p\q \ell {\tS} {\coG}$ has $Q$ as its projections,
if $Q(\p,\q)$ is empty (no message has been yet sent between $\p$ and
$\q$) and $Q$ is also projection for each continuation
$\coG_{\dcog{i}}$ (where the message has been already sent and
received). Ultimately \rulename{q-proj-recv} states that a message
$\comsgsi \p\q {\ell_j} \ell {\tS} {\coG}$ has $Q'$ as its
projections, if $\deq\ Q' (\p,\q) = ((\ell_j,\tS_j),Q)$ and if $Q$ is
projection for each continuation $\coG_{\dcog{i}}$.

\begin{remark} \label{remark:qproj-on-prefix-A} For the sake of on-paper
  presentation, the above definition is presented as a coinductive
  predicate, dealing with coinductive objects (trees). Albeit this
  definition carries the correct concept, formally it is not accurate:
  in our formalisation, \code{qProject} is defined in Coq not as a
  codatatype, but as a datatype, \emph{inductively} on prefixes for
  global trees (see Remark \ref{remark:prefixes-A}).
\end{remark}

Queue projection has allowed us to associate to a global tree, in one
shot, all the queues involved in the protocol collected in a queue
environment. Along the same lines, we will consider all the local
types of the protocol at once, by defining the type of \emph{local
  environments}.
\begin{definition}[Local environments]
\label{def:env-A}
We call \emph{local environment}, or simply \emph{environment}, any
finitely supported function $\dcol{E}$ that maps participants into
local types.
\end{definition}
In Coq we write the above type as
$\env = \{\fmap\ \ \ \role \rightarrow \rlty\}$ (Notation \env~in
\code{Local/Semantics.v}).

As anticipated, we are interested in those environments that are
defined on the participants of a protocol (global tree $\coG$) and
that map each participant $\p$ to the projection of $\coG$ onto such
$\p$. 
\begin{definition}[Environment projection](Definition\\ \code{eProject} in \code{Projection/CProject.v}.)
\label{def:epro-A}
We say that $\dcol{E}$ is an environment projection for $\coG$,
notation $\coG\upharpoonright\dcol{E}$, if it holds that
\ $\forall \p.\ \coproj \p  \coG {(\dcol{E\ \p})}$.
\end{definition}
%

In the next sections we establish a relation between the semantics for
global and local types. In the local case, we define the semantics on
local environments together with queue environments. In the statements
of our soundness and completeness results we therefore consider the
projection of a global tree both on local environments and on queue
environments, together in one shot.
\begin{definition}[One-shot projection](Definition\\ \code{Projection} in \code{Projection.v})
\label{def:ospro-A}
We say that the pair of a local environment and of a queue environment $(\dcol{E},Q)$ is a (one-shot) projection for the global tree $\coG$, notation $\coG\osproj(\dcol{E},Q)$ if it holds that:
$\coG\upharpoonright\dcol{E} \quad \text{and} \quad \qproj \coG Q$.
\end{definition}
\begin{example}
  Let us consider the global tree:
  $\coG = \comsgs \p \q \ell$\\ $\dcog{\ell(\tS).}\comsgn \q \p
  \dcog{\ell(\tS).}\comsgn \q \p \allowbreak\dcog{\ell(\tS).\dots}\ \
  $. Participant $\p$ has sent a message to $\q$, $\q$ will receive it
  next (but has not yet) and then the protocol continues indefinitely
  with $\q$ sending a message to $\p$ after the other. We define
  $\dcol{E}$---with support $\{\p,\q\}$, since these are the only two
  participants involved---such that:
  $\ \ \dcol{E\ \p}=\colrcv \q \dcol{\ell(\tS).}\colrcv \q
  \dcol{\ell(\tS).\dots}\ \ $ and
  $\ \ \dcol{E\ \q}=\colrcv \p \dcol{\ell(\tS).}\colsnd \p
  \dcol{\ell(\tS).}\colsnd \p \dcol{\ell(\tS).\dots}\ \ $. We then
  define $Q$---with support a subset of $\{(\p,\q),(\q,\p)\}$, since
  messages are sent only from $\p$ to $\q$ or from $\q$ to $\p$---such
  that: $\ \ Q(\p,\q)=[(\ell,\tS)]\ \ $ and
  $\ \ Q(\q,\p)=\code{None}$. Following the definitions in this
  section it is easy to verify that $\coG\osproj(\dcol{E},Q)$; observe
  that the only ``message'' enqueued in $Q$ is $(\ell,\tS)$, since
  this is the only one sent, but not yet received (at this stage of
  the execution).
\end{example}

\subsection{Labelled Transition Relations for Tree Types}
\label{subsec:step-A}

We define \emph{trace semantics} both for types and for processes. At
the core of trace semantics that we define for session types, lies a
labelled transition system (LTS) defined on trees, with regard to
\emph{actions}. In this section we present the basic definitions and
results---up to soundness and completeness of the local reduction with
respect to the global one---, following the structure of our Coq
formalisation.

The basic \emph{actions} (datatype \code{act} in \code{Common/Actions.v}) of
our asynchronous communication are objects, ranged over by $a$, of the
shape either:
\begin{itemize}
\item $!\p\q(\ell,\tS)$: send $!$ action, from participant $\p$ to
  participant $\q$, of label $\ell$ and payload type $\tS$, or
\item $?\q\p(\ell,\tS)$: receive $?$ action, from participant $\p$ at
  participant $\q$, of label $\ell$ and payload type $\tS$.
\end{itemize}
We define the \emph{subject} of an action $a$ (definition
\code{subject} in \code{Common/Actions.v}), $\subject a$, as $\p$ if
$a=!\p\q(\ell ,\tS)$ and as $\q$ if $a=?\q\p(\ell,\tS)$.%
\footnote{The representation of actions is directly taken from
  \cite{DenielouYoshida2013}, however we have swapped the order of
  $\p$ and $\q$ in the receive action, so that the subject of an
  action always occurs in first position.}

Given an action, our types (represented as trees) can perform a
reduction step.

\begin{definition}[LTS for global trees (\code{step} in \code{Global/Semantics.v})]\label{def:gstep-A}
  \ \\
  The \emph{labelled transition relation for global trees}
  (\emph{global reduction} or \emph{global step} for short) %
  is, for each action $a$, the relation
  $\_\ \stepa{a}\ \_\ :\ \rel{\cogty}{\cogty} $ inductively specified
  by the following clauses:
  \[
\begin{array}{l}
  \rulename{g-step-send}\\
  \dfrac{a=!\p\q(\ell_j,\tS_j)}{\comsgni \p \q \ell \tS \coG \stepa{a} \comsgsi \p \q {\ell_j} \ell \tS \coG}\\[4mm]
  \rulename{g-step-recv}\\
  \dfrac{a=?\q\p(\ell_j,\tS_j)}{  \comsgsi \p \q {\ell_j} \ell \tS \coG \stepa{a} \coG_{\dcog{j}} }\\[4mm]

\rulename{g-step-str1} \\
\dfrac{\subject a \neq \p \quad \subject a \neq \q \quad \forall i\in I. \coG_{\dcog{i}}\stepa{a}\dcog{\coG'_i}}{ \comsgni \p \q \ell \tS \coG  \stepa{a} \comsgni \p \q \ell \tS {\coG'} }\\[4mm]
 \rulename{g-step-str2}\\
\dfrac{\subject a \neq \q \quad \coG_{\dcog{j}}\stepa{a}\dcog{\coG'_j}\quad \forall i\in I\backslash\{ j \}. \coG_{\dcog{i}}=\dcog{\coG'_i}}{ \comsgsi \p \q {\ell_j} \ell \tS \coG  \stepa{a} \comsgsi \p \q {\ell_j} \ell \tS {\coG'} }

\end{array}
\]
\end{definition}

The step relation describes a labelled transition system for global
trees with the following intuition:
\begin{description}
\item[\rulename{g-step-send}] with the sending action
  $!\p\q(\ell_j,\tS_j)$, the global tree
  $\comsgni \p \q \ell \tS \coG$ can perform a step into
  $\comsgsi \p \q {\ell_j} \ell \tS \coG$: this is the \emph{sending
    base case}, where a message with label $\ell_j$ and payload type
  $\tS_j$ is sent by $\p$, but not yet received by $\q$;
\item[\rulename{g-step-recv}] with the receiving action
  $?\q\p(\ell_j,\tS_j)$, the global tree
  $\comsgsi \p \q {\ell_j} \ell \tS \coG$ can perform a step into
  $\coG_{\dcog{j}}$: this is the \emph{receiving base case}, where a
  message with label $\ell_j$ and payload type $\tS_j$, that was
  previously sent by $\p$, is now received by $\q$;
\item[\rulename{g-step-str1}] with an action $a$ a step is allowed
  to be performed under a sending constructor
  $\comsgn \p \q \dcog{\dots}$ each time that the subject of that
  action is different from $\p$ and from $\q$ and, coinductively, each
  continuation steps accordingly, namely
  $\forall i\in I. \coG_{\dcog{i}}\stepa{a}\dcog{\coG'_i}$;
\item[\rulename{g-step-str2}] with an action $a$ a step is allowed
  to be performed under a receiving constructor
  $\comsgs \p \q {\ell_j} \dcog{\dots}$ each time that the subject of
  that action is different from $\q$ ($\p$ has already sent the
  message and the label $\ell_j$ has already been selected), the
  continuation corresponding to the label $\ell_j$ steps accordingly,
  namely $\coG_{\dcog{j}}\stepa{a}\dcog{\coG'_j}$, and each other
  continuation stays the same, namely
  $\forall i\in I\backslash\{ j \}. \coG_{\dcog{i}}=\dcog{\coG'_i}$.
\end{description}
\begin{remark}\label{remark:nondet-A}
The semantics allows for some degree of non-determinism. For instance,
$\comsgni \p \q \ell \tS \coG$ could perform a step according to both
rules \rulename{g-step-send} and \rulename{g-step-str1} (depending
on the subject of the action).
\end{remark}


Below we formalise the intuition from Example \ref{ex:queue-A}, and we
define a transition system for environments of local trees, together
with environments of queues.
\begin{definition}[LTS for environments (\code{l\_step} in \code{Local/Semantics.v})]\label{lstep-A}\ \\
  The \emph{labelled transition relation for environments}
  (\emph{local reduction} or \emph{local step} for short) is, for each
  action $a$, the relation
  $\_\ \stepa{a}\ \_\ :\ \rel{(\env * \qenv)}{(\env * \qenv)} $
  inductively specified by the following clauses:
\[
\begin{array}{l}
  \rulename{l-step-send}\\
  \dfrac{a=!\p\q(\ell_j,\tS_j)\quad \dcol{E\ \p} = \colsend \q \ell \tS \colT}
        {(\dcol{E},Q) \stepa{a} (\dcol{E[\p\mapsfrom \colT_j]},\enq\ Q\ (\p,\q)\ (\ell_j,\tS_j))}\\[5mm]
  \rulename{l-step-recv} \qquad a=?\q\p(\ell_j,\tS_j)\\
  \dfrac{\dcol{E\ \q} = \colrecv \p \ell \tS \colT \quad Q(\p,\q)=(\ell_j,\tS_j)\#s}
        {(\dcol{E},Q) \stepa{a} (\dcol{E[\q\mapsfrom \colT_j]},Q[(\p,\q)\mapsfrom s])}
\end{array}
\]
Notice that, if the third condition in the premise of\\
\rulename{l-step-recv} is satisfied, in its conclusion
$Q[(\p,\q)\mapsfrom s] = \pi_2\ (\deq\ Q\ (\p,\q))$ (where $\pi_2$ is
the projection on the second component of a pair).
\end{definition}

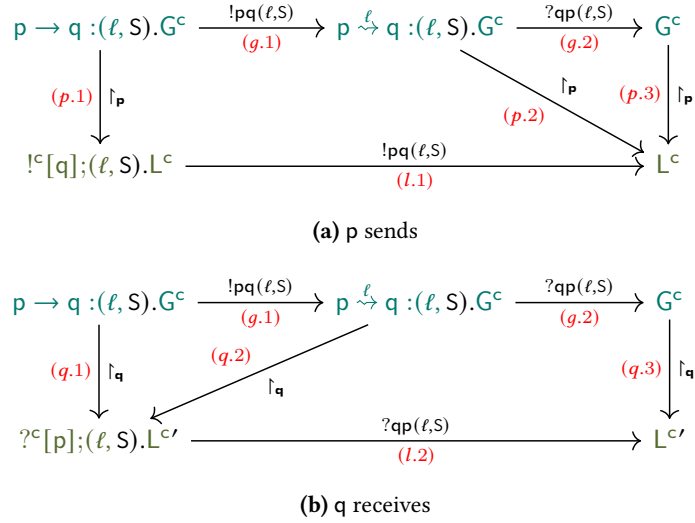
\begin{figure}
\begin{subfigure}{\columnwidth}
  \[
  \hspace{-4mm}
\begin{tikzcd}
\comsgn \p \q \dcog{(\ell, \tS).} \coG \arrow[rr, "{!\p\q (\ell,\tS)}","(g.1)"' red]\arrow[dd,"\upharpoonright_{\p}","(p.1)"' red]
& &
\comsgs \p \q {\ell}  \dcog{(\ell, \tS).} \coG\arrow[rr, "{?\q\p (\ell,\tS)}","(g.2)"' red]\arrow[rrdd, "\upharpoonright_{\p}","(p.2)"' red]
& &
\coG\arrow[dd, "\upharpoonright_{\p}","(p.3)"' red] \\
& & & &\\
\colsnd \q \dcol{(\ell, \tS).} \colT \arrow[rrrr, "{!\p\q (\ell,\tS)}","(l.1)"' red]& & & &\colT
\end{tikzcd}
\]
\subcaption{$\p$ sends}\label{subfig:send-A}
\end{subfigure}\newline\vspace{2mm}
\begin{subfigure}{\columnwidth}
  \[
  \hspace{-4mm}
\begin{tikzcd}
\comsgn \p \q \dcog{(\ell, \tS).} \coG\arrow[rr, "{!\p\q(\ell,\tS)}","(g.1)"' red]\arrow[dd,"\upharpoonright_{\q}","(q.1)"' red]
& &
\comsgs \p\q \ell \dcog{(\ell, \tS).} \coG\arrow[rr, "{?\q\p(\ell,\tS)}","(g.2)"' red]\arrow[lldd, "\upharpoonright_{\q}","(q.2)"' red]
& &
\coG\arrow[dd, "\upharpoonright_{\q}","(q.3)"' red] \\
& & & &\\
\colrcv \p \dcol{(\ell, \tS).} \dcol{\colT'} \arrow[rrrr, "{?\q\p(\ell,\tS)}","(l.2)"' red]& & & &\dcol{\colT'}
\end{tikzcd}
\]
\subcaption{$\q$ receives}\label{subfig:recv-A}
\end{subfigure}
\vspace{-2mm}
\caption{Basic send/receive steps for global and local trees}\label{fig:steps-A}
\vspace{-3mm}
\end{figure}

\begin{example}[Basic send/receive steps for global and local trees]
  Figure \ref{subfig:send-A} shows the transitions for a global tree,
  regulating the sending of a message from $\p$ to $\q$, and
  simultaneously the local transition for its projection on $\p$. The
  asynchronicity of our system is witnessed by the two different
  steps: $\color{red}{(g.1)}$, for the sending action
  $!\p\q(\ell,\tS)$, and $\color{red}{(g.2)}$, for the receiving one
  $?\q\p(\ell,\tS)$. Projecting
  $\comsgn \p \q \dcog{(\ell, \tS).} \coG$ on $\p$ (arrow
  $\color{red}{(g.1)}$) gives us a local tree that performs a sending
  step $\color{red}{(l.1)}$ corresponding to $\color{red}{(g.1)}$, and
  projection is preserved (arrow $\color{red}{(p.2)}$). However this
  does not happen for the receiving step $\color{red}{(g.2)}$: here
  the projections on $\p$ of
  $\comsgs \p\q \ell \dcog{(\ell, \tS).} \coG$ along
  $\color{red}{(p.2)}$ and of $\coG$ along $\color{red}{(p.3)}$ are
  the same. The situation is dual if we consider the projection on the
  receiving participant $\q$, Figure \ref{subfig:recv}. Here the
  projections along $\color{red}{(q.1)}$ and $\color{red}{(q.2)}$,
  corresponding to the global tree performing a sending action, result
  in the same local tree. We have instead a local step
  $\color{red}{(l.2)}$ preserving the local projections on $\q$ along
  $\color{red}{(q.2)}$ and $\color{red}{(q.3)}$ for the receiving
  action along $\color{red}{(g.2)}$.
\end{example}

Figure \ref{fig:steps-A} confirms our intuition: when the global tree
performs one step, \emph{there is one local tree}, projection of the
global tree on a participant, such that it performs a corresponding
step. We have indeed defined semantics for collections of local trees,
as opposed to single local trees. The formal relation of the
small-step reductions with respect to projection is established with
soundness and completeness results.

\begin{restatable}[Step Soundness]{theorem}{stepsoundA}(Theorem \code{Project\_step} in \code{TraceEquiv.v})\label{thm:step-sound-A}\ \\
  If $\coG \stepa{a} \dcog{\coG'}$ and $\coG \osproj (\dcol{E},Q)$,
  there exist $\dcol{E'}$ and $Q'$ such that
  $\dcog{\coG'} \osproj (\dcol{E'},Q')$ and
  $(\dcol{E},Q) \stepa{a} (\dcol{E'},Q')$.
\end{restatable}

\begin{proof}[Proof Outline.] The proof follows the intuition displayed by Figure \ref{fig:steps-A}. We identify three major proof steps:
\begin{enumerate}
\item we explicitly build the pair $(\dcol{E'},Q')$ from $(\dcol{E},Q)$;
\item we prove  $\dcog{\coG'} \osproj (\dcol{E'},Q')$;
\item we prove $(\dcol{E},Q) \stepa{a} (\dcol{E'},Q')$.
\end{enumerate}

$(1)$ The pair $(\dcol{E'},Q')$ is the result of the function \code{run\_step} in \code{Local/Semantics.v}, applied to $a$ and $(\dcol{E},Q)$. It is defined as follows:
\begin{itemize}
\item if $a=!\p\q(\ell_{j},\tS_{j})$ and $\dcol{E\ \p}=\colsend \q \ell \tS \colT$, then $\code{run\_step}\ a\ (\dcol{E},Q)=(\dcol{E'},Q')$, where $E'=E[\p\mapsfrom\dcol{\colT_{j}}]$ and $Q'=\enq\ Q\ (\p,\q)\ (\ell_{j},\tS_{j})$;
\item if $a=?\q\p(\ell_{j},\tS_{j})$,
 $\dcol{E\ \p}=\colrecv \p \ell \tS \colT$
 and $Q(\p,\q)=(\ell_{j},\tS_{j})\#s$,
 then $\code{run\_step}\ a\ (\dcol{E},Q)=(\dcol{E'},Q')$, where $E'=E[\q\mapsfrom\dcol{\colT_{j}}]$ and $Q'=\pi_2 (\deq\ Q\ (\p,\q))$;
\item if none of the above $\code{run\_step}\ a\ (\dcol{E},Q)=(\dcol{E},Q)$ (this is just intended as a default output to formally define the function in Coq).
\end{itemize}

Note that we have built $(\dcol{E'},Q')$ according to the effect that we expect that the one-step local reduction has on $(\dcol{E},Q)$.

$(2)$ In order to prove that our candidate $(\dcol{E'},Q')$ is indeed projection for $G'$. In the formalisation we have outsourced this to the lemma \code{runstep\_proj} in \code{TraceEquiv.v}. The proof proceed by induction on the step relation (Definition \ref{def:gstep}) in hypothesis $\coG \stepa{a} \dcog{\coG'}$. The base cases, corresponding to rules \rulename{g-step-send} and \rulename{g-step-rcv}, are handled by the two following lemmas (both in \code{TraceEquiv.v}):
\begin{itemize}
\item \code{Projection\_send}: if $\comsgni \p \q \ell \tS \coG \osproj (\dcol{E},Q)$ then $\comsgsi \p \q {\ell_{j}}  \ell \tS \coG \osproj (\dcol{E'},Q')$, with $(\dcol{E'},Q')$ as defined in $(1)$ with $a=!\p\q(\ell_{j},\tS_{j})$;
\item \code{Projection\_recv}: if $\comsgsi \p \q {\ell_{j}} \ell \tS \coG \osproj (\dcol{E},Q)$ then $\dcog{\coG_{j}}\osproj(\dcol{E'},Q')$, with $(\dcol{E'},Q')$ as defined in $(1)$ with $a=?\q\p(\ell_{j},\tS_{j})$.
\end{itemize}
The two recursive cases, corresponding to rules \rulename{g-step-str1} and \rulename{g-step-str2},
are also handled separately. These cases are less intuitive and more tedious to prove. We omit the details, however the method is the same for both:
\begin{itemize}
\item first we prove that we can describe $(\dcol{E_i},Q_i)$ the one shot projection for each tree continuation of $\comsgni \p \q \ell \tS \coG$ (respectively $\comsgsi \p \q {\ell_{j}} \ell \tS \coG$) in terms of the function \code{run\_step} above---lemmas \code{Proj\_None\_next} and \code{Proj\_Some\_next} in \code{TraceEquiv.v}---;
\item then we use the induction hypothesis to obtain $(\dcol{E_i},Q_i)\stepa{!\p\q(\ell_{j},\tS_{j})}(\dcol{E'_i},Q'_i)$ (respectively $(\dcol{E_i},Q_i)\stepa{?\q\p(\ell_{j},\tS_{j})}(\dcol{E'_i},Q'_i)$) as projections for the continuations in $\comsgni \p \q \ell \tS \coG'$ (respectively $\comsgsi \p \q {\ell_{j}} \ell \tS \coG'$);
\item finally we build back $(\dcol{E'},Q')$ from these, such that $\comsgni \p \q \ell \tS \coG'\osproj(\dcol{E'},Q')$ (respectively $\comsgsi \p \q {\ell_{j}} \ell \tS \coG'\osproj(\dcol{E'},Q')$).
\end{itemize}
The proof above requires ``compatibility'' and ``synchronisation'' lemmas, e.g., to make sure that when we build $(\dcol{E'},Q')$ from the different $(\dcol{E'_i},Q'_i)$, we obtain exactly the result of applying \code{run\_step} to  $(\dcol{E},Q)$.

$(3)$ Lastly we need to prove $(\dcol{E},Q)\stepa{a}(\dcol{E'},Q')$, and proceed by induction on
$\coG \stepa{a} \dcog{\coG'}$. Here we need to show that if $\coG$ performs a step with the action $a$, then its one-shot projection $(\dcol{E},Q)$ \emph{is able to perform a step} with the same action $a$; then we now that this step will be performed $(\dcol{E'},Q')$, which has been defined via \code{run\_step} exactly with this purpose. In \code{Local/Semantics.v} we define a predicate, \code{runnable : env * qenv $\rightarrow$ bool}, that formalises the concept that an environment is able to perform a step, returning \emph{true} or \emph{false} accordingly. Thus we conclude, by proving the next results:
\begin{itemize}
\item if $\coG \stepa{a} \dcog{\coG'}$ and $\coG\osproj(\dcol{E},Q)$ then $\code{runnable}\ (\dcol{E},Q)$ (lemma \code{local\_runnable} in \code{TraceEquiv.v});
\item if $\code{runnable}\ (\dcol{E},Q)$ then $(\dcol{E},Q)\stepa{a}(\dcol{E'},Q')$, where $(\dcol{E'},Q')=\code{run\_step}\ (\dcol{E},Q)$ (lemma \code{run\_step\_sound} in \code{Local/Semantics.v}).
\end{itemize}
\end{proof}

Dually, we prove completeness for step semantics on trees. The
intuition is the same as for soundness, but reading Figures
\ref{subfig:send-A} and \ref{subfig:recv-A} from bottom to top: each time
a local tree in the environment performs a step, the global tree also
performs one.

\begin{restatable}[Step Completeness]{theorem}{stepcomplA}(Theorem\\ \code{Project\_lstep} in \code{TraceEquiv.v})\label{thm:step-compl-A}\ \\
  If $(\dcol{E},Q) \stepa{a} (\dcol{E'},Q')$ and
  $\coG \osproj (\dcol{E},Q)$, there exist $\dcog{\coG'}$ such that
  $\dcog{\coG'} \osproj (\dcol{E'},Q')$ and
  $\coG \stepa{a} \dcog{\coG'}$.
\end{restatable}

\begin{proof}[Proof Outline.]
The proof structure is the following (again, the intuition for the base cases is carried by Figure \ref{fig:steps-A}):
\begin{enumerate}
\item we prove that exists $\dcog{\coG'}$ such that $\coG \stepa{a} \dcog{\coG'}$;
\item we prove that for this very $\dcog{\coG'}$ it must hold that $\dcog{\coG'} \osproj (\dcol{E'},Q')$.
\end{enumerate}

$(1)$ is taken care of by lemma \code{Project\_gstep} in \code{TraceEquiv.v}. The proof of such lemma proceed by induction on the prefix of the global tree $\coG$ (see Remark \ref{remark:prefixes-A}). The case for $\pgend$ is outsourced to the lemma \code{CProj\_step} in \code{TraceEquiv.v}. The induction cases in lemma \code{Project\_gstep}, including the one handled by \code{CProj\_step}, are all solved thanks to a---quite tedious---combination of case analysis and inversion lemmas about projections (collected in the lemma \code{Project\_inv} \code{Projection/CProject.v}).

The proof for $(2)$ is more interesting. The goal itself is handled by \code{Project\_gstep\_proj} in \code{TraceEquiv.v}. First we observe that, given $(1)$, namely $\coG \stepa{a} \dcog{\coG'}$, and the hypothesis $\coG \osproj (\dcol{E},Q)$, we know that for $(\dcol{E''},Q'')=\code{run\_step}\ (\dcol{E},Q)$ it holds that $\dcog{\coG'} \osproj (\dcol{E''},Q'')$ (see proof of Theorem \ref{thm:step-sound} and lemma \code{runstep\_proj} in \code{TraceEquiv.v}). Then we observe that, again by case analysis and inversion, we can prove lemma \code{lstep\_eq} in \code{Local/Semantics.v}:
\[
\text{If}\quad(\dcol{E},Q) \stepa{a} (\dcol{E'},Q')
\quad\text{and}\quad(\dcol{E},Q) \stepa{a} (\dcol{E''},Q'')
\text{,}\quad\text{then}\quad (\dcol{E''},Q'')=(\dcol{E''},Q'')\text{.}
\]
We conclude by lemma \code{run\_step\_compl} in \code{Local/Semantics.v}, that combines the above result with lemma \code{run\_step\_sound} in \code{Local/Semantics.v}. Indeed this guarantees that the hypothesis $\quad(\dcol{E},Q) \stepa{a} (\dcol{E''},Q'')$ in \code{lstep\_eq} above is satisfied (remember that we have chosen $(\dcol{E''},Q'')=\code{run\_step}\ (\dcol{E},Q)$; see again the proof for the soundness, Theorem \ref{thm:step-sound-A}).
\end{proof}

\subsection{Trace Semantics and Trace Equivalence}
\label{subsec:trace-eq-A}

To conclude the presentation of the metatheory we show trace
equivalence for global and local types. The end result is the Coq
formalisation of an adaptation of Theorem 3.1 in
\cite{DenielouYoshida2013} to our definition of semantics via
coinductive trees.

Traces are defined simply as streams of actions.

\begin{definition}[Traces]
  \label{def:traces-A}%
  (Codatatype \code{trace} in \code{Common/Action.v}), ranged
  over by $t$, are terms generated \emph{coinductively} by
  $\ \ t\ \ \ ::=\quad \trend \;\SEP\; \trnext a t\quad$ where $a$
  is either a sending action $!\p\q(\ell,\tS)$ or a receiving one
  $?\q\p(\ell,\tS)$ (\sec\ref{subsec:step-A}). We use the same notation
  as for lists, however we bare in mind that this definiton is
  coinductive, hence it generates possibly infinite streams.
\end{definition}

We associate traces to the execution of global and local trees.

\begin{definition}[Admissible traces for a global
  tree]\label{def:g-traces-A}
  We say that a trace is admissible for a global tree if the
  coinductive relation $\glts \_ \_$ (definition \code{g\_lts} in
  \code{Global/Semantics.v}) holds:

  \[
    \newDfrac{}{\glts  \trend \cogend} \quad %
    \newDfrac{\coG \stepa{a}  {\dcog{\coG'}}\quad \glts t {\dcog{\coG'}} }{ \glts {\trnext a t} \coG}
  \]
\end{definition}

For local trees, we consider the whole protocol, namely the pair of
local and queue environments.

\begin{definition}[Admissible traces for environments]\label{l-traces-A}
  We say that a trace is admissible for a pair of a local environment
  and a queue environment if the coinductive relation $\llts \_ \_$
  (definition \code{l\_lts} in \code{Local/Semantics.v}) holds:

\[
\newDfrac{\forall \p. \dcol{E}\ \p=\code{None} 
}{\llts  \trend (\dcol{E},\epsilon)} \quad %
\newDfrac
{(\dcol{E},Q) \stepa{a}  (\dcol{E'},Q')\quad \llts t {(\dcol{E'},Q')}}
{ \llts {\trnext a t} {(\dcol{E},Q)}}
\]

\end{definition}

Observe that, given the element of non-determinism in our semantics
(see \sec\ref{subsec:step-A}), generally more than one execution trace
are admissible for a global tree (or for an environment).

We can now state and prove the \emph{trace equivalence} theorem for
multiparty session types.

\begin{theorem}[Trace equivalence](Theorem\\ \code{TraceEquivalence} in \code{TraceEquiv.v}.)\label{thm:trace-equiv-A}\ \\
  If $\coG\osproj(\dcol{E},Q)$ then
$\ \glts t \coG\ $
if and only if $\ \llts t (\dcol{E},Q)\ $.
\end{theorem}
\begin{proof}[Proof Sketch.] (Theorem \code{TraceEquivalence} in \code{TraceEquiv.v}.)

  \textbf{(If)} We assume $\coG\osproj(\dcol{E},Q)$ and $\glts t \coG$
  and we proceed by coinduction (exploiting the techniques from the
  Paco library \cite{paco}) on the $\llts \_ \_$ relation in the goal,
  followed by a case analysis on $\glts \_ \_$ in hypothesis. The base
  $\trend$ case is handled simply by inversion lemmas, while the
  coinductive one is solved thanks to the soundness theorem (Theorem
  \ref{thm:step-sound-A}).

  \textbf{(Only If)} We assume $\coG\osproj(\dcol{E},Q)$ and
  $\llts t (\dcol{E},Q)$. Again we proceed by coinduction (again
  exploiting the Paco techniques \cite{paco}) on the $\glts \_ \_$
  relation in the goal, followed by a case analysis on $\llts \_ \_$
  in hypothesis. The base $\trend$ case is handled simply by inversion
  lemmas, while the coinductive one is solved by the completeness
  result (Theorem \ref{thm:step-compl-A}).
\end{proof}

The above result concludes our formalisation effort of the metatheory
of multiparty session types, from their syntactic specification to the
equivalence of global and local semantics. We have built the
formalisation of the type-related part of \theDiagram : squares
\textcolor{orange}{(M.1)} and \textcolor{orange}{(M.2)} in Figure
\ref{fig:dia}. 

\begin{figure}

\hspace{-2mm}\begin{tikzpicture}
	\node(G0)at (0,0) {$\G$};
	\node(G1)at (3,0) {$\coG$};
	\node(G2)at (6,0) {\small{global trace}};
	\node(L0)at (0,-1.5) {$\lT$};
	\node(L1)at (3,-1.5) {$\colT$};
	\node(L2)at (6,-1.5) {\small{local trace}};
	\node(M1)at (1.5,-0.75) {\small{\textcolor{orange}{(M.1)}}};
	\node(M2)at (4.5,-0.75) {\small{\textcolor{orange}{(M.2)}}};
	\path[commutative diagrams/.cd,every arrow,font=\scriptsize]
	(G0) edge node[above] {$\Re$} (G1)
	(G1) edge node[above] {LTS} (G2)
	(L0) edge node[above] {$\Re$} (L1)
	(L1) edge node[above] {LTS} (L2)
	(G0) edge node[right] {$\upharpoonright$} (L0)
	(G1) edge node[right] {$\upharpoonright^\textsf{c}$} (L1)
	(G2) edge[<->] node[right] {$=$} (L2)
	;

\end{tikzpicture}

\end{figure}

\section{Process extraction}
\label{appendix:extraction}

This function, available in \code{Proc.v}, translates a \dslName
process into a monadic value.

\begin{lstlisting}[language=Coq]
Section ProcExtraction.
  Fixpoint extract_proc (d : nat) (p : Proc) : MP.t unit :=
    match p with
    | Finish => MP.pure tt
    | Jump v => MP.set_current (d - v)
    | Loop p => MP.loop d (extract_proc d.+1 p)
    | Recv p a =>
      MP.recv (fun l =>
                 (fix run_alt a :=
                    match a with
                    | A_sing T l' k =>
                      if l == l'
                      then MP.bind (MP.recv_one (coq_ty T) p)
                                   (fun x => extract_proc d (k x))
                      else MP.pure tt
                    | A_cons T l' k a =>
                      if l == l'
                      then MP.bind (MP.recv_one (coq_ty T) p)
                                   (fun x => extract_proc d (k x))
                      else run_alt a
                    end) a)
    | Send p T l v k =>
      MP.bind (MP.send p l v) (fun=>extract_proc d k)
    end.
End ProcExtraction.
\end{lstlisting}

The translation is defined recursively on the structure of processes,
and it constructs a sequence of monadic actions using bind
connecting each action to its continuation.

\subsection{Constructing a Recursive Ping-pong Client} \label{app:pingpong}
We present now several examples implementing the clients of a ping-pong server.
The global protocol that describes the behaviour of all these participants is:

\vskip.1cm
\noindent
$
\begin{array}{l}
\coqDef \; \Rpingpong := \grec{\gX}{}
  \msg{\Alice}{\Bob}\ \{\\\quad \ell_1(\tunit). \; \gend; \;\ \ell_2(\tnat) .
  \msg{\Bob}{\Alice} \ell_3(\tnat) . \gX\}.
\end{array}
$
\vskip.1cm

\noindent
Here, $\Alice$ acts as the client for $\Bob$, which is the ping-pong server.
$\Alice$ can send zero or more \emph{ping} messages (label $\ell_2$), and
finally quitting (label $\ell_1$). $\Bob$, for each ping received, will reply a
\emph{pong} message (label $\ell_3$).

Just as in the $\Rpipe$ example, we project $\Rpingpong$, and get the local type
for $\Alice$: $\AliceLT$. We define several different implementations of
$\AliceLT$ adhering to the protocol specification. The first client,
$\AliceProc_0$ simply quits without sending any ping. To be able to typecheck it
against $\AliceLT$, we need to specify the missing labels in the process
specification:

\vskip.1cm
\noindent
$
\begin{array}{@{}l@{}}
  \coqDef \; \AliceProc_0: \zooidTy{\AliceLT}
  :=
  \dproc{\code{loop}} \; {\lX} \\\ %
  \begin{array}[t]{@{}l@{}}
    (\zselect{\Bob}{%
    \begin{array}[t]{@{}l@{}}
      [\zdflt{\ell_1}{\code{tt} : \tunit} \zend
      \\ \mid
      \zskip{\ell_2}{\tnat} \lrcv{\Bob} \ell_3(\tnat) ; \lX])
    \end{array}
    }
  \end{array}
\end{array}
$
\vskip.1cm

\noindent
The $\zselect{\Bob}{}$ construct specifies that the default branch is to send
$\ell_1$, and then finish, and that the unimplemented behaviour is to send
$\ell_2$ and a $\tnat$, and then receiving $\ell_3$ from $\Bob$, and then
jumping to $\dproc{\code{loop}}\; {\lX}$.  Similarly, we define the process that
keeps sending $\ell_2$ to $\Bob$:

\vskip.1cm
\noindent
$
\begin{array}{@{}l@{}}
  \coqDef \; \AliceProc_1: \zooidTy{\AliceLT}
  :=
  \dproc{\code{loop}} \; {\lX} \\\ %
  \begin{array}[t]{@{}l@{}}
    (\zselect{\Bob}{%
    \begin{array}[t]{@{}l@{}}
      [\zskip{\ell_1}{\code{tt}} \lend
      \\ \mid
      \zdflt{\ell_2}{ 0 : \tnat} \\\ \zrecv{\Bob}{\ell_3}{x : \tnat} \zjump{\lX}])
    \end{array}
    }
  \end{array}
\end{array}
$
\vskip.1cm

\noindent Since $\Proc$ and the local types are inductively defined, there will be
sometimes valid processes with a local type that is not \emph{exactly} the
projection of a participant in the global type. In such cases, we require proofs
that the local type of the process is equal up to unravelling to the local type
projected from the global type. For example, $\AliceProc_0$ could be defined
without using $\dproc{\code{loop}}$, by providing the local type that results of
unravelling once $\AliceLT$:

\vskip.1cm
\noindent
$
\begin{array}{@{}l@{}}
  \coqDef \; \AliceProc_3: \azooid
  :=\\\ %
  \begin{array}[t]{@{}l@{}}
    [\code{proc}\\\ \zselect{\Bob}{%
    \begin{array}[t]{@{}l@{}}
      [\zdflt{\ell_1}{\code{tt} : \tunit} \zend
      \\ \mid
      \zskip{\ell_2}{\tnat} \lrcv{\Bob} \ell_3(\tnat) ; \AliceLT]] 
    \end{array}
    }
  \end{array}
\end{array}
$
\vskip.1cm

\noindent
The type $\azooid$ is the dependent pair type: $\{ \lT \mathbin{\&}
\zooidTy{\lT}\}$. The notation $[\code{proc}\; \zooid ]$ is defined as:

$\code{existT} \; (\tfun{\lT}{\zooidTy \; \lT}) \; \_ \; \zooid$.

\noindent The underscore $\_$
is inferred by Coq, since \dslName{} constructs fully determine their local type
from the inputs.
The first projection $\code{projT1} \; \AliceProc_3$ is the inferred local
type. To ensure that $\AliceProc_3$ behaves as prescribed by $\Rpingpong$, we
need to prove that its inferred local type is equal to $\AliceLT$ up to
unravelling. But for this example, it is enough to unfold $\AliceLT$ once, and
compare the result syntactically with $\code{projT1} \;
\AliceProc_3$. Similarly, if we define a process that sends a fixed number of
$n$ pings and then finishes, we would need to prove that its local type is
syntactically equal to the $n$-th unfolding of $\AliceLT$, which can be done
simply by evaluating its comparison.

Suppose now that we wish to implement a client that sends an undefined number of
pings, until the server replies a natural number greater than some $k$. We show
below the \dslName{} specification:

\vskip.1cm
\noindent
$
\begin{array}{@{}l@{}}
  \coqDef \; \AliceProc_4: \azooid
  := [\code{proc} \\
  \begin{array}[t]{@{}l@{}}
    \zselect{\Bob}{
    \begin{array}[t]{@{}l@{}}
      [\zskip{\ell_1}{\tunit} \lend
      \\ \mid
      \zdflt{\ell_2}{0 : \tnat}\\\ \ %
      \dproc{\code{loop}} \; \lX \; (
      \zrecv{\Bob}{\ell_3}{x : \tnat}
      \\ \qquad
      \zselect{\Bob}{%
      \begin{array}[t]{@{}l@{}}
        [ \zcase{x \geq k}{\ell_1}{\code{tt} : \tunit}\\\ \ \zend
        \\ \mid
        \zdflt{\ell_2}{x : \tnat}\\ \ \ \zjump{\lX}
        ])
      ]]
      \end{array}
      }
    \end{array}
    }
  \end{array}
\end{array}
$
\vskip.1cm

\noindent
The local type for $\AliceProc_4$ is not syntactically equal to $\AliceLT$:

\vskip.1cm
\noindent

$
\begin{array}{l}
\AliceLT =
\dlt{\mu \lX}. \; \lsnd{\Bob} \{
   \ell_1(\tunit). \lend; \\ \quad
   \ell_2(\tnat). \lrcv{\Bob} \ell_3(\tnat). \lX\}

  \\[0.5em]

\code{projT1} \; \AliceProc_4 = \lsnd{\Bob} \{\ %
   \ell_1(\tunit). \lend;
   \; \ell_2(\tnat). \dlt{\mu \lX}.\\\quad \lrcv{\Bob} \ell_3(\tnat).
   \lsnd{\Bob} \{
   \ell_1(\tunit). \lend; \;
   \ell_2(\tnat). \lX.
    \} \}

\end{array}
$

\noindent
However, a simple proof by coinduction can show
that both types unravel to the same
local tree.



\end{document}